\documentclass{amsart}

\calclayout

\usepackage[utf8]{inputenc}
\usepackage[greek,english]{babel}

\usepackage{hyperref}
\usepackage{graphicx}
\usepackage{calligra, mathrsfs, bbm, bm}
\usepackage[shortlabels]{enumitem}
\usepackage{amsmath, amsthm, amssymb, amsopn, amsbsy, amsfonts}
\usepackage{thmtools}
\usepackage{slashed}
\usepackage{tikz-cd}
\usepackage{mathtools}
\usepackage{accents}
\allowdisplaybreaks

\DeclareMathAlphabet{\mathcalligra}{T1}{calligra}{m}{n}

\theoremstyle{plain}
\newtheorem{theorem}{Theorem}[section]
\newtheorem{proposition}[theorem]{Proposition}
\newtheorem{lemma}[theorem]{Lemma}
\newtheorem{corollary}[theorem]{Corollary}

\theoremstyle{definition}
\newtheorem{definition}[theorem]{Definition}

\theoremstyle{remark}
\newtheorem*{rk}{Remark}

\renewcommand{\d}{\ensuremath{\mathrm{d}}}
\renewcommand{\div}{\mathop{\mathrm{div}}}

\renewcommand{\c}{\mathrm{c}}

\renewcommand{\restriction}{\mathord{\upharpoonright}}
\renewcommand{\sc}{\mathrm{sc}}

\newcommand{\relmiddle}[1]{\mathrel{}\middle#1\mathrel{}}

\newcommand{\pair}[2]{\left\langle #1 \, , #2 \right\rangle}
\newcommand{\norm}[1]{\left\| #1 \right\|}
\newcommand{\nnorm}[1]{{\left\vert\kern-0.25ex\left\vert\kern-0.25ex\left\vert #1
    \right\vert\kern-0.25ex\right\vert\kern-0.25ex\right\vert}}

\newcommand{\R}{\mathbb{R}}

\newcommand{\N}{\mathbb{N}}

\newcommand{\rmG}{\mathrm{G}}

\newcommand{\rmS}{\mathrm{S}}
\newcommand{\rmT}{\mathrm{T}}

\newcommand{\Sol}{\mathsf{Sol}}

\newcommand{\scrC}{\mathscr{C}}
\newcommand{\scrD}{\mathscr{D}}
\newcommand{\scrE}{\mathscr{E}}
\newcommand{\scrF}{\mathscr{F}}

\newcommand{\scrI}{\mathscr{I}}

\newcommand{\scrM}{\mathscr{M}}
\newcommand{\scrN}{\mathscr{N}}

\newcommand{\scrS}{\mathscr{S}}

\newcommand{\calB}{\mathcal{B}}

\newcommand{\calD}{\mathcal{D}}

\newcommand{\calF}{\mathcal{F}}
\newcommand{\calG}{\mathcal{G}}

\newcommand{\calM}{\mathcal{M}}

\newcommand{\calO}{\mathcal{O}}

\newcommand{\calS}{\mathcal{S}}
\newcommand{\calT}{\mathcal{T}}
\newcommand{\calU}{\mathcal{U}}
\newcommand{\calV}{\mathcal{V}}
\newcommand{\calW}{\mathcal{W}}

\newcommand{\calY}{\mathcal{Y}}

\DeclareMathOperator{\grad}{grad}
\DeclareMathOperator{\Tr}{Tr}
\DeclareMathOperator{\supp}{supp}

\DeclareMathOperator{\sgn}{sgn}

\DeclareMathOperator{\id}{id}

\newcommand{\loc}{\mathrm{loc}}

\newcommand{\del}{\partial}

\newcommand{\tc}{\mathrm{tc}}
\newcommand{\fc}{\mathrm{fc}}
\newcommand{\pc}{\mathrm{pc}}

\newcommand{\Secs}{\Gamma}
\newcommand{\Lie}{\pounds}

\newcommand{\deriv}[2]{\ensuremath{\frac{\d #1}{\d #2}}}

\newcommand{\pderiv}[2]{\ensuremath{\frac{\partial #1}{\partial #2}}}

\title[``Two-sided'' characteristic problems for linear wave equations]{On the global ``two-sided'' characteristic Cauchy problem for linear wave equations on manifolds}
\author{Umberto Lupo}
\email{umberto.lupo@gmail.com}
\dedicatory{To \textgreek{Γεωργία}}
\date{\today}
\address{Formerly at the Albert Einstein Center for Fundamental Physics, Institute for Theoretical Physics, University of Bern, Sidlerstrasse 5, 3012 Bern, Switzerland}

\begin{document}

\begin{abstract} The global characteristic initial value problem for linear wave equations on globally hyperbolic Lorentzian manifolds is examined, for a class of smooth initial value hypersurfaces satisfying favourable global properties.  First it is shown that, if geometrically well-motivated restrictions are placed on the supports of the (smooth) initial datum and of the (smooth) inhomogeneous term, then there exists a continuous global solution which is smooth ``on each side'' of the initial value hypersurface.  A uniqueness result in Sobolev regularity $H^{1/2+\varepsilon}_\loc$ is proved among solutions supported in the union of the causal past and future of the initial value hypersurface, and whose product with the indicator function of the causal future (resp.\ past) of the hypersurface is past compact (resp.\ future compact).  An explicit representation formula for solutions is obtained, which prominently features an invariantly defined, densitised version of the null expansion of the hypersurface.  Finally, applications to quantum field theory on curved spacetimes are briefly discussed.
\end{abstract}

\maketitle

\section{Introduction}

This paper streamlines, extends and supersedes the author's previous analysis \cite[Ch.\ 2]{lupo2015aspects} of the global characteristic Cauchy problem for (scalar) linear wave equations on globally hyperbolic Lorentzian manifolds, when the problem is posed on  hypersurfaces with favourable global properties.  Existence and uniqueness will be proved in the framework of what will be referred to as the \emph{two-sided characteristic Cauchy problem}.  

To explain the meaning of the latter, consider first the usual non-characteristic Cauchy problem for such equations.  If $M$ is a manifold, $\scrC \subset M$ is a hypersurface, and $P$ is a linear second-order differential operator on $M$ whose principal symbol yields a Lorentzian metric $g$ relative to which $\scrC$ is spacelike, then it is well known that, under standard global assumptions on $\scrC$ and on the Lorentzian manifold $\scrM = (M,g)$, the Cauchy problem for $P$ is well-posed on $\scrC$.  In particular, the solution given arbitrary Cauchy data and an arbitrary right-hand side $F$ may be obtained by ``merging'' a forward solution $\phi^+$, defined to the future of $\scrC$, with a backward solution $\phi^-$, defined to the past of $\scrC$.  Owing to the finite speed of propagation property for $P$, $\phi^+$ (resp.\ $\phi^-$) is zero outside the union of the future (resp.\ past) domains of influence of the following two sets: (a) the joint support of the initial data; (b) the support of $F$ multiplied with the indicator function $\bm{1}^+_\scrC$ (resp.\ $\bm{1}^-_\scrC$) of the future (resp.\ past) domain of influence of $\scrC$.\footnote{In keeping with standard Lorentzian geometric terminology, the latter will henceforth be referred to as the causal future (resp.\ past) of $\scrC$.}  It is in the sense just explained that the standard spacelike Cauchy problem is naturally ``two-sided''.

The situation is more intricate if the initial value hypersurface $\scrC$ is characteristic for $P$---i.e.\ null relative to $g$.  For, on the one hand, one cannot hope to freely prescribe the full first-order (i.e.\ transverse) Cauchy datum once the desired zeroth-order datum (the \emph{characteristic datum}) and right-hand side have been specified.  And, on the other hand, for a generic such $\scrC$ neither the existence nor the uniqueness of global or even semi-global solutions are guaranteed for arbitrary right-hand sides and characteristic data.  A well-known framework in which \emph{semi-global} existence and uniqueness are guaranteed, for arbitrary (modulo regularity assumptions) characteristic data and right-hand sides, is the \emph{Goursat problem}.  There, one restricts attention to a particular subclass of initial value hypersurfaces $\scrC$, the typical example being the one in which $\scrC$ is the boundary of the causal future (i.e.\ the ``light cone'') of a point.  Then, existence and uniqueness are guaranteed in the entire causal future of this point (i.e.\ ``inside the light cone'').  A generalisation of this classical Goursat problem was studied by H{\"o}rmander \cite{hormander1990remark}, whose results were recently extended, as well as reformulated in a global geometric language, by B{\"a}r and Wafo \cite{baer2015initial}.  What these works show is the semi-global existence and uniqueness of solutions of characteristic Cauchy problems posed on hypersurfaces which are achronal and whose causal future is ``compact in the past''.  Semi-global here means to the future of the hypersurface; that is, to ``one side'' of it.

\subsection{Main results}\label{mainres}

In this work, a new framework is described in which to investigate the \emph{global} existence and uniqueness of solutions to characteristic Cauchy problems.  The framework will be defined chiefly by: (a) a choice of a class of (smooth) characteristic hypersurfaces; (b) restrictions on the supports of the allowed right-hand sides and (perhaps more importantly) of the allowed characteristic data; (c) a new natural notion of well-posedness,\footnote{Strictly speaking, issues of continuous dependence of solutions will not be investigated in great detail here and the only result to this effect which shall be presented is Corollary \ref{corcontdependence}.} which might be referred to as ``causal well-posedness'' and which is in effect a generalisation of the usual notion valid for spacelike Cauchy surfaces.  Then, the main set of results in this paper (Theorems \ref{existencethmnodistrsol}, \ref{distrsolnthm} and \ref{uniquenessconcrete}) may be summarised by the statement that the global characteristic Cauchy problem on globally hyperbolic Lorentzian manifolds and for hypersurfaces in the class referred to in (a) is causally well-posed whenever the characteristic data and the right-hand sides satisfy the support conditions referred to in (b).

Points (a), (b) and (c) above will now be individually explained.  The conditions singling out the class of initial value hypersurfaces referred to in (a) are, essentially, \emph{achronality} together with a certain completeness condition on the null geodesic generators; it is important to note that the latter condition is \emph{never} satisfied in the scenarios considered in \cite{baer2015green}.  With $\bm{1}^{+}_\scrC$ (resp.\ $\bm{1}^-_\scrC$) denoting, as before, the indicator function of the causal future (resp.\ past) of $\scrC$, the support restrictions referred to in (b) can be broken down as follows: (b1) the support of the characteristic datum will be ``compact in time'' (i.e.\ both compact along future directions and along past directions) when regarded as a subset of the ambient Lorentzian manifold $\scrM$; (b2) the support of the right-hand side $F$ will be contained in the union of the causal future and causal past of $\scrC$ in $\scrM$; (b3) $\supp{(\bm{1}^+_\scrC F)}$ will be ``compact in the past'' (i.e.\ compact along past directions), and $\supp{(\bm{1}^-_\scrC F)}$ will be ``compact in the future'' (i.e.\ compact along future directions).  Finally, and modulo precise statements about regularity for the time being, ``causal well-posedness'' is the statement that solutions exist and are unique \emph{among} those distributional solutions which themselves satisfy support conditions analogous to (b2) and (b3).

Notice that if $\scrC$ were a smooth \emph{spacelike Cauchy surface}, then the support conditions (b1), (b2) and (b3) would be trivially satisfied by any initial data (which, in this case, would also include some independent normal derivative data) and right-hand side.  Similarly, in this case the notion of causal well-posedness just introduced is identical to the usual notion of well-posedness.  On the other hand, if $\scrC$ is not a spacelike Cauchy surface, conditions (b1), (b2) and (b3) will generally select proper subspaces of (choosing spaces of smooth functions for definiteness) $C^\infty(\scrC)$ and $C^\infty(M)$.  An example is shown in Figure \ref{setupfigure}, and note that there and henceforth the characteristic initial value hypersurface of interest will be denoted by $\scrN$ instead of $\scrC$ in order to highlight its null character. 

\begin{figure}\label{setupfigure}
\includegraphics[scale=0.58]{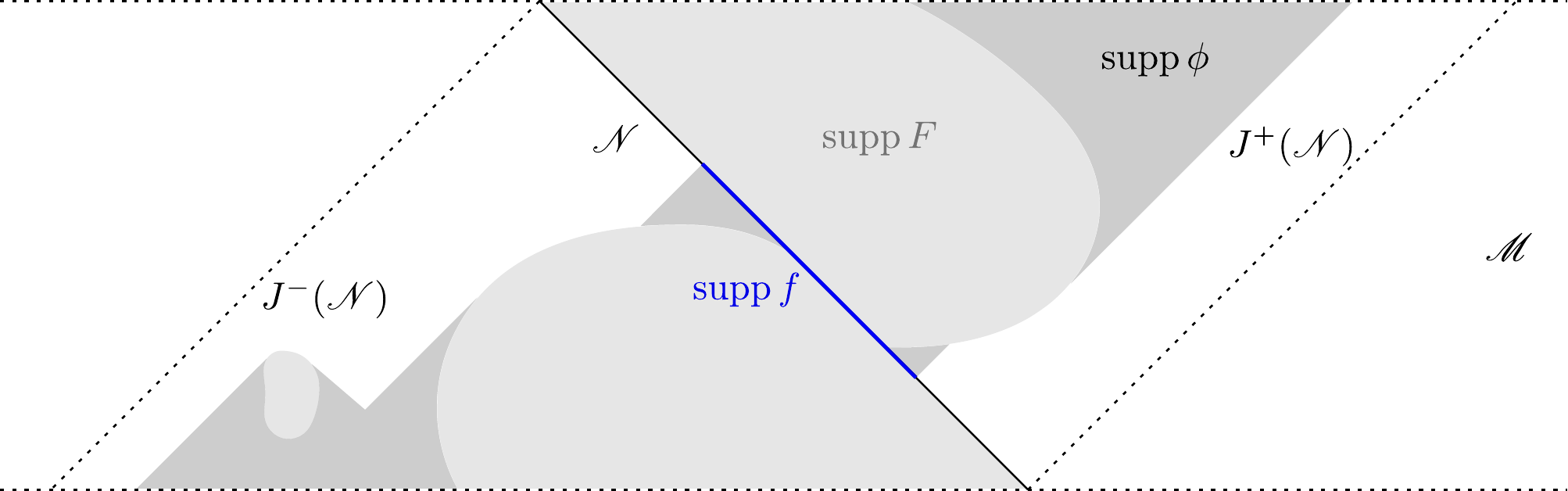}
\begin{caption}{\small{Setup for the two-sided characteristic Cauchy problem considered in this paper.  In this case, $\scrM$ is a ``temporal slab'' in (1+1)-dimensional Minkowski spacetime.  $\scrN$ is the characteristic initial value surface.  The thick blue portion of $\scrN$ and the light grey areas in $\scrM$ represent the supports of a smooth initial datum $f$ and of a smooth inhomogeneous term $F$ (respectively) to which the results in this paper apply.  The union of the light grey and dark grey areas contains the support of the resulting solution $\phi$, uniquely determined by the support properties described in the text.  Notice that $\scrN$ is not a Cauchy surface for $\scrM$ and $J^+(\scrN) \cup J^-(\scrN)$ is a proper subset of $M$.}}\end{caption}
\end{figure}

As in the standard spacelike Cauchy problem, here too the global solution will be obtained by merging a ``forward solution'' defined in the causal future of $\scrN$, with a ``backward solution'' defined in the causal past of $\scrN$.    This justifies the expression ``two-sided characteristic Cauchy problem'' to indicate the problem formulated and solved in this paper.  And in effect, the forward (resp.\ backward) solution $\phi^+$ (resp.\ $\phi^-$) may be seen as arising from semi-globally solving an appropriate forward (resp.\ backward) Goursat problem in the following way---illustrated, for simplicity, in the case of a vanishing right-hand side $F$: for a given characteristic datum $f$, let $\scrN^+$ (resp.\ $\scrN^-$) be \emph{any} achronal characteristic hypersurface which is compact in the past (resp.\ future) and which includes all points in the intersection between $\scrN$ and the causal future (resp.\ past) of $\supp{f}$.\footnote{Equivalently, these are the points of $\scrN$ which may be reached from $\supp{f}$ by travelling in the future (resp.\ past) along the null generators of $\scrN$.}  Then $\phi^\pm$ will be the extension by zero to the rest of $J^{\pm}(\scrN)$ of the unique solution to the forward/backward Goursat problem on $\scrN^\pm$ with characteristic datum equal to $f$ on $\scrN^\pm \cap \scrN$, and zero elsewhere on $\scrN^\pm$.

At the analytical level, the existence statements in this paper (Theorems \ref{existencethmnodistrsol} and \ref{distrsolnthm}) will be restricted for simplicity to smooth characteristic data and smooth right-hand sides.  The construction will show that both the forward and the backward solution can be extended to global smooth functions, making the ``merge'' (continuous and) smooth ``on each side'' of the characteristic hypersurface separately.  However, discontinuities in transverse derivatives across the hypersurface will arise in general. Consequently, uniqueness will have to be established in a setting of sufficiently low regularity.  Indeed, Theorem \ref{uniquenessconcrete} will be a uniqueness result valid among all distributions in the kernel of $P$ with the aforementioned support properties, with Sobolev regularity $H^{1/2+\varepsilon}_\loc$ for some $\varepsilon > 0$, and with vanishing trace on the initial value hypersurface.

The arguments for existence given here and in \cite[Ch.\ 2]{lupo2015aspects} have in common the heavy reliance on a circle of ideas used in an influential 1990 paper by Rendall \cite{rendall1990reduction} (see also \cite{chrusciel2012many} for a review).  Rendall was able to provide a local existence result applicable to characteristic Cauchy problems for \emph{quasi-}linear hyperbolic PDEs posed on a portion of the union of two transversely intersecting null hypersurfaces $\scrN_1$ and $\scrN_2$, namely the portion which lies to the causal future of the submanifold $\scrN_1 \cap \scrN_2$ (with similar statements of course with ``future'' replaced by ``past'').  In effect, the \emph{geometric} setup of Rendall's paper may be regarded as a special case of the one in the (generalised) Goursat problem treated in \cite{hormander1990remark,baer2015initial}.  Furthermore, both Rendall's approach and the one in \cite{baer2015initial} are based on a reduction to the inhomogeneous spacelike Cauchy problem with vanishing data in the far past.   However, in the case of \emph{smooth} data, Rendall's method immediately yields smooth (in the one-sided sense) solutions with no difficulty, whereas to establish the same using the energy-based methods in \cite{hormander1990remark,baer2015initial} would seemingly require more work.  Some quick comments on this issue and on the relation of the present work with other literature on characteristic problems will be given in Section \ref{exlit}.

The main technical tool in \cite{rendall1990reduction} was the celebrated extension theorem of Whitney's \cite{whitney1934analytic}.  The latter allows to construct, in the first step, a function solving the characteristic Cauchy problem to infinite order on the initial value hypersurface; then, the error given by the failure of this ``approximate solution'' to be a true solution away from the initial value hypersurface can be subtracted away by means of a reduction to the standard spacelike Cauchy problem.  In fact, as brought to the author's attention by A.\ D.\ Rendall after the work in \cite{lupo2015aspects} was completed, a considerably streamlined version of the arguments in \cite{rendall1990reduction} had later been presented in \cite{rendall1992stability}, in which it was shown that \emph{Borel's lemma} could equally well be used to construct the approximate solution.\footnote{And notice that perhaps a precursor to Rendall's method in the linear case may be found in the local existence proof for the classical Goursat problem sketched in \cite[Thm.\ 5.4.2]{friedlander1975wave}.}  Correspondingly, a global version of Borel's lemma is shown here to be sufficient to obtain the main existence result, Theorem \ref{existencethmnodistrsol}.

Typical arguments for uniqueness in characteristic Cauchy problems rely on adjoint formulae for $P$ which are simple to establish when the arguments are smooth functions.  Such formulae can be recast as identities---referred to in this paper as \emph{jump formulae}---showing that the action on a smooth function of the commutator of $P$ with multiplication by the indicator function of a regular domain $D$ yields a certain distribution concentrated on the topological boundary of $D$.  If such a commutator identity can be extended to less regular arguments, it may be used to prove low-regularity uniqueness for characteristic Cauchy problems.  This is the approach pursued in \cite{baer2015initial}.\footnote{Jump formulae for characteristic domains in hyperbolic problems have also recently been studied in \cite{lerner2017unique}, in the context of unique continuation for ill-posed characteristic problems.}  Here, it will be argued that it is sufficient to simply extend to rough solutions of $P$ the special case of this identity which is valid on smooth solutions of $P$.  In this way it will be possible to improve on what could be directly obtained using the methods in \cite{baer2015initial}.

Aside from being used to prove uniqueness, the jump formulae mentioned in the previous paragraph will eventually lead to a representation formula [Equation (\ref{repformulafinal})] neatly expressing the unique solution of a two-sided characteristic Cauchy problem in terms of the application of the retarded-minus-advanced Green operator to a sum of ``single-layer'' distributions (see Appendix \ref{gencommidentityapp}) supported on $\scrN$.  One of these single layers is defined in terms a certain density on $\scrN$, introduced here and called the \emph{expansion density} of $\scrN$ by analogy with the well-known (but \emph{not} invariantly defined) notion of null expansion in Lorentzian geometry.  In the case in which the tangential component of the vector field representing possible first-order terms in $P$ vanishes, the argument of the retarded-minus-advanced Green operator in the representation formula enjoys a remarkably simple scaling behaviour when the ambient Lorentzian metric is multiplied by a conformal factor which is constant on null generators (Theorem \ref{equivariancethm}).  Finally if, in addition, the expansion density of the hypersurface is identically zero, the representation formula becomes particularly simple.  Possible implications for quantum field theory on curved spacetimes will be discussed in Section \ref{expdensityhadamard}.

As a final note, it is worth mentioning that the careful treatment of the aforementioned lack of smoothness for the solutions to two-sided characteristic Cauchy problems has been a serious (and interesting) technical hurdle in certain literature on the rigorous quantisation of linear field theories on curved spacetimes, and in particular on spacetimes possessing null (event or Killing) horizons.  Indeed, existence, ``two-sided smoothness'', and uniqueness of solutions ``falling/propagating entirely through'' characteristic hypersurfaces have been assumed and used in the literature on the Hawking or Unruh effect; see, e.g., \cite[Ch.\ 7]{wald1994quantum} and references therein.  In particular, a precise understanding of what global degree of regularity can be generically expected from two-sided characteristic Cauchy problems was of crucial importance in the seminal work by Kay and Wald \cite{kay1991theorems} on the interplay between bifurcate Killing horizon structures and the so-called Hadamard condition on quantum states for linear Klein--Gordon fields.  This issue will be revisited in Section \ref{SASBareOK}.

\subsection{Plan of the paper}

The paper is organised as follows: Section \ref{generalitieswaveeqns} summarises known background on second-order hyperbolic partial differential equations on manifolds; Section \ref{waveeqnsdiffconsequencessec} is a self-contained exposition of (essentially known) facts on the \emph{propagation equations} obtained by evaluating linear wave equations and their differential consequences on characteristic hypersurfaces; the construction of global solutions to two-sided characteristic Cauchy problem is carried out in Section \ref{existencesec}; the regularity of the solutions is briefly examined in Section \ref{regularitythm}; two fundamental jump formulae are derived in Section \ref{commidentitiessec} (Theorem \ref{thmcommidentity} and Corollary \ref{corsecondcommutator}), the second of which is used in Section \ref{uniquenesslindipsec} to establish uniqueness; in Section \ref{repformulaeunivsec}, a first representation formula for solutions is derived (Proposition \ref{corprotrepformula}) which immediately yields a statement on continuous dependence (Corollary \ref{corcontdependence}) and which, upon using the other jump formula in Section \ref{commidentitiessec}, yields in Section \ref{repformulasubsec} the final representation formula in the homogeneous case; in Section \ref{expdensitysubsec}, the expansion density of a null hypersurface in a Lorentzian manifold is shown to be invariantly defined and its behaviour under certain changes in the ambient metric is studied, leading in Section \ref{universalitysubsec} to some remarks on universality in the two-sided characteristic Cauchy problem; the relation with the existing literature, together with some applications to quantum field theory on curved spacetimes, is discussed in Section \ref{applicationssec}.  There are four appendices: Appendix \ref{causcomplapp} covers material on Lorentzian causality theory needed in the main body of text; Appendix \ref{nullhypappx} collects properties of null hypersurfaces in Lorentzian manifolds; Appendix \ref{borelslemmaapp} contains a self-contained statement and proof of a global version of the Borel Lemma; a version of the divergence theorem applicable to densities on manifolds is recalled in Appendix \ref{divthmappendix}, together with its consequences on formal adjoints of differential operators.

\subsection{Notation and terminology}\label{notationsec}

Throughout, a ``manifold'' will always be a \emph{smooth}, second-countable, Hausdorff manifold.  The expression ``smooth hypersurface'' will be used to indicate a smoothly \emph{embedded} codimension-$1$ submanifold (without boundary).  A ``topological hypersurface'' will be an embedded topological submanifold of codimension $1$.  The topological boundary, closure and interior of a set $A$ in a topological space $X$ will be denoted by $\del{A}$, $\overline{A}$ and $A^{\circ}$ respectively.  Given a smoothly embedded submanifold $S$ of $M$, $N^*S \subseteq T^*M$ will denote the (total space of the) conormal bundle of $S$.

If $M$ is a manifold, $\mathfrak{X}(M)$ will denote the set of global $C^\infty$ vector fields on $M$.  The density bundle of $M$ will be denoted by $\calD M$; a (rough, $C^k$) global section of this bundle will be called a (rough, $C^k$) density on $M$.  For any $k \in \N_0 \cup \{ \infty\}$, $\Gamma^k_{(\c)}(E)$ will denote the vector space of $C^k$-regular (and compactly supported) sections of $E$.  The continuous dual of the LF-space $\Gamma_{\c}^\infty(\calD M \otimes E^*)$ will be called the space of \emph{distributional} sections of $E$ and denoted by $\Gamma^{-\infty}(E)$.  $\Gamma^{-\infty}(M \times \R)$ will be denoted by $\scrD'(M)$, and its elements simply \emph{distributions} on $M$.  $\scrE'(M)$ will denote the subspace of $\scrD'(M)$ consisting of compactly supported distributions.  When convenient, the notation $\pair{\phi}{\nu}$ will be used to indicate the evaluation of $\phi \in \scrD'(M)$ on an element $\nu \in \Secs^\infty(\calD M)$ such that $\supp{\nu} \cap \supp{\phi}$ is compact.  Square brackets will be used around scalar-valued functions to define multiplication operators on sections: namely, for any section $\phi$ and scalar function $v$, $[v]\phi$ is the section $x \mapsto v(x) \cdot \phi(x)$.

The signature convention $(+,-,\ldots,-)$ for Lorentzian metrics is adopted here.  If $(M,g)$ is a Lorentzian manifold, $\nabla$ will denote the Levi-Civita connection arising from $g$; ${}^\sharp$ and ${}^\flat$ will indicate index raising and lowering (respectively) using $g$.  The \emph{gradient}, \emph{divergence} and \emph{d'Alembert} operators associated with $g$ are (respectively) 
\begin{align*} &\grad : C^\infty(M) \ni \phi \mapsto [\d \phi]^\sharp \in \mathfrak{X}(M),\\
	&\div : \mathfrak{X}(M) \ni X \mapsto \Tr[\nabla X] = \nabla_a X^a \in C^\infty(M), \\
	&\square_g : C^\infty(M) \ni \phi \mapsto \div{\grad{\phi}} = \nabla_a \nabla^a \phi \in C^\infty(M).
\end{align*}
A Lorentzian manifold $(M,g)$ is said to be \emph{time orientable} if there exists a global, smooth, timelike vector field.  In this case, two global and continuous notions of ``future-directedness'' of tangent vector are possible on $(M,g)$, and for one such choice denoted by $\mathfrak{t}$ the triple $(M,g,\mathfrak{t})$ will be called a \emph{time-oriented} Lorentzian manifold.  As usual (see, e.g., \cite{oneill1983semi-riemannian}), for a subset $A$ of a time-oriented Lorentzian manifold, the chronological future/past of $A$ will be denoted by $I^{+/-}(A)$, the causal future/past of $A$ by $J^{+/-}(A)$, and the causal shadow of $A$ by $J(A) \doteq J^+(A) \cup J^-(A)$.  A time-oriented Lorentzian manifold $(M,g,\mathfrak{t})$ is said to be \emph{globally hyperbolic} if it does not allow for closed causal curves and if the set $J^+(p) \cap J^-(q)$ are compact for any two $p,q \in M$ \cite{bernal2007globally}.  Further results and notation are provided in Appendix \ref{causcomplapp} and in Appendix \ref{nullhypappx}.

A final note concerns the use of the letter $n$ throughout the text: while it will stand for a \emph{vector} field in Section \ref{waveeqnsdiffconsequencessec}, in most of Section \ref{existencesec} and in Appendix \ref{nullhypappx}, it will stand for a \emph{covector} field in the proof of Theorem \ref{distrsolnthm}, in Sections \ref{commidentitiessec}, \ref{uniquenesslindipsec}, \ref{repformulaeunivsec} and \ref{expdensunivsec}, and in Appendix \ref{divthmappendix}.

\section{Generalities on normally hyperbolic differential operators}\label{generalitieswaveeqns}

Whenever a manifold is equipped with a Lorentzian metric tensor, there arises a distinguished class of second-order differential operators. 

\begin{definition}[Normally hyperbolic operators \cite{baer2007wave}]\label{normhypdefn}
	Let $(M,g)$ be a Lorentzian manifold and $E \to M$ be a vector bundle.  A second-order linear differential operator $P : \Gamma^\infty(E) \to \Gamma^\infty(E)$ is said to be \emph{normally hyperbolic} if
	\begin{equation*}
	\sigma_P(\xi) = g^{-1}(\xi,\xi) \, \id_E \quad \text{for all } \xi \in T^*M,
	\end{equation*}
	where $\sigma_P :T^*M \to \mathrm{End}(E)$ denotes the principal symbol of $P$ \cite[Ch.\ 10]{nicolaescu2007lectures}, and $g^{-1} \in \Secs^\infty(TM \otimes TM)$ is the inverse metric tensor.
\end{definition}

As effectively first noticed by Leray \cite{leray1953hyperbolic} (for modern treatments, see \cite{baer2007wave, ginoux2009linear, ringstrom2009cauchy, baer2017lectures}), whenever $(M,g)$ is (time orientable and) globally hyperbolic such operators admit a well-posed initial value formulation on (smooth) spacelike Cauchy surfaces.  To wit: If $P$ is as in Definition \ref{normhypdefn} then, given (i) a smooth spacelike Cauchy surface $\scrC$, (ii) a smooth vector field $\nu$ along $\scrC$ and transverse to $\scrC$, (iii) a pair $(\varphi, \pi)$ of smooth (resp.\ distributional) sections of $E|_\mathscr{C}$ (the \emph{Cauchy data}), and (iv) a smooth (resp.\ distributional) section $F$ of $E$ (the \emph{inhomogeneity}), the system
\begin{equation}\label{Puequalsf} \begin{cases} P\phi = F \\ \phi\restriction_{\scrC} = \varphi \\ [\nabla_\nu \phi]\restriction_{\scrC} = \pi \end{cases} \end{equation}
admits a unique smooth (resp.\ distributional) and global solution.  Furthermore, if certain suitably topologized spaces of Cauchy data and of inhomogeneities are chosen together with a suitably topologized codomain for the solution operator, it is possible to prove that the latter is continuous.  This state of affairs is referred to as the \emph{well-posedness} of the Cauchy problem.

Although normally hyperbolic operators form a wide class, attention in this paper will be restricted for simplicity to the case of normally hyperbolic operators acting on smooth real-valued functions, i.e.\ on sections of the trivial line bundle $E = M \times \R$.  The most general such operator takes the form
\begin{equation}\label{normallyhyperbolicscalar} P = \square_g + X + [q],
\end{equation}
where $\square_g$ is the d'Alembert operator associated with $g$, $X$ is a smooth vector field, and $q$ is a smooth (real-valued) function.

\subsection{Formal adjoints}

The following definition is standard and is recalled here to set notation.

\begin{definition}\label{formaladjointdefn}
	Let $M$ be a manifold and $\mu$ be a nowhere-vanishing smooth volume density.  Let $E_1, \, E_2$ be vector bundles over $M$ and over a common field $\mathbb{K}$, and $E_1^*, \, E_2^*$ be the respective dual bundles.  Let $P : \Gamma^\infty(E_1) \to \Gamma^\infty(E_2)$ be a differential operator.  The \emph{formal adjoint} $P^\dagger$ of $P$ with respect to $\mu$ is the differential operator $P^\dagger : \Gamma^\infty(E_2^*) \to \Gamma^\infty(E_1^*)$ uniquely defined by
	\begin{equation}\label{geneqnadjoint} \int_M \beta (P\alpha) \, \mu = \int_M [P^\dagger\beta] (\alpha) \, \mu \end{equation}
	for all $\beta \in \Gamma^\infty(E_2^*)$ and $\alpha \in \Gamma^\infty(E_1)$ such that $\supp{\beta} \cap \supp{\alpha}$ is compact.
\end{definition}
\begin{rk} Distributionally, Equation (\ref{geneqnadjoint}) reads simply as $[P\alpha](\beta \mu) = \alpha( [P^\dagger \beta]\mu)$, and this formulation allows to extend $P$ by (weak-*) continuity to distributional sections of $E_1$. The resulting extension will still be denoted by $P$ and elements in its kernel will be referred to as \emph{distributional solutions} of $P$.
\end{rk}

If $E_1=E_2=E$ and there is a preferred vector bundle isomorphism $E \to E^*$ covering the identity, then $P^\dagger$ can be regarded as an operator $\Gamma^\infty(E) \to \Gamma^\infty(E)$, and $P$ is said to be \emph{formally self-adjoint} relative to this identification if $P^\dagger=P$.

Infinitesimally, at least whenever $M$ is orientable, one has the following equivalent characterisation of formal adjoints, which is a version of the \emph{Green--Vinogradov} formula---see \cite{alonso-blanco2002green, alonso-blanco2004green} and references therein.  $P^\dagger$ is the unique differential operator for which there exists a vector-field--valued bilinear and bidifferential operator $\calG : \Gamma^\infty(E_2^*) \times \Gamma^\infty(E_1) \to \mathfrak{X}(M)$, called here a \emph{Green vector field} for $P$, such that
\begin{equation}\label{greenvinogradov} \beta (P\alpha) -  [P^\dagger \beta] (\alpha) = {\div}_\mu \calG[\beta,\alpha], \end{equation}
where ${\div}_\mu$---the divergence operator on vector fields relative to a fixed density---is defined by Equation (\ref{defnofdivden}) in Appendix \ref{divthmappendix}.\footnote{Indeed, if Equation (\ref{greenvinogradov}) holds then Equation (\ref{geneqnadjoint}) follows upon applying the divergence theorem for densities (see Appendix \ref{divthmappendix}).}
 
Henceforth, whenever a background Lorentzian metric $g$ is introduced, all formal adjoints will be automatically understood to be so relative to the induced volume density $\mu_g$.  If $P$ is of the form given in Equation (\ref{normallyhyperbolicscalar}), its formal adjoint is  
	\begin{equation}\label{adjointnormhypeq} P^\dagger = \square_g - X + [q - \div X] = P - 2X - [\div X].
	\end{equation}
Indeed, this follows from the formal self-adjointness of $\square_g$ together with Corollary \ref{coradjointvecfield}.  Since $2X + [\div X] = 0$ if and only if $X = 0$, $P$ is formally self-adjoint if and only if $X=0$.  Equation (\ref{adjointnormhypeq}) then also allows to find a Green vector field for $P$, namely
\begin{equation}\label{greenformscalarnormhyp}
	j[\chi, \phi] \doteq \chi \,\grad{\phi} - \phi \, \grad{\chi} + \chi \phi X, \quad \chi, \phi \in C^\infty(M).
\end{equation}
For an analogous general result the reader may consult Lem.\ 3.2.1 in \cite{baer2007wave}.  Note that $j$ is antisymmetric if and only if $X=0$.

In Appendix \ref{gencommidentityapp}, some formulae are presented for the commutator between a generic partial differential operator and the operator of multiplication by the indicator function of a suitable domain.  When specialised to the case of (scalar) normally hyperbolic operators, such formulae will play a crucial role at several points in this paper.  Accordingly, let $(M,g)$ be a Lorentzian manifold, $\mu = \mu_g$ be the metric volume density, and $P$ be of the form given in Equation (\ref{normallyhyperbolicscalar}).  The fundamental identity is Equation (\ref{commidentity1}) in Corollary \ref{corcommidentitygen} when applied to this case.  Namely:
\begin{equation}\label{commidentity2}
\pair{[[\bm{1}_D], P]\phi}{\chi \mu_g} = \int_D \div j[\chi,\phi] \, \mu_g =  \int_{\del{D}} \big[\chi n^\sharp \phi - \phi n^\sharp \chi + \phi \chi \, n(X)\big] \, \iota_n \mu_g
\end{equation}
whenever $D$ is a domain with locally Lipschitz boundary in $N$, $n$ is a (possibly rough and/or almost everywhere defined) field of outward-pointing conormals to $\del D$, and $\phi, \chi \in C^\infty(M)$ are such that $\supp{\chi} \cap \supp{\phi} \cap \overline{D}$ is compact.

\subsection{Green hyperbolicity of normally hyperbolic operators}\label{greenhypsubsec}

As well as providing well-posedness for the spacelike Cauchy problem, normally hyperbolic operators on globally hyperbolic manifolds belong to the (wider) class of \emph{Green hyperbolic} operators studied in detail in \cite{baer2012classical,khavkine2014covariant,baer2015green}. Namely, it holds for any globally hyperbolic Lorentzian manifold $\scrM$ and normally hyperbolic operator $P : \Gamma^\infty(E) \to \Gamma^\infty(E)$ that there exist a \emph{retarded} ($+$) and an \emph{advanced} ($-$) \emph{Green operator} $\widehat{\rmG_\pm}$ for $P$; by definition, $\widehat{\rmG_\pm}$ is a linear map $\Gamma_{\c}^\infty(E) \to \Gamma^\infty(E)$ such that
	\begin{enumerate}[label=\textbf{(\roman*)}]
		\item $P \circ \widehat{\rmG_\pm} = \widehat{\rmG_\pm} \circ P\restriction_{\Gamma_{\c}^\infty(E)} = \id_{\Gamma_{\c}^\infty(E)}$;
		\item $\supp{\widehat{\rmG_\pm} \psi} \subseteq J^\pm(\supp{\psi})$ for all $\psi \in \Gamma_{\c}^\infty(E)$.
	\end{enumerate}
It can be shown that the retarded and advanced Green operators are uniquely determined by the requirements \textbf{(i)} and \textbf{(ii)}.  Note also that since $P^\dagger$ is also normally hyperbolic, it possesses operators with analogous properties.  $\widehat{\rmG_\pm}$ may be used to solve a particular Cauchy problem in which, given an inhomogeneity $F \in \Gamma_{\c}^\infty(E)$, a solution $\phi$ to $P\phi = F$ is sought which vanishes on any spacelike Cauchy surface $\mathscr{C}^\mp$ such that $\supp{F} \subset J^\pm(\mathscr{C}^\mp)$---equivalently, which vanishes outside $J^\pm(\supp{F})$.  Namely, $\phi_\pm \doteq \widehat{\rmG_\pm} F$ does the job.  Actually, similar statements may be made for larger classes of inhomogeneities as will now be recalled.  The reader is referred to Appendix \ref{causcomplapp} for notation on distinguished spaces of sections, and to \cite{baer2015green} for proofs.

\begin{theorem}\label{baerthm}
	If $P$ is a normally hyperbolic operator acting on sections of a vector bundle $E$ over a globally hyperbolic Lorentzian manifold, then it is bijective as a map $\Secs^{-\infty}_\pc(E) \to \Secs^{-\infty}_\pc(E)$ and as a map $\Secs^{-\infty}_\fc(E) \to \Secs^{-\infty}_\fc(E)$.  The resulting inverses are extensions 
	\begin{equation*} \rmG_+ : \Secs^{-\infty}_\pc(E) \to \Gamma^{-\infty}_\pc(E) \quad \text{and} \quad \rmG_- : \Gamma^{-\infty}_\fc(E) \to \Gamma^{-\infty}_\fc(E) \end{equation*}
	of $\widehat{\rmG_+}$ and $\widehat{\rmG_-}$ (respectively) which enjoy support properties analogous to \textbf{\emph{(ii)}} above.  Furthermore, $\rmG_+[\Secs^\infty_\pc(E)] = \Secs^\infty_\pc(E)$ and $\rmG_-[\Secs^\infty_\fc(E)] = \Secs^\infty_\fc(E)$. \qed
\end{theorem}

If $P$ and the underlying Lorentzian manifold are as in Theorem \ref{baerthm}, one may use the maximally extended advanced and retarded Green operators for $P$ to define the (maximally extended) \emph{causal Green operator} for $P$.  This is
\begin{equation*}
\rmG \doteq \rmG_+ - \rmG_- : \Gamma_{\tc}^{-\infty}(E) \to \Gamma^{-\infty}(E),
\end{equation*}
extending $\widehat{\rmG} \doteq \widehat{\rmG_+} - \widehat{\rmG_-} : \Gamma_\c^{-\infty}(E) \to \Gamma^\infty(E)$, and it holds that $\supp{\rmG \tau} \subseteq J(\supp{\tau})$ for any $\tau \in \Gamma^{-\infty}_\tc(E)$.   It is also evident that $\rmG$ maps to distributional solutions.  In fact, $\rmG$ is surjective on the space of all distributional solutions of $P$ (the argument for this being essentially the same as the one given in Footnote \ref{footnoteinverseop}).

\section{Existence and uniqueness in the two-sided characteristic Cauchy problem for linear wave equations}\label{CIVPrendall}

Let $P$ be a normally hyperbolic differential operator acting on sections of a vector bundle $E$ over a Lorentzian manifold $(M,g)$.  By the very definition of normal  hyperbolicity, the characteristic set of $P$ consists of all non-zero covectors $\xi$ which are null for $g$ in the sense that $g^{-1}(\xi, \xi)=0$---equivalently, such that $\xi^\sharp$ is a ($g$-)null vector.  Hence, the characteristic Cauchy problem for $P$ is
\begin{equation*} \begin{cases} P\phi=F \\ \phi\restriction_{\scrN} = f, \end{cases} \end{equation*}
where $\scrN \subset M$ is a null hypersurface, and $F$ (resp.\ $f$) is a prescribed sections of $E$ (resp.\ of $E|_\scrN$).  In this paper, $E$ will be taken to be the trivial line bundle.

This section contains the main existence and uniqueness results in this paper.  As already discussed in the introduction, the existence result uses ideas presented in influential work of Rendall's \cite{rendall1990reduction,rendall1992stability}.  Rendall's arguments were local, but valid for a class of nonlinear hyperbolic equations.  Here, a geometric framework will be provided in which Rendall's arguments can be globalised in the linear case.

\subsection{Scalar wave equations and their differential consequences on characteristic hypersurfaces}\label{waveeqnsdiffconsequencessec}

It is known for generic partial differential operators of any order $k$ that, in local coordinates adapted to a characteristic hypersurface, the equation $P\phi=F$ cannot be rewritten in such a way as to express a $k$th order transverse derivative of $\phi$ at $\scrN$ in terms of the values at $\scrN$ of $\phi$ and of its other derivatives up to order $k$.  Instead, when evaluated on $\scrN$ the equation $P\phi=F$ reduces to compatibility conditions between $\phi\restriction_\scrN$, its derivatives along $\scrN$, and derivatives along $\scrN$ of the transverse derivatives of $\phi$ up to order $k-1$.  This is completely unlike the situation in the Cauchy problem posed on non-characteristic hypersurfaces.

This situation will now be studied in detail in the case of interest for the present paper, namely that of scalar normally hyperbolic operators on Lorentzian manifolds.  The final result of this subsection, Corollary \ref{characteristicdecompositioncor}, will later allow to view the evaluation on a null (initial value) hypersurface $\scrN$ of any normally hyperbolic equation (resp.\ of all its differential consequences) as a parametrised family of ordinary differential equations (ODEs) for the first (resp.\ higher-order) transverse derivative of the putative solution along the null generators of $\scrN$.

Let $(M,g)$ be a Lorentzian manifold of dimension $d+1$, and let $\{ e_0, \ldots, e_{d} \}$ be an arbitrary local frame for $TM$ around a point $p$.  Further let $\{ \epsilon^0, \ldots, \epsilon^d \}$ be its dual coframe, uniquely defined by the requirement that $\epsilon^\mu(e_\nu) \equiv \delta^\mu_\nu$ for all $\mu,\nu=0,\ldots,d$.  Then
\begin{equation*}
\nabla \grad \phi = \nabla[\epsilon^\mu(\grad \phi) e_\mu] = \nabla[{(\epsilon^\mu}^\sharp \phi) e_\mu] = \nabla[{\epsilon^\mu}^\sharp \phi] \otimes e_\mu + [{\epsilon^\mu}^\sharp \phi] \nabla e_\mu.
\end{equation*}
The following simple representation of the d'Alembert operator follows:
\begin{equation}\label{boxopbasisdualbasis}
\square_g \phi = \Tr{[\nabla \grad \phi]} = (e_\mu + [\div e_\mu])({\epsilon^\mu}^\sharp \phi).
\end{equation}

The interested reader might find it instructive, before considering hypersurfaces, to first examine the form of the d'Alembert operator along a single null \emph{curve}, in terms of an arbitrary local frame adapted to the curve.\footnote{Alternatively, the reader may proceed to the paragraph preceding Equation (\ref{DscrNt}) and refer back only to those calculations which are explicitly needed there.}  Suppose that $p \in M$, that $\gamma$ is a smoothly embedded null curve through $p$ with image $\Gamma \subset M$, and that a smooth frame field $\{n,t,m_1,\ldots,m_{d-1}\}$ with dual coframe field $\{ \nu, \tau, \mu^1, \ldots, \mu^{d-1} \}$ is defined around $p$ with the following properties: (a) $n$ is tangent to $\Gamma$ (and hence null there); (b) on $\Gamma$, $m_i$ ($i=1, \ldots, d-1$) is orthogonal to $n$.  It is clear that:
\begin{itemize}
\item on $\Gamma$, $\tau^\sharp$ is proportional to $n$, i.e.\ $\tau^\sharp = \nu(\tau^\sharp) n + \xi = g(\tau^\sharp, \nu^\sharp) n + \xi$, where $\xi$ is a local vector field vanishing on $\Gamma$;
\item on $\Gamma$, $\nu^\sharp$ does not belong to the null hyperplanes spanned by $n$ and the $m_i$'s, i.e.\ $\nu^\sharp = \nu(\nu^\sharp)n + \tau(\nu^\sharp)t + \mu^i(\nu^\sharp)m_i = g(\nu^\sharp,\nu^\sharp)n + g(\tau^\sharp,\nu^\sharp)t + g({\mu^i}^\sharp,\nu^\sharp)m_i$ with $g(\tau^\sharp,\nu^\sharp) \neq 0$ on $\Gamma$;
\item on $\Gamma$, each ${\mu^i}^\sharp$ belongs to those hyperplanes but is not proportional to $n$.
\end{itemize}
In this frame, and after a minor rearrangement, Equation (\ref{boxopbasisdualbasis}) reads
\begin{equation}\label{boxopbasisdualbasisnull}
\square_g \phi = \underbrace{n(\nu^\sharp \phi) + t(\tau^\sharp \phi)}_{\textrm{(I)}} + \underbrace{\div n \cdot \nu^\sharp \phi}_{\textrm{(II)}} + \underbrace{m_i\big({\mu^i}^\sharp \phi\big) + \big(\div t \cdot \tau^\sharp + \div m_i \cdot {\mu^i}^\sharp\big) \phi}_{\textrm{(III)}}.
\end{equation}
The properties specific to the chosen local frame may now be used to further polish the expression denoted by \textrm{(I)}.  To wit,
\begin{equation*} n(\nu^\sharp \phi) = g(\tau^\sharp,\nu^\sharp) \cdot n(t\phi) + n(g(\tau^\sharp,\nu^\sharp)) \cdot t\phi + \cdots \end{equation*}
where the ellipses indicate here and below that the remaining terms are smooth linear combinations (with coefficients independent of $\phi$) of terms of the type $X \phi$ or $Y(X \phi)$ with $X$ and $Y$ vector fields in the collection $\{n, m_1, \ldots, m_{d-1}\}$ and $n$ always to the left of any of the $m_i$'s---in particular, no ``$t$ derivatives'' are present in those terms.  It will be convenient to introduce the notation $\beta \doteq g(\tau^\sharp, \nu^\sharp)$.  Recalling that $\xi \doteq \tau^\sharp - \beta n$ vanishes on $\Gamma$, for any function $\phi$ it follows that $V \xi \phi = [V,\xi] \phi$ on $\Gamma$.  Therefore, on $\Gamma$ all the following equalities hold:
\begin{multline*} t(\tau^\sharp \phi) = t(\beta \cdot n\phi + \xi\phi) = \beta \cdot t(n\phi) + t(\xi \phi) + \cdots = \beta \cdot t(n\phi) + [t, \xi] \phi + \cdots \\
= \beta \cdot t(n\phi) + \tau([t, \xi]) \cdot t\phi + \cdots  = \beta \cdot t(n\phi) + g(\tau^\sharp,[t, \xi]) \cdot t\phi + \cdots \\
= \beta \cdot  t(n\phi) + \beta g(n,[t, \xi]) \cdot t\phi + \cdots \ .
\end{multline*}
Putting together the results so far, and commuting vectors once more,
\begin{align*} \textrm{(I)} &= \beta \big\{n(t\phi) + t(n\phi)\big\} + \big\{\beta g(n,[t, \xi]) + n\beta\big\} \cdot t\phi + \cdots \\
&= 2\beta \cdot n(t\phi) + \beta \cdot [t,n]\phi +  \big\{\beta g(n,[t, \xi]) + n\beta\big\} \cdot t\phi + \cdots \\
&= 2\beta \cdot n(t\phi) + \Big\{\beta^2g(n,[t,n]) + \beta g(n,[t, \xi]) + n\beta\Big\} \cdot t\phi + \cdots
\end{align*}
on $\Gamma$.  Finally, one can deal with expressions \textrm{(II)} and \textrm{(III)} in a similar way.  In particular (and again with all equalities valid on $\Gamma$),
\begin{align*}
    \textrm{(II)} &= \beta \div{n} \cdot t \phi + \cdots \\
    \text{and} \qquad m_i\big({\mu^i}^\sharp \phi\big) &=
    m_i\left(\nu\big({\mu^i}^\sharp\big)n\phi + \mu^j\big({\mu^i}^\sharp\big)m_j\phi\right) \\
    &= g\big({\mu^i}^\sharp, \nu^\sharp\big)[m_i, n] \phi + \cdots \\
    &= \beta g\big({\mu^i}^\sharp, \nu^\sharp\big)g(n,[m_i, n]) \cdot t \phi + \cdots \ .
\end{align*}
As a result Equation (\ref{boxopbasisdualbasisnull}), when both sides are evaluated on $\Gamma$, becomes
\begin{equation}\label{boxopbasisdualbasisode} \square_g\phi =  2\beta \cdot n(t\phi) + \zeta \cdot t\phi + D \phi,
\end{equation}
where
\begin{equation}\label{zeta}\zeta \doteq \beta\Big\{\beta g(n,[t,n]) + g(n,[t, \xi]) + \div n + g\big({\mu^i}^\sharp, \nu^\sharp\big)g(n,[m_i, n])\Big\} + n\beta,\end{equation}
while $D$ is a second-order differential operator involving compositions of the vector fields $n$, $m_1$, \ldots, $m_{d-1}$ but no $m_i (n \phi)$ terms.

Suppose now that: $\gamma$ is parametrised as an integral curve of $n$; $P = \square_g + X + [q]$ with $X$ a smooth vector field and $q$ a smooth function; $F$ is a smooth function defined in a neighbourhood of $\Gamma$.  Let the following information be known about a $C^2$ function $\phi$ defined in a neighbourhood of $\Gamma$:
\begin{itemize}
\item its values on $\Gamma$;
\item the values of $m_i \phi$ and $m_i(m_j \phi)$ on $\Gamma$ ($i,j=1, \ldots, d-1$);
\item $P \phi = F$ on $\Gamma$ (but not necessarily elsewhere).
\end{itemize}
[Notice that the above data also uniquely determine $n \phi$, $n(n\phi)$ and $n(m_i \phi)$ on $\Gamma$.] Then, denoting the parameter along $\gamma$ by $s$ so that $n$ is just $\deriv{}{s}$ there, and decomposing the vector field $X$ in the chosen frame as $X = \tau(X)t + \tilde{X} = \beta g(n,X)t + \tilde{X}$, Equation (\ref{boxopbasisdualbasisode}) immediately yields
\begin{equation} 0 = (P \phi - F) \circ \gamma = 2\beta \deriv{(t\phi \circ \gamma)}{s} + \Big\{[\zeta + \beta g(n, X)] \cdot t\phi + D' \phi - F\Big\} \circ \gamma, \label{prototypepropagationeq}\end{equation}
where $D' \doteq D + \tilde{X} + [q]$ and $\zeta$ is as defined above. Equation (\ref{prototypepropagationeq}) is simply a first-order ODE determining the evolution of the only unknown $t \phi$ along $\gamma$.  This ODE is the prototype ``propagation equation''---to use the same terminology as in \cite{rendall1990reduction} and in related literature---for a linear, second-order, normally hyperbolic PDE along an arbitrary null curve and in a given choice of local frame.  Conversely, of course, any $\phi$ satisfying Equation (\ref{prototypepropagationeq}) will also solve $P \phi = F$ on $\Gamma$.

The situation of principal interest, namely that involving a smooth null \emph{hypersurface} $\scrN \subset M$, will now be examined.  Assume there exists a non-zero vector field $t$, defined and smooth in an open set $\calU$ with $\scrN \cap \calU \neq \emptyset$, and transverse to $\scrN$.\footnote{No further restrictions will be placed on the causal character of $t$ in this subsection.}  Analogous arguments to the ones used in the proof of Proposition \ref{nullhypglobalnullfield} show that $t$ uniquely singles out a smooth nowhere-vanishing section $n$ of the null line bundle of $\scrN \cap \calU$, such that $g(n,t) = 1$.  Lie transporting $n$ off $\scrN$ along the flow of $t$ further defines an extension of $n$---still denoted by $n$ in what follows---to a vector field defined and smooth in a sufficiently small open neighbourhood of $\scrN \cap \calU$ in $\calU$; by construction, $[t,n]=\Lie_t n=0$ there.  By shrinking $\calU$ if necessary, this neighbourhood may be assumed to be the whole of $\calU$.  In a similar way (again shrinking $\calU$ if necessary), $t$ also singles out a smooth one-form $\tau$ on $\calU$ by the requirements that $\tau(t) = 1$ and that $\ker \tau$ be invariant under the pushforward by the flow of $t$.  Let $m_1, \ldots, m_{d-1}$ be smooth vector fields on $\scrN \cap \calU$ commuting with the restriction of $n$ to $\scrN$, and such that $\{n, m_1, \ldots, m_{d-1}\}$ is a local frame field for $T \scrN$.  Once more (and possibly after shrinking $\calU)$, Lie transporting $m_1, \ldots, m_{d-1}$ off $\scrN$ along the flow of $t$ promotes them to smooth vector fields on $\calU$ (denoted by the same letters), and $\calF \doteq \{n, t, m_1, \ldots, m_{d-1}\}$ is a frame field for $\calU$.  $\tau$ was chosen so that $\tau(t)=1$, $\tau(n) = \tau(m_i)= 0$ for each $i=1,\ldots,d-1$.  Hence, the dual coframe field to $\calF$ is $\calF^* \doteq \{\nu, \tau, \mu^1, \ldots, \mu^{d-1}\}$ for some $\nu, \mu^1, \ldots, \mu^{d-1}$.  Moreover, $\tau^\sharp = n$ on $\scrN \cap \calU$ and hence also $\beta = g(\tau^\sharp, \nu^\sharp) = 1$ there.  It is then not hard to see that, on $\scrN \cap \calU$, the function $\zeta$ defined in (\ref{zeta}) simply equals $g(n,[t,\tau^\sharp]) + \div n$.  Accordingly, a statement analogous to Equation (\ref{boxopbasisdualbasisode}) holds: defining the second-order differential operator
\begin{equation} D_{(\scrN ,t)} \doteq \square_g - 2nt - [g(n,[t,\tau^\sharp]) + \div n] \cdot t,\label{DscrNt}\end{equation}
and letting $i : \scrN \cap \calU \to \calU$ be the inclusion, one sees that $i^*[D_{(\scrN ,t)}\phi]=0$ whenever $i^*\phi = 0$.  It follows---from the Peetre theorem \cite{peetre1959characterisation} if one wishes to adopt such an abstract viewpoint---that there exists a unique linear and at most second-order differential operator $\Delta_{(\scrN ,t)}$ on $\scrN$ such that $i^* \circ D_{(\scrN ,t)} = \Delta_{(\scrN,t)} \circ i^*$.  More generally, if $L$ is a differential operator on $M$, $i : S \to M$ is a smoothly embedded submanifold, and $i^*[L \chi] = 0$ whenever $i^*\chi = 0$, there exists a unique differential operator $\Lambda$ on $S$, of at most the same order as $L$, such that $i^* \circ L = \Lambda \circ i^*$.  The phrase ``$L$ pulls back to $\Lambda$'' or, symbolically, $i^*L = \Lambda$, will be used to describe this situation compactly.  The following theorem, which is a ``null hypersurface version'' of Equation (\ref{prototypepropagationeq}), follows.
\begin{theorem}\label{characteristicdecompositionthm}
    Let $(M,g)$ be a Lorentzian manifold, $\scrN$ be a smooth null hypersurface, and $t$ be a smooth non-zero vector field everywhere transverse to $\scrN$.  Let $P \doteq \square_g + X + [q]$, where $X$ is a smooth vector field and $q$ is a smooth function. All objects apart from the metric $g$ need only be defined in a suitable open set $\calU$ in $M$ intersecting $\scrN$, and the embedding of $\scrN \cap \calU$ into $\calU$ will be denoted by $i$.  Then
     \begin{equation*} P \phi = 2n(t\phi) + K_{(\scrN,t,X)} \cdot t\phi + D_{(\scrN,t,X,q)}\phi,\end{equation*}
    where $K_{(\scrN,t,X)} \doteq g(n,[t,\tau^\sharp]) + \div n + g(n, X)$, $D_{(\scrN,t,X,q)} \doteq D_{(\scrN,t)} + X - g(n,X)t + [q]$, $D_{(\scrN,t)}$ is defined in Equation (\ref{DscrNt}), and the vector fields $n$ and $\tau^\sharp$ are obtained in a canonical way from the pair $(\scrN,t)$ and the metric $g$ in the manner described above.  $D_{(\scrN,t,X,q)}$ is a second-order differential operator which pulls back to a differential operator on $\scrN \cap \calU$.
      More generally, for any $r \in \N_0$,
      \begin{equation}\label{highorder} t^rP\phi = 2n(t^{r+1}\phi) + K^{(r)} \cdot t^{r+1}\phi + \sum_{k=0}^r D^{(r)}_k t^k \phi, \end{equation}
      where each $K^{(r)}$ is a smooth function independent of $\phi$ and each $D^{(r)}_k$ is a linear differential operator of order at most $2$, independent of $\phi$, which pulls back to a differential operator on $\scrN \cap \calU$.
\end{theorem}
\begin{proof}
Equation (\ref{highorder}) will be proved by induction on $r$.  The base case $r=0$ has already been proven---with $K^{(0)} \doteq K_{(\scrN,t,X)}$ and $D^{(0)}_0 \doteq D_{(\scrN,t,X,q)}$.  Suppose the statement holds true for $r$. Then, since $[t,n] = 0$,
\begin{align*}t^{r+1}P &= 2nt^{r+2} + K^{(r)} \cdot t^{r+2} + tK^{(r)} \cdot t^{r+1} + \sum_{k=0}^r t D^{(r)}_k t^k.
\end{align*}
The last two terms can be rewritten as
\begin{align*}
t(K^{(r)}) \cdot t^{r+1} + \sum_{k=0}^r t D^{(r)}_k t^k &= tK^{(r)} \cdot t^{r+1} + \sum_{k=1}^{r+1} D^{(r)}_{k-1} t^{k} + \sum_{k=0}^r [t, D^{(r)}_k] t^k \\
&= \{D^{(r)}_r + [tK^{(r)}]\} t^{r+1} + \sum_{k=1}^{r} D^{(r)}_{k-1} t^k + \sum_{k=0}^{r}[t, D^{(r)}_k] t^k.
\end{align*}
Clearly, one need only show that the last term is of the required form, i.e.\ that $\sum_{k=0}^r[t, D^{(r)}_k] t^k = \sum_{k=0}^{r+1} L_k t^k$ where each $L_k$ is a differential operator of order at most $2$ and which pulls back to a differential operator on $\scrN \cap \calU$.  Each $[t, D^{(r)}_k]$ is of order at most $1$.  Therefore, shrinking the neighbourhood $\calU$ further if necessary, $[t, D^{(r)}_k]$ admits a ``decomposition'' adapted to $t$ and $\scrN \cap \calU$ in the sense that $[t, D^{(r)}_k] = h_k \cdot t + R_k$ where each $h_k$ is a smooth function and each $R_k$ is a first-order differential operator on $\calU$ which pulls back to a differential operator on $\scrN \cap \calU$. [To prove this, one may first cover $\scrN \cap \calU$ by open neighbourhoods equipped with coordinates adapted to $\scrN \cap \calU$ and to the flow of $t$; expressing the differential operator locally in each coordinate neighbourhood, and running a simple partition of unity argument, then yields the desired global form.]  Thus,
\begin{equation*} \sum_{k=0}^r[t, D^{(r)}_k] t^k = \sum_{k=1}^{r+1} h_{k-1} \cdot t^{k} + \sum_{k=0}^r R_k t^k, \end{equation*}
and the claim immediately follows.
\end{proof}

\begin{corollary}\label{characteristicdecompositioncor}
Let $(M,g)$, $\scrN$, $t$, $n$, $P$, $\calU$ and $i: \scrN \cap \calU \to \calU$ be as in the statement of Theorem \ref{characteristicdecompositionthm}.  Let also $F$ be a given smooth function on $\calU$.  Then, for any $r \in \N_0$, there exist smooth functions $\kappa^{(r)}$ and a collection $\{\Delta^{(r)}_k\}_{k = 0, \ldots, r}$ of differential operators on $\scrN \cap \calU$ of order at most $2$, such that
\begin{equation}\label{arbitraryorderpdeonN}i^*t^r[P \phi - F] = 2\mathfrak{n}\phi_t^{(r+1)} + \kappa^{(r)} \cdot \phi_t^{(r+1)} + \sum_{k=0}^r\Delta^{(r)}_k\phi_t^{(k)} - i^*[t^rF] \quad \forall \ \phi \in C^\infty(\calU),\end{equation}
    where $\mathfrak{n} \doteq i^*n \in \mathfrak{X}(\scrN \cap \calU)$ and $\phi_t^{(r)} \doteq i^*[t^r\phi]$.  Consequently, $\phi$ solves $P\phi = F$ on $\scrN \cap \calU$ to order $r$ if and only if, for every $\ell = 0, \ldots, r$, its transverse derivative $\phi_t^{\ell + 1}$ restricted to $\scrN \cap \calU$ is compatible with the collection $\{\phi_t^j\}_{j=0,\ldots,\ell}$ of all lower-order transverse derivatives restricted to $\scrN \cap \calU$, in the sense that the right-hand side of Equation (\ref{arbitraryorderpdeonN}) equals zero. \qed
\end{corollary}

\subsection{Existence of solutions to two-sided characteristic Cauchy problems}\label{existencesec}

Throughout this subsection, $\scrM = (M,g,\mathfrak{t})$ is a given globally hyperbolic Lorentzian manifold of dimension $d+1$, and $\scrN$ is an achronal, closed, smooth null hypersurface whose null generators, when reparametrised as null geodesics entirely contained in $\scrN$, are future and past inextensible in $\scrM$.  The reader is referred to Proposition \ref{intersectJN} for facts relevant to this setup, and for notation introduced there and adopted in this section.  Let $P$ be a normally hyperbolic, scalar differential operator of the form given in Equation (\ref{normallyhyperbolicscalar}).  Let the (extended) Green operators $\rmG_+, \rmG_-$, and the (extended) causal Green operator $\rmG = \rmG_+ - \rmG_-$, be as in Section \ref{greenhypsubsec}.

Building on the preparatory work in the previous subsection, the existence and regularity properties of solutions to ``two-sided'' characteristic Cauchy problems will now be established.  Let $t$ be a global, smooth, timelike vector field---clearly, $t$ is transverse to $\scrN$.  Corollary \ref{characteristicdecompositioncor} gives the form, relative to $t$, of the sequence of restrictions to $\scrN$ of the differential consequences of $P \phi = F$ obtained by taking arbitrarily many $t$-derivatives of the equation.  Now let $\scrS$ be an arbitrary cross-section $\scrN$, which exists under the current assumptions by Proposition \ref{crosssecglobhyp}.  Using $\scrS$ together with the $\scrN$-vector field $\mathfrak{n}$ introduced in the statement of Corollary \ref{characteristicdecompositioncor}, a diffeomorphism $\chi$ may be obtained from $\scrN$ to an open subset $\calW$ of $\R_s \times \scrS$ containing $\{ 0 \} \times \scrS$.\footnote{This diffeomorphism amounts to choosing the zero value of the parameter on each integral curve of $\mathfrak{n}$ to be at the curve's unique intersection with $\scrS$.}  The condition of vanishing of $t^r[P \phi - F]$ on $\scrN$ becomes, by Equation (\ref{arbitraryorderpdeonN}),
\begin{equation}\label{arbitraryorderodeonN}
 2\pderiv{\widetilde{\phi_t^{(r+1)}}}{s} + \widetilde{\kappa^{(r)}} \cdot \widetilde{\phi_t^{(r+1)}} + \sum_{k=0}^r\widetilde{\Delta^{(r)}_k \phi_t^{(k)}} - \widetilde{i^*[t^rF]} = 0 \qquad \text{on $\calW$},
\end{equation}
where tildes on functions denote pre-compositions with the inverse of $\chi$.  If the last two terms on the left-hand side of Equation (\ref{arbitraryorderodeonN}) are regarded as \emph{known} quantities and denoted collectively by $\widetilde{V^{(r)}} \in C^\infty(\calW)$, then one is simply looking for solutions $u$ to
\begin{equation}\label{arbitraryorderodeonN2}
 2\pderiv{u}{s} + \widetilde{\kappa^{(r)}} \cdot u + \widetilde{V^{(r)}} = 0 \qquad \text{on $\calW$}.
\end{equation}
Letting $I_x \subseteq \R$, for each $x \in \scrS$, denote the projection onto the first factor of $\calW \cap (\R \times \{ x \})$, a single smooth solution $u$ to the above equation on $\calW$ is equivalently characterised by a ``smooth family'' $\{ u_x \in C^\infty(I_x)\}_{x \in \scrS}$ of functions, each solving the ODE (on $I_x$) obtained by replacing $u$, $\widetilde{\kappa^{(r)}}$ and $\widetilde{V^{(r)}}$ in Equation (\ref{arbitraryorderodeonN2}) by $u_x$, $\widetilde{\kappa^{(r)}}(\cdot,x)$ and $\widetilde{V^{(r)}}(\cdot,x)$ respectively.  It is in this sense that the vanishing of $t^r[P \phi - F]$ on $\scrN$---given prior knowledge of the (smooth) restrictions of $t^\ell \phi$ to $\scrN$ for each $\ell=0,\ldots,r$---can be regarded as equivalent to the existence of a ``smooth family'' of solutions to a  family of ODEs with coefficient functions smoothly dependent on parameters in $\scrS$.  Since these ODEs are linear (up to the presence of a inhomogeneity), standard results on the smooth dependence of solutions of ODEs on parameters and initial conditions (see, e.g., \cite[Sec.\ 1.7]{coddington1955theory}) guarantee that, if smooth initial conditions are imposed on $\{0\} \times \scrS$, a solution $u$ in $\calW$ exists, is unique, and is everywhere smooth.\footnote{These results are only local in the case of nonlinear ODEs, but become global if the ODEs are linear.  Note also that the potential dependence of the ODE domain $I_x$ on the parameter $x$ does not cause trouble here.}

The above analysis, together with the last statement in Corollary \ref{characteristicdecompositioncor}, suggests the following procedure for obtaining ``transverse derivative data up to order $r+1$'' on $\scrN$ which may arise from a function $\phi$, defined and $C^{r+2}$ in an open neighbourhood of $\scrN$ and which solves $P\phi = F$ to order $r$ on $\scrN$ (but not necessarily elsewhere):
\begin{enumerate}[I.]
\item Set $\phi_t^{(0)} \doteq f$ for some $f \in C^\infty(\scrN)$;
\item by imposing the initial conditions $\phi_t^{(\ell+1)}\restriction_\scrS=g^{(\ell+1)}$ for given $g^{(\ell)} \in C^\infty(\scrS)$, recursively solve each ``$\ell$th order propagation equation'', i.e.\ Equation (\ref{arbitraryorderpdeonN}) with $r$ replaced by $\ell$ and the latter ranging from $0$ to $r$, to uniquely obtain the order $\ell+1$ transverse derivative datum $\phi_t^{(\ell+1)} \in C^\infty(\scrN)$.
\end{enumerate}
Clearly, the procedure may accommodate $r=\infty$, and indeed this is the case of interest in the remainder of this paper. While, in the setting of a general characteristic Cauchy problem, the initial value $f$ is precisely what is assumed given as part of the problem, the cross-section $\scrS$ and sequence $( g^{(r)} )_{r=0}^\infty$ of initial conditions on $\scrS$ are not.  However, in the present context of wave-like equations---which are well known to enjoy finite speed of propagation properties---natural choices for these do arise if one also puts some restrictions on the supports of the inhomogeneity $F$ and of the initial value $f$.

Namely, under the assumptions and notation of this subsection, let 
\begin{equation}\label{scrFdef}
\scrF \doteq \overline{\bigcup_{n=0}^\infty \supp{[i^*(t^n F)]}};
\end{equation}
since $\scrN$ is assumed closed, the closure above may equally well be taken in $M$ or in $\scrN$ with the relative topology, and $\scrF \subseteq \scrN$.  In what follows, $\supp{f}$ and $\scrF$ will be freely regarded as subsets of $M$.  If both these sets are past [resp.\ future] compact, then it is possible to find a smooth spacelike Cauchy surface $\mathscr{C}^-$ [resp.\ $\mathscr{C}^+$] such that 
\begin{equation}\label{cauchysurfchardata}
\supp{f} \cup \scrF \subset I^+(\mathscr{C}^-)\quad \text{[resp. } \supp{f} \cup \scrF \subset I^-(\mathscr{C}^+)\text{].}
\end{equation}
The set $\scrS^+ \doteq \mathscr{C}^+ \cap \scrN$ [resp.\ $\scrS^- \doteq \mathscr{C}^- \cap \scrN$] is a cross-section of $\scrN$ by Lemma \ref{lemintersectonce}.  By construction and by the recursive structure of the tower of parametrised ODEs in question, the following is clearly true of any sequence $(\phi_t^{(r)})_{r=0}^\infty$ of transverse derivative data on $\scrN$ compatible with the characteristic Cauchy problem and with its differential consequences.

\begin{lemma}\label{lemmaodes} Let $r \in \N_0$, and suppose that $\supp{f}$ and $\scrF$ are past compact subsets of $M$.  Then, using the notation introduced in Proposition \ref{intersectJN} and letting the cross-section $\scrS^-$ of $\scrN$ be as just described, $\supp{\phi_t^{(r+1)}} \subseteq \supp{f}^+ \cup \scrF^+$ whenever $\supp{\phi_t^{(r)}} \subseteq \supp{f}^+ \cup \scrF^+$ and $\phi_t^{(r+1)} = 0$ on $\scrS^-$.  An analogous result (given by replacing $+$ with $-$) holds if instead $\supp{f}$ and $\scrF$ are future compact subsets of $M$. \qed
\end{lemma}

It will henceforth be assumed that $\supp{f}$ and $\scrF$ are simultaneously past and future compact subsets of $M$; that is, both sets are \emph{temporally} compact.  Then, with $\scrS^\pm$ as above, the sequence of $\scrS^-$-parametrised initial conditions given by $\phi_t^{(r)}\restriction_{\scrS^-} = 0$ for all $r \geq 1$ will be said to be \emph{of future type}; similarly, the sequence of $\scrS^+$-parametrised initial conditions, given by $\phi_t^{(t)}\restriction_{\scrS^+} = 0$ for all $r \geq 1$, will be referred to as being \emph{of past type}.  It is clear that any two distinct Cauchy surfaces $\mathscr{C}_1^-, \, \mathscr{C}_2^-$ for which (\ref{cauchysurfchardata}) holds yield the same infinite sequence of compatible transverse derivative data if initial conditions of future type are imposed in both cases; a similar statement holds for two distinct $\mathscr{C}_1^+, \, \mathscr{C}_2^+$ and initial conditions of past type.  Therefore, the notion just introduced is actually independent of all choices made.  The two sequences arising in this way from the pair $(f,F)$ will henceforth be denoted by $(\phi_{t, +}^{(r)})_{r=0}^\infty$ (future type) and $(\phi_{t, -}^{(r)})_{r=0}^\infty$ (past type).  By virtue of Lemma \ref{lemmaodes} one may characterise the two types in more succinct and geometrical terms as follows:
\begin{enumerate}[leftmargin=0.5in]
	\item[\textbf{future}] $\phi_{t,+}^{(0)} = f$ and $\supp{\phi_{t,+}^{(r)}}  \subseteq \supp{f}^+ \cup \scrF^+$ for all $r \geq 1$;
	\item[\textbf{past}] $\phi_{t,-}^{(0)} = f$ and $\supp{\phi_{t,-}^{(r)}}  \subseteq \supp{f}^- \cup \scrF^-$ for all $r \geq 1$.
\end{enumerate}
By the global Borel Lemma (Theorem \ref{globalborel}) there exists a globally defined function $\phi_{\mathrm{app}}^+ \in C^\infty(M)$ whose $t$-derivative to order $r$, restricted to $\scrN$, equals the $r$th term in the sequence $(\phi_{t, +}^{(r)})_{r=0}^\infty$.  Similarly, there exists a globally defined function $\phi_{\mathrm{app}}^- \in C^\infty(M)$ whose $t$-derivative to each order $r \in \N_0$, restricted to $\scrN$, equals the $r$th term in the sequence $(\phi_{t, -}^{(r)})_{r=0}^\infty$.  As a result, $\phi^\pm_{\mathrm{app}}\restriction_\scrN = f$ and $P\phi_{\mathrm{app}}^\pm - F$ vanishes to infinite order on $\scrN$---that is, $\phi^\pm_{\mathrm{app}}$ is an ``approximate'' solution to the characteristic Cauchy problem.  What's more, the explicit procedure given in the proof of Theorem \ref{globalborel}---with the globally timelike vector field $t$ playing the role of $V$ there---clearly shows that one can ensure that $\supp{\phi^\pm_{\mathrm{app}}} \subseteq J(\scrN)$ and that
\begin{equation*}
\supp{\phi^\pm_{\mathrm{app}}} \cap J^\pm(\scrN) \subseteq J^\pm\left( \bigcup_{r \in \N_0} \supp{\phi_{t, \pm}^{(r)}} \right) \subseteq J^\pm\left(\supp{f}^\pm \cup \scrF^\pm\right) \subseteq J^\pm\left(\supp{f} \cup \scrF \right).
\end{equation*}
In particular, by Lemma \ref{causcompletelem2} and the assumptions on $\supp{f}$ and $\supp{F}$ made so far, and letting $\bm{1}^\pm$ denote the indicator function of $J^\pm(\scrN)$, $\supp{(\bm{1}^+ P\phi^+_{\mathrm{app}})}$ is past compact and $\supp{(\bm{1}^- P\phi^-_{\mathrm{app}})}$ is future compact.  If it is also the case that $\supp{F} \subseteq J(\scrN)$, then in turn $\supp{(P\phi_{\mathrm{app}}^\pm - F)} \subseteq J(\scrN)$.  In particular, $P\phi_{\mathrm{app}}^\pm - F$ vanishes to infinite order on $\partial J(\scrN)$---as well as on $\scrN$ as already mentioned.  By item \ref{JN} in Proposition \ref{intersectJN}, this means that $P \phi^\pm_\mathrm{app} - F$ vanishes to infinite order on $\partial J^\pm(\scrN)$.  It follows that
\begin{equation*}
e^\pm \doteq \bm{1}^\pm [P \phi^\pm_\mathrm{app} - F] \end{equation*}
which, on $J^\pm(\scrN)$, describes the failure of $\phi_{\mathrm{app}}^\pm$ to be a true solution to the characteristic Cauchy problem, is smooth on $M$.  Furthermore, the support property
\begin{equation*}
\supp{e^\pm} \subseteq \supp{(\bm{1}^\pm P \phi^\pm_{\mathrm{app}}}) \cup \supp{(\bm{1}^\pm F}),
\end{equation*}
in conjunction with the observations made above, implies that $e^+ \in C^\infty_\pc(M)$ if $\supp{(\bm{1}^+ F})$ is past compact and $e^- \in C^\infty_\fc(M)$ if $\supp{(\bm{1}^- F})$ is future compact. 
\begin{lemma}\label{scrFlem}
Under the geometric assumptions listed in the first paragraph of this section, let $F \in C^\infty(M)$.  Let $\bm{1}^\pm$ denote the indicator function of $J^\pm(\scrN)$.  If $\supp{( \bm{1}^+ F )}$ is past compact then the set $\scrF$ defined by Equation (\ref{scrFdef}) is past compact.  Similarly, if $\supp{( \bm{1}^- F )}$ is future compact then so is $\scrF$.
\end{lemma}
\begin{proof}
Suppose that $\supp{( \bm{1}^+ F )}$ is past compact (the proof of the time-reversed statement is of course completely analogous).  Then, by Lemma \ref{causcompletelem3}, there exists a Cauchy surface $\scrC$ such that $\supp{( \bm{1}^+ F )} \subseteq J^+(\scrC)$.  Hence, $\bm{1}^+ F = 0$ on the open set $I^-(\scrC)$, and this in turn can only hold if $F$ vanishes on $J^+(\scrN) \cap I^-(\scrC)$.  In particular, $F$ vanishes to infinite order on $\scrN \cap I^-(\scrC)$, and the restriction of any partial derivative of $F$ to $\scrN$ has support contained in the fixed closed and past compact set $\scrN \cap J^+(\scrC)$.
\end{proof}

In view of the analysis so far, some compact notation will now be introduced.
\begin{definition}\label{defncharacteristicdata} Under the geometric assumptions listed in the first paragraph of this section, and letting $\bm{1}^\pm$ denote the indicator function of $J^\pm(\scrN)$, the following linear function spaces will be defined for any $k \in \N_0 \cup \{ \infty\}$:
\begin{align*}
C^k_\tc(\scrN) \doteq \left\{ f \in C^k(\scrN) \relmiddle| \right.& \supp{f} \text{ is temporally compact in } \scrM \big\}; \\
C^k_\scrN(\scrM) \doteq \left\{ F \in C^k(M) \relmiddle| \right. & \supp{F} \subseteq J(\scrN), \ \supp{(\bm{1}^+ F )} \text{ is past compact in } \scrM,  \\ 
& \supp{(\bm{1}^- F )} \text{ is future compact in } \scrM \big\}.
\end{align*}
\end{definition}
Returning to the discussion preceding the statement of Lemma \ref{scrFlem}, assume that $f \in C^\infty_\tc(\scrN)$ and $F \in C^\infty_\scrN(\scrM)$, so that $e^+$ has past compact support and $e^-$ has future compact support.  Then, upon defining
	\begin{equation}\label{defnofphipm} \phi_{(f,F)}^\pm \doteq \phi^\pm_\mathrm{app} - \rmG_\pm e^\pm, \end{equation}
	the following hold: 
	\begin{enumerate}[(a)]
	\item $\phi^\pm_{(f,F)} \in C^\infty(M)$;
	\item $\phi^\pm_{(f,F)}\restriction_\scrN = f$;
	\item $P \phi_{(f,F)}^\pm = P\phi^\pm_\mathrm{app} - e^\pm = P\phi^\pm_\mathrm{app} - \bm{1}^\pm [P \phi^\pm_\mathrm{app} - F] = F$ on $J^\pm(\scrN)$;
	\item $\supp{\phi^\pm_{(f,F)}} \subseteq J(\scrN)$, $\supp{\big(\bm{1}^+ \phi^+_{(f,F)}\big)}$ is past compact and $\supp{\big(\bm{1}^- \phi^-_{(f,F)}\big)}$ is future compact.
	\end{enumerate}
Now consider the function $\phi_{(f,F)} : M \to \R$ defined as follows:
\begin{equation}\label{defnofphif}
	\phi_{(f,F)}(x) = \begin{cases} \phi_{(f,F)}^+(x) &\text{if } x \in J^+(\scrN) \\ \phi_{(f,F)}^-(x) &\text{if } x \in J^-(\scrN) \\
0 &\text{if } x \notin J(\scrN).
	\end{cases}
\end{equation}
Since $ J^+(\scrN) \cap J^-(\scrN) = \scrN$ [see item \ref{J+capJ-} in Proposition \ref{intersectJN}] and both $\phi_{(f,F)}^+$ and $\phi_{(f,F)}^-$ restrict to $f$ on $\scrN$, $\phi_{(f,F)}$ is well-defined.  Together with the fact that $\phi_{(f,F)}^\pm$ vanishes (to infinite order) on $\del{J}^\pm(\scrN) {\setminus } \scrN$, this guarantees that $\phi_{(f,F)}$ is globally continuous.  In addition, the construction ensures that $\phi_{(f,F)}$ is smooth on the open set $M {\setminus } \scrN$, and solves $P \phi_{(f,F)} = F$ there.  In fact, more is true: the partial derivatives to all orders of the restriction of $\phi_{(f,F)}$ to either $I^+(\scrN)$ or $I^-(\scrN)$ can be continuously extended to $\scrN$.  Equivalently, if one regards $J^+(\scrN)$ and $J^-(\scrN)$ as manifolds with boundary, the smooth boundary being $\scrN$ in both cases, then the restriction of $\phi_{(f,F)}$ to $J^\pm(\scrN)$ is a smooth function in the sense appropriate to manifolds with boundaries.  However, the two resulting sets of partial derivative extensions to $\scrN$ will differ in general, preventing $\phi_{(f,F)}$ from being globally $C^1$ and thus a fortiori a global classical (i.e.\ $C^2$) solution.  On the other hand, by standard results \cite[Thm.\ 3.5.1]{agranovich2015sobolev} $\phi_{(f,F)} \in H^s_\loc(M)$ (the space of locally Sobolev functions of order $s$ on $M$) for all $s < 3/2$.  To summarise:

\begin{theorem}\label{existencethmnodistrsol}
Let $\scrM = (M,g,\mathfrak{t})$ be a globally hyperbolic Lorentzian manifold and $\scrN$ be an achronal, closed, smooth null hypersurface whose null generators, when reparametrised as null geodesics entirely contained in $\scrN$, are future and past inextensible in $\scrM$.  Let $P$ be a normally hyperbolic scalar operator.  Then, for any $f \in C^\infty_\tc(\scrN)$ and $F \in C^\infty_\scrN(\scrM)$, there exists a function $\phi_{(f,F)} \in C^0_\scrN(\scrM)$ with the following properties: (a) it equals $f$ on $\scrN$; (b) it is smooth on $M {\setminus } \scrN$ and solves $P\phi_{(f,F)} = F$ there; (c) its restriction to either $I^+(\scrN)$ or $I^-(\scrN)$ has partial derivatives to all orders which can be continuously extended to $\scrN$. \qed
\end{theorem}

Theorem \ref{existencethmnodistrsol} does not yet establish that $P\phi_{(f,F)} = F$ everywhere in the distributional sense.  Equation (\ref{commidentity2}) will now be used to establish this.\footnote{The author would like to thank A.\ Strohmaier for suggesting the plan of attack used here.}  Recall the well-known fact \cite[Prop.\ 6.3.1]{hawking1973large} that, in any time-oriented Lorentzian manifold, the topological boundary of the causal (equivalently, chronological) future or past of an arbitrary subset is a (closed, achronal and) locally Lipschitz topological hypersurface.

\begin{theorem}\label{distrsolnthm}
	Under the geometric assumptions of Theorem \ref{existencethmnodistrsol}, let $f \in C^\infty(\scrN)$ and $F \in C^\infty(M)$ with $\supp{F} \subseteq \overline{J(\scrN)}$.  Further let $\phi^+$ and $\phi^-$ be functions in $C^\infty(M)$ and with the following properties: (a) $\phi^+\restriction_\scrN = \phi^-\restriction_\scrN = f$; (b) $P \phi^+ = F$ on $J^+(\scrN)$ and $P \phi^- = F$ on $J^-(\scrN)$; (c) $\phi^+$ vanishes to first order on $\del{J}^+(\scrN) {\setminus } \scrN$ and $\phi^-$ vanishes to first order on $\del{J}^-(\scrN) {\setminus } \scrN$.  Define $\phi \in C^0(M)$ by
	\begin{equation*}
	\phi(x) = \begin{cases} \phi^+(x) &\text{if } x \in J^+(\scrN) \\ \phi^-(x) &\text{if } x \in J^-(\scrN) \\
0 &\text{if } x \notin J(\scrN)
	\end{cases}
	\end{equation*}
[cf.\ (\ref{defnofphif})].  Then $P \phi = F$ on $M$ in the distributional sense.
\end{theorem}
\begin{proof}
	The equality to prove is 
	\begin{equation*}\int_M \left(\phi P^\dagger\chi - F \chi\right) \mu_g = 0 \quad \forall \ \chi \in C_{\c}^\infty(M).\end{equation*}
	By assumption, $\supp{F} \cup \supp{\phi} \subseteq \overline{J(\scrN)}$.  Using item \ref{JN} in Proposition \ref{intersectJN} since $\scrN$ is assumed achronal, $\overline{J(\scrN)} {\setminus } J(\scrN) = \del{J}(\scrN) \subseteq \del{J}^+(\scrN) \cup \del{J}^+(\scrN)$; hence, this set has zero measure.  Similarly, $J^+(\scrN) \cap J^-(\scrN) = \scrN$ has zero measure.  Together with Equation (\ref{greenformscalarnormhyp}), these considerations imply that
	\begin{multline*}
		\int_M \phi P^\dagger \chi \, \mu_g = \sum_{\pm}\int_{J^\pm(\scrN)} \phi P^\dagger \chi \, \mu_g = \sum_{\pm}\int_{J^\pm(\scrN)} \phi^\pm P^\dagger \chi \, \mu_g \\
= \sum_{\pm}\int_{J^\pm(\scrN)} \left([P \phi^\pm] \chi - \div{j[\chi, \phi^\pm]} \right) \mu_g.
\end{multline*}
On $J^\pm(\scrN)$, $P \phi^\pm = F$ by assumption.  Hence,
\begin{equation}\label{divterm}
\int_M \left(\phi P^\dagger \chi - F \chi \right) \mu_g = - \sum_{\pm}\int_{J^\pm(\scrN)} \div{j[\chi, \phi^\pm]} \, \mu_g,
\end{equation}
and it remains to show that the right-hand side of Equation (\ref{divterm}) vanishes.  The vector field $j[\chi, \phi^\pm]$ is smooth and has compact support, and $\del{J}^\pm(\scrN)$ is a locally Lipschitz topological hypersurface.  Hence, by Equation (\ref{commidentity2}),
\begin{equation}\label{divterm2} \int_{J^\pm(\scrN)} \div j[\chi, \phi^\pm] \, \mu_g = \int_{\del{J}^\pm(\scrN)} n^\pm\big(j[\chi, \phi^\pm]\big) \, \iota_{n^\pm} \mu_g, \end{equation}
where $n^+$ and $n^-$ are any fields of (almost everywhere defined) outward-pointing conormals to $\del{J}^\pm(\scrN)$.  By Equation (\ref{greenformscalarnormhyp}) and assumption \emph{(c)} in the statement of this theorem, no contribution to the integral on the right-hand side of Equation (\ref{divterm2}) comes from $\del{J}^\pm(\scrN) {\setminus } \scrN$.  On $\scrN$ one may take $n^- = -n^+$ and, letting $n \doteq n^+\restriction_{\scrN}$:
\begin{equation*} \sum_{\pm} \int_{\scrN} n^\pm\big(j[\chi, \phi^\pm]\big) \, \iota_{n^\pm} \mu_g = \int_{\scrN} n\big(j[\chi, \phi^+ - \phi^-]\big) \, \iota_{n} \mu_g.
\end{equation*}
Since $\scrN$ is null, $n^\sharp$ is everywhere tangent to it.  Since $\phi^+\restriction_{\scrN} = \phi^-\restriction_{\scrN}= f$ by construction, both $v \doteq \phi^+ - \phi^-$ and $n^\sharp v$ vanish identically on $\scrN$.  Therefore,
\begin{equation*} n\big(j[\chi, v]\big) = \chi n^\sharp v - v n^\sharp \chi + v \chi \, n(X) = 0 \quad \text{on }\scrN, \end{equation*}
and the proof is complete.
\end{proof}

\subsection{Regularity of solutions to two-sided characteristic Cauchy problems}\label{regularitythm}

It is worthwhile to further comment on the regularity of the solutions $\phi_{(f,F)}$ constructed in the proof of Theorem \ref{existencethmnodistrsol}.  As already stated, for any $(f,F) \in C^\infty_\tc(\scrN) \times C^\infty_\scrN(\scrM)$ it is automatic that $\phi_{(f,F)} \in C^0(M) \cap H^s_\loc(M)$ for all $s < 3/2$.  Obstructions to higher regularity come from the fact that $\phi_{(f,F)}$ is obtained by ``merging'' together two functions whose derivatives in directions transverse to the characteristic hypersurface $\scrN$ are obtained by recursively solving the aforementioned propagation equations on $\scrN$ with two strictly \emph{different} sets of initial conditions---referred to as conditions of past and future type in the previous subsection.  Clearly, if $k \geq 1$ then $\phi_{(f,F)}$ is in $C^k(M)$ \footnote{In which case, again by standard results \cite[Thm.\ 3.5.1]{agranovich2015sobolev}, it is also in $H^{k+s}_\loc(M)$ for all $s < 3/2$.} if and only if, for the pair $(f,F)$, the resulting sequences $(\phi^{(r)}_{t,+})_{r=0}^\infty$ and $(\phi^{(r)}_{t,-})_{r=0}^\infty$---of transverse derivative data of future and past type (respectively) on $\scrN$---agree up to and including the $r=k$ term.  Since the propagation equations are ODEs along the null generators of $\scrN$, this in turn holds if and only if, for each $0 < r \leq k$ and for each null generator $\Gamma$ of $\scrN$, there exists a single point $x_{r, \Gamma} \in \Gamma$ such that $\phi^{(r)}_{t,+}(x_{r, \Gamma}) = \phi^{(r)}_{t,-}(x_{r, \Gamma})$.  Yet another equivalent condition is the following: for each $0 < r \leq k$, it holds that $\phi^{(\ell)}_t \doteq \phi^{(\ell)}_{t,+} = \phi^{(\ell)}_{t,-}$ for all $\ell = 0,\ldots, r-1$ and that, for each null generator $\Gamma$ of $\scrN$, the two-parameter flow solving the $r$th order propagation equation along $\Gamma$---itself determined uniquely by the collection $(\phi^{(\ell)}_t)_{\ell=0}^{r-1}$ together with $F$, the coefficients of $P$ and the geometry of $\scrN$---takes zero in the far past (according to the parameter along the null generator) to zero in the far future.

\subsection{Jump formulae with null boundaries}\label{commidentitiessec}

The fundamental jump formula, Equation (\ref{commidentity2}), already played a pivotal role in establishing the distributional solution property in Theorem \ref{distrsolnthm}.  In the remainder of this paper, it will be further leveraged to prove uniqueness and to obtain compact representation formulae for solutions.  In what follows, notation and results from Appendix \ref{gencommidentityapp} will be used.  Note that the letter ``$N$'' will be temporarily used in favour of ``$M$'' to denote a generic Lorentzian manifold; this is to avoid later confusion, since the results of this subsection will eventually be applied to a (possibly proper) open submanifold of the ambient manifold $M$ seen so far. 

\begin{theorem}\label{thmcommidentity} Let $P = \square_g + X  +[q]$ be a scalar normally hyperbolic operator on a Lorentzian manifold $(N, g)$.  Let $D \subseteq N$ be a domain with regular and everywhere null boundary, and assume that there exists a smooth, everywhere outward-pointing conormal field $n$ along $\del{D}$.\footnote{The latter condition would automatically be satisfied if $(N,g)$ were known to be time orientable.}  Denote $n^\sharp \in \mathfrak{X}(\del{D})$ by $\mathfrak{n}$, and the vector field $X|_{\del{D}}$ by $\tilde{X}$.  Then, for any $\Phi \in C^\infty(N)$:
\begin{equation}
[[\bm{1}_D], P]\Phi = (\calS_{\del{D},\mu_g} \circ \rmT_{\del{D}, n} )(\Phi\restriction_{\del{D}}), \label{commutator1} 
\end{equation}
where $\rmT_{\del{D}, n} : C^\infty(\del{D}) \to \Gamma^\infty(\calD[\del{D}])$ is the first-order differential operator 
\begin{equation}\label{defnofTop}
\rmT_{\del{D}, n} \varphi = \left\{2 \mathfrak{n}\varphi + [n(\tilde{X}) + {\div}_{\iota_n \mu_g} \mathfrak{n}]\varphi\right\} \iota_n \mu_g = \left\{2 \mathfrak{n}\varphi + n(\tilde{X}) \varphi\right\} \iota_n \mu_g + \varphi \Lie_\mathfrak{n}\iota_n \mu_g.
\end{equation}
\end{theorem}
\begin{proof}
Let $\chi \in C_\c^\infty(N)$ and denote the restrictions of $\Phi$ and of $\chi$ to $\del{D}$ by $\tilde{\Phi}$ and $\tilde{\chi}$, respectively.  By Corollary \ref{cordivthmdelta} and since the Green vector field here is $\calG = j$ given by Equation (\ref{greenformscalarnormhyp}),
\begin{equation}\label{thmcommidentityprf1}
	\pair{[[\bm{1}_D], P]\Phi}{\chi \mu_g} = \pair{\calS_{\del{D},\mu_g}([\tilde{\chi} \Lie_\mathfrak{n} \tilde{\Phi} - \tilde{\Phi} \Lie_\mathfrak{n}\tilde{\chi} + \tilde{\Phi} \tilde{\chi} n(\tilde{X})] \, \iota_n \mu_g)}{\mu_g}.
\end{equation}
By the Leibniz rule, the argument of $\calS_{\del{D},\mu_g}$ can be rewritten as
\begin{align*}
[\tilde{\chi} \Lie_\mathfrak{n} \tilde{\Phi} - \tilde{\Phi} \Lie_\mathfrak{n}\tilde{\chi} &+ \tilde{\Phi} \tilde{\chi} n(\tilde{X})] \, \iota_n \mu_g \\
&=  [2\tilde{\chi} \Lie_\mathfrak{n} \tilde{\Phi} + \tilde{\Phi} \tilde{\chi} n(\tilde{X})] \, \iota_n \mu_g + \tilde{\chi} \tilde{\Phi} \Lie_\mathfrak{n} \iota_n \mu_g - \Lie_\mathfrak{n} (\tilde{\chi}\tilde{\Phi} \, \iota_n \mu_g)  \\
&= \tilde{\chi}[2 \Lie_\mathfrak{n} \tilde{\Phi} + \tilde{\Phi} n(\tilde{X}) + \tilde{\Phi} \, {\div}_{\iota_n \mu_g} \mathfrak{n}] \, \iota_n \mu_g - \Lie_\mathfrak{n} (\tilde{\chi}\tilde{\Phi} \, \iota_n \mu_g) \\
&= \tilde{\chi} \rmT_{\del{D}, n} \tilde{\Phi} - \Lie_\mathfrak{n} (\tilde{\chi}\tilde{\Phi} \, \iota_n \mu_g).
\end{align*}
By the remarks following Definition \ref{singlayerdistr}, the second term in the last line gives no contribution when the expression is reinserted into Equation (\ref{thmcommidentityprf1}).  Hence,
\begin{equation*}
\pair{[[\bm{1}_D], P]\Phi}{\chi \mu_g} = \pair{\calS_{\del{D},\mu_g}(\tilde{\chi} \rmT_{\del{D}, n} \tilde{\Phi})}{\mu_g} = \pair{\calS_{\del{D},\mu_g}(\rmT_{\del{D}, n} \tilde{\Phi})}{\chi \mu_g},
\end{equation*}
which completes the proof.
\end{proof}

In the setting of Theorem \ref{thmcommidentity}, the right-hand side of Equation (\ref{commutator1}) may be rewritten as the composition of a differential operator on $N$ with multiplication by a delta distribution supported on $\del{D}$, as follows: let $\hat{n}$ denote an extension of $n$ to a smooth covector field on $N$, and let $\Theta \in C^\infty(N)$ denote an extension of ${\div}_{\iota_n \mu_g} \mathfrak{n}$.\footnote{Since $\del{D}$ is properly embedded into $N$, such extensions can always be found.}  Then, with $\calY \doteq 2\hat{n}^\sharp + \left[\Theta + \hat{n}(X)\right]$,
\begin{equation}
[[\bm{1}_D], P]\Phi = \calY \Phi \cdot \delta_{\del{D}, n} \quad \forall \ \Phi \in C^\infty(N). \label{commutator1null}
\end{equation}
This equation will now be massaged further.  By the Leibniz rule, $\calY \Phi \cdot \delta_{\del{D}, n} = \calY(\Phi \cdot \delta_{\del{D}, n}) - 2\Phi \cdot \hat{n}^\sharp \delta_{\del{D}, n}$.  A simple calculation based on Corollary \ref{coradjointvecfield} and on the fact that $\hat{n}^\sharp$ is tangent to $\del{D}$ shows that, for any $\chi \in C_\c^\infty(N)$,
\begin{align*}
\pair{\hat{n}^\sharp \delta_{\del{D}, n}}{\chi \mu_g} &= \pair{\delta_{\del{D}, n}}{[(\hat{n}^\sharp)^\dagger\chi] \mu_g} \\
&= -\pair{\delta_{\del{D}, n}}{[\hat{n}^\sharp\chi] \mu_g} - \pair{(\div{\hat{n}^\sharp}) \cdot \delta_{\del{D},n}}{\chi \mu_g} \\
&= - \int_{\del{D}} [\mathfrak{n}\tilde{\chi}] \iota_n \mu_g - \pair{(\div{\hat{n}^\sharp}) \cdot \delta_{\del{D},n}}{\chi \mu_g} \\
&= \int_{\del{D}} \tilde{\chi} \Lie_\mathfrak{n} \iota_n \mu_g - \pair{(\div{\hat{n}^\sharp}) \cdot \delta_{\del{D},n}}{\chi \mu_g} \\
&= \pair{[\Theta - (\div{\hat{n}^\sharp})]\cdot \delta_{\del{D},n}}{\chi \mu_g}.
\end{align*}
Putting together these considerations, one can immediately record the following result, which will be of pivotal importance in the arguments for uniqueness used in the next subsection.
\begin{corollary}\label{corsecondcommutator} Under the same assumptions as in Theorem \ref{thmcommidentity}, let $\hat{n} \in \Gamma^{\infty}(T^*N)$ be an arbitrary extension of $n$ and $\Theta \in C^\infty(N)$ be an arbitrary extension of ${\div}_{\iota_n \mu_g} \mathfrak{n}$.  Then, defining the first-order differential operator $\calT \doteq 2\hat{n}^\sharp + [ 2 {\div}{(\hat{n}^\sharp)} - \Theta + \hat{n}(X)]$, for any $\Phi \in C^\infty(N)$ one has
\begin{equation} 
[[\bm{1}_D], P] = \calT (\Phi \cdot \delta_{\del{D}, n}). \label{commutator2null}
\end{equation}
In particular,
	\begin{equation}
	P(\Phi \cdot \bm{1}_D) = - \calT (\Phi \cdot \delta_{\del{D}, n}) \label{commidentityhom}
	\end{equation}
	whenever $\Phi \in C^\infty(N) \cap \ker{P}$.
\end{corollary}

\subsection{Uniqueness}\label{uniquenesslindipsec}

The abstract essence of the uniqueness argument for two-sided characteristic Cauchy problems may now be revealed.  Let us assume yet again all the hypotheses of Theorem \ref{existencethmnodistrsol}.  By item \ref{JN} in Proposition \ref{intersectJN}, $N \doteq J(\scrN)$ is an open submanifold of $M$.   There is therefore a continuous restriction map of distributions, $r : \scrD'(M) \to \scrD'(N)$.  It is also clear that, as subsets of $N$ with the relative topology, $J^+(\scrN)$ and $J^-(\scrN)$ are domains with regular boundary---the boundary being equal to $\scrN$ in both cases.  Any global, past-directed, smooth section of the null line bundle of $\scrN$ yields via index lowering a smooth conormal field $n$ along $\scrN$ which is everywhere outward pointing [resp.\ inward pointing] relative to $J^+(\scrN)$ [resp.\ $J^-(\scrN)$].  Let also $\calT$ be the first-order differential operator defined in Equation (\ref{commutator2null}).

Suppose that the following favourable scenario occurs:
\begin{enumerate}[\textbf{(\roman*)}]
\item \label{domainexist} There exist linear subspaces $E_{\bm{1}}$ and $E_{\delta}$ of $\scrD'(N)$, containing $C^\infty(N)$ and such that the operators of multiplication by the indicator functions of $J^+(\scrN)$ and $J^-(\scrN)$, and the operator of multiplication by $\delta_{\scrN,n}$, extend to linear maps
\begin{equation*} \calM_+, \calM_- : E_{\bm{1}} \to \scrD'(N) \quad \text{and} \quad \calM_\delta :  E_{\delta} \to \scrD'(N).
\end{equation*}
\item For any $\Phi \in E_{\bm{1}}$, $\supp{\calM_\pm \Phi} \subseteq \supp{\Phi} \cap J^{\pm}(\scrN)$ and $[\calM_+ + \calM_-]\Phi = \Phi$. \label{multmapsprop}
\item There is a subspace $E$ of $\scrD'(M)$ such that $C^\infty(M) \subseteq E$ and $r(E) \subseteq E_{\bm{1}} \cap E_{\delta}$.  Furthermore, for any $\phi \in E$ with $P \phi = 0$, it holds that
\begin{equation}
P\calM_\pm r{\phi} = \mp \calT \calM_\delta r{\phi} \quad \text{in }\scrD'(N).
 \label{commidentity4}
\end{equation}
 \label{generalisedcommut}
\end{enumerate}
	In the presence of these desiderata, let $\phi \in E$ be such that $P \phi = 0$ and that $\supp{(\calM_+ r \phi)}$ is past compact in $\scrM$.  Henceforth, $\rmG^\dagger_+, \rmG^\dagger_-$ will denote the retarded ($+$) and advanced ($-$) Green operators for $P^\dagger$.  Let now $\zeta \in C_{\c}^\infty(N)$ be arbitrary and notice that, under the assumptions on $\phi$ just given,
	\begin{equation*}
	\supp{(\calM_+ r \phi)} \cap J^-(\supp{\zeta}) \quad \text{is compact in $M$ and therefore also in $N$.}
\end{equation*}
Let $\theta \in C_{\c}^\infty(N)$ be equal to $1$ in a neighbourhood of this set.  Since $P^\dagger \rmG^\dagger_- \zeta = \zeta$ and $\supp{(\rmG^\dagger_-\zeta)} \subseteq J^-(\supp{\zeta})$, it is easy to see that $P^\dagger (\theta \rmG^\dagger_- \zeta) - \zeta$ has (compact) support disjoint from $\supp{(\calM_+ r \phi)}$.  Whereupon, using Equation (\ref{commidentity4}), one obtains
	\begin{align*}
	[\calM_+ r \phi] (\zeta \mu_g) &= [\calM_+ r \phi]\Big(P^\dagger (\theta \rmG^\dagger_- \zeta) \, \mu_g\Big) \\
	&= [P \calM_+ r \phi]\left(\theta \rmG^\dagger_- \zeta \, \mu_g\right) \\
	&= - [\calT \calM_\delta r \phi]\left(\theta \rmG^\dagger_- \zeta \, \mu_g\right).
	\end{align*}
It follows that $\calM_+ r \phi = 0$ if $\calM_\delta r \phi = 0$.  Similar arguments show that if $\phi \in E$, $P\phi = 0$, and $\supp{(\calM_- r \phi)}$ is future compact in $\scrM$, then $\calM_\delta r \phi = 0$ implies that $\calM_- r \phi = 0$.  In particular, if both $\supp{(\calM_+ r \phi)}$ is past compact and $\supp{(\calM_- r \phi)}$ is future compact then, by property \ref{multmapsprop}, $r{\phi} = 0$ whenever $\phi \in E$, $P\phi = 0 $ and $\calM_\delta r \phi = 0$.
	
	The relation with characteristic Cauchy problems is as follows: when $\phi \in C^\infty(M)$, Equation (\ref{commidentity4}) reduces to Equation (\ref{commidentityhom}) with $N = J(\scrN)$, the domain $D=J^\pm(\scrN)$, and $\Phi \doteq r \phi$.  In particular, in that case $\calM_\delta r \phi = 0$ is equivalent to the statement that $\phi$ vanishes on $\scrN$, and $\calM_\pm r \phi = 0$ if and only if $\phi$ itself vanishes on $J^\pm(\scrN)$.  The following proposition summarises what has been shown so far for future reference.  

\begin{proposition}\label{uniquenessabstract}
Let $\scrM = (M,g,\mathfrak{t})$, $\scrN$ and $P$ be as in Theorem \ref{existencethmnodistrsol}.  With $N = J(\scrN)$, let there be linear subspaces $E_{\bm{1}}, E_\delta$ of $\scrD'(N)$, $E$ of $\scrD'(M)$,  and generalised multiplication maps $\calM_\pm, \calM_\delta$ with the properties listed in \emph{\ref{domainexist}}, \emph{\ref{multmapsprop}} and \emph{\ref{generalisedcommut}}.   Then, for any $\phi \in E$ with $P\phi = 0$ and $\calM_\delta r \phi = 0$:
\begin{enumerate}
\item $\calM_+ r \phi = 0$ whenever $\supp{(\calM_+ r \phi)}$ is past compact in $\scrM$;
\item $\calM_- r \phi = 0$ whenever $\supp{(\calM_- r \phi)}$ is future compact in $\scrM$;
\item $r \phi = 0$ whenever $\supp{(\calM_+ r \phi)}$ is past compact in $\scrM$ and $\supp{(\calM_- r \phi)}$ is future compact in $\scrM$.
\end{enumerate}
In particular, if $\phi \in C^\infty(M)$ solves $P\phi = 0$, vanishes on $\scrN$, and has $\supp{(\calM_+ r \phi)}$ past compact [resp.\ $\supp{(\calM_- r \phi)}$ future compact] in $\scrM$, then $\phi = 0$ on $J^+(\scrN)$ [resp.\ on $J^-(\scrN)$]. \qed
\end{proposition}	
	
	Proposition \ref{uniquenessabstract} does not yet provide a concrete uniqueness result in low regularity.  It simply reduces the problem of finding such a result to the problem of finding spaces $E_{\bm{1}}, E_\delta, E$ and generalised multiplication maps $\calM_+$, $\calM_-$ and $\calM_\delta$ as per \ref{domainexist}, \ref{multmapsprop} and \ref{generalisedcommut}---with $E$ a strict enlargement of $C^\infty(M)$.  This will be done in what follows.
	
\subsubsection*{Uniqueness in the concrete}	
	
	Henceforth, the view will be taken that the product between distributions on manifolds should be intended (whenever it exists) in the ``Fourier--Ambrose'' sense \cite{ambrose1980products, oberguggenberger1986products}.  In this framework, it is not difficult to use known results to suggest candidate subspaces $E_{\bm{1}}, E_\delta$ satisfying the desiderata \ref{domainexist} and \ref{multmapsprop} in the previous section, with $\calM_+$, $\calM_-$ and $\calM_\delta$ being the products with the indicator function of $J^+(\scrN)$, the indicator function of $J^-(\scrN)$, and $\delta_{\scrN, n}$ (respectively) in the Fourier--Ambrose sense.  The notion of \emph{Sobolev} (or $H^s$) \emph{wave front set} $WF^s(u)$ of a distribution $u$, introduced in \cite[p.\ 201]{duistermaat1972fourier}, will be used, and the reader is referred to \cite[App.\ B]{junker2002adiabatic} for a review.
	
	\begin{proposition}[Sobolev wave front set criterion for existence of distribution products, Cor.\ 3.1 in \cite{oberguggenberger1986products}]\label{sobwfsetcrit} Let $N$ be a manifold and $u,v \in \scrD'(N)$.  The product $u \cdot v$ exists in $\scrD'(N)$ and in the Fourier--Ambrose sense if, for every $(x, \xi) \in T^*N {\setminus } \bm{0}$, there exist $s, t \in \R$ with $s+t \geq 0$, $(x, \xi) \notin WF^s(u)$ and $(x, \xi) \notin WF^t(v)$. \qed
	\end{proposition}
	
	\begin{lemma}\label{lemextension}
	Let $\scrM = (M,g,\mathfrak{t})$, $\scrN$ and $P$ be as in Theorem \ref{existencethmnodistrsol}.  Set $N = J(\scrN)$.  Then
	\begin{align*}\label{wfcondn} 
	E_{\bm{1}} &\doteq \Bigg\{ \Phi \in \scrD'(N) \ \Bigg| \bigcap_{s>-1/2} WF^s(\Phi) \cap N^*\scrN = \emptyset \Bigg\} \quad \text{and} \\
	E_{\delta} &\doteq \Bigg\{ \Phi \in \scrD'(N) \ \Bigg| \bigcap_{s>1/2} WF^s(\Phi) \cap N^*\scrN = \emptyset \Bigg\} \subset E_{\bm{1}}
	\end{align*}
	satisfy \emph{\ref{domainexist}} and \emph{\ref{multmapsprop}} on p.\ \pageref{domainexist}, $\calM_+$, $\calM_-$ and $\calM_\delta$ being operators of multiplication, in the Fourier--Ambrose sense, with the indicator function of $J^+(\scrN)$, the indicator function of $J^-(\scrN)$, and $\delta_{\scrN, n}$ (respectively).  \end{lemma}
	\begin{proof}
	If $N$ is a manifold and $D \subseteq N$ is a domain with regular boundary,
	\begin{equation*}
	WF^s(\bm{1}_D) = \begin{cases} \emptyset &\text{if } s < 1/2 \\ N^*\del{D} {\setminus } \bm{0} &\text{if } s \geq 1/2 \end{cases}
	\end{equation*}
	\cite[Thm.\ 3.4.1]{agranovich2015sobolev} and, if $n$ is a smooth field of conormals to $\del{D}$,
	\begin{equation*}
	WF^s(\delta_{\del{D},n}) = \begin{cases} \emptyset &\text{if } s < -1/2 \\ N^*\del{D} {\setminus } \bm{0} &\text{if } s \geq -1/2 \end{cases}
	\end{equation*}
	\cite[Eq.\ (106)]{junker2002adiabatic}.  The lemma then follows at once using Proposition \ref{sobwfsetcrit}. 
	\end{proof}
	
For any manifold $N$ and $s \in \R$, let
\begin{equation*}
H^{s+}_\loc(N) \doteq \bigcup_{\ell > s} H^{\ell}_\loc(N) \quad \text{and} \quad H^{s-}_\loc(N) \doteq \bigcap_{\ell < s} H^{\ell}_\loc(N).
\end{equation*}
	By the result just proved, for any $\phi \in H^{1/2+}_\loc(M)$ it clearly holds that $r{\phi} \in E_{\bm{1}} \cap E_{\delta}$.  A concrete uniqueness theorem can finally be given, the idea being that a density/continuity argument shows that $E \doteq H^{1/2+}_\loc(M)$ satisfies desideratum \ref{generalisedcommut} on p.\ \pageref{domainexist}.  As preparation, let us recall that, for any $s \in \R$, the extended causal Green operator $\rmG : \scrD'_{\tc}(M) \to \scrD'(M)$ maps $H^s_\c(M)$ to $H^{s+1}_\loc(M)$---indeed, it does so continuously in the standard topologies of these spaces \cite[Thm.\ 6.5.3]{duistermaat1972fourier}.  By a simple exhaustion argument (omitted), $\rmG$ then also maps $H^s_\tc(M)$ to $H^{s+1}_\loc(M)$.  In particular,
\begin{equation}\label{HslockerP}
H^{s}_\loc(M) \cap \ker P = \rmG H^{s-1}_\tc(M)
\end{equation}
as can be seen using a standard procedure.\footnote{\label{footnoteinverseop}Let $\scrC_1, \scrC_2$ be spacelike Cauchy surfaces for $\scrM$ with $\scrC_2 \subset I^+(\scrC_1)$.  Let $\rho_+$ be any smooth function on $M$ which equals $0$ on $J^-(\scrC_1)$ and $1$ on $J^+(\scrC_2)$.  If $\phi \in H^{s}_\loc(M)$ and $P \phi = 0$ then $\eta \doteq P(\rho_+ \phi)$ is in $H^{s-1}_{\tc}(M)$ by the Leibniz rule, and $\rmG \eta = \phi$.}
	
\begin{theorem}[Uniqueness in $H^{1/2+}_\loc$]\label{uniquenessconcrete}
Let $\scrM = (M,g,\mathfrak{t})$, $\scrN$ and $P$ be as in Theorem \ref{existencethmnodistrsol}.  Furthermore, let $r : \scrD'(M) \to \scrD'(J(\scrN))$ be the restriction map and $\calM_\pm$ be as in Lemma \ref{lemextension}.  If $\phi$ belongs to $H^{1/2+}_\loc(M) \cap \ker{P}$ and has vanishing trace on $\scrN$, $\supp{(\calM_+ r \phi)}$ past compact in $\scrM$ and $\supp{(\calM_- r \phi)}$ future compact in $\scrM$, then $\phi = 0$ on $J(\scrN)$.
\end{theorem}
\begin{proof}
Denote $J(\scrN)$ by $N$.  Since $\phi$ is in $H^{1/2+}_\loc(M) \cap \ker{P}$, by virtue of Equation (\ref{HslockerP}) there exists an $\ell > -1/2$ and $\eta \in H^{\ell}_{\tc}(M)$ such that $\rmG \eta = \phi$.  It needs to be shown that 
\begin{equation}\label{eqn}
[P \calM_\pm \pm \calT \calM_\delta]r \phi = [P \calM_\pm \pm \calT \calM_\delta]r \rmG \eta = 0
\end{equation}
[the operator $\calT$ is as defined in Equation (\ref{commutator2null}), while $\calM_\delta$ is as in Lemma \ref{lemextension}], for then Proposition \ref{uniquenessabstract} yields the desired result.  The following density argument shows that the claim holds true in the special case in which $\eta$ has compact support (equivalently, in which $\phi$ has spatially compact support):  $r \circ \rmG : H^\ell_c(M) \to H^{\ell+1}_\loc(N)$ is continuous by the remarks preceding the statement of this theorem.  $\delta_{\scrN,n}$ belongs to $H^{-1/2-}_\loc(N)$, and the indicator functions of $J^+(\scrN)$ and $J^-(\scrN)$ belong to $H^{1/2-}_\loc(N)$.  Hence, Thm.\ 8.3.1 in \cite{hormander1997lectures} shows that there exist $k, k' \in \R$ such that
\begin{equation*}
\calM_\delta : H_\loc^\ell(N) \to H^{k}_\loc(N) \quad \text{and} \quad \calM_\pm : H_\loc^\ell(N) \to H^{k'}_{\loc}(N)
\end{equation*} 
continuously.  Since $P$ and $\calT$ are differential operators, it follows that both $P \circ \calM_\pm \circ r \circ \rm G$ and $\calM_\delta \circ \calT \circ r \circ \rmG$ are continuous from $H^\ell_\c(M)$ to $\scrD'(N)$.  Finally, Equation (\ref{eqn}) can be shown to hold by considering a sequence $(\eta_n)_{n \in \N}$ of functions in $C_\c^\infty(M)$ tending to $\eta$ in the topology of $H_{\c}^\ell(M)$; such a sequence may be found since $C_\c^\infty(M)$ is dense in $H_{\c}^\ell(M)$.  In other words, Equation (\ref{eqn}) certainly holds when $\phi \in H^{1/2+}_\sc(M)$. \\
The general case follows by a reduction to the particular case just described.  Namely, consider the pairing of the left-hand side of Equation (\ref{eqn}) with a test density $\nu \in \Secs^\infty_\c(\calD N)$.  If $\hat{\phi} \in \scrD'(M)$ is any distribution such that $\hat{\phi}=\phi$ on a relatively compact open neighbourhood $W$ of $\supp{\nu}$, then
\begin{equation}\label{eqn2}
\pair{[P \calM_\pm \pm \calT \calM_\delta]r \hat{\phi}}{\nu} = \pair{[P \calM_\pm \pm \calT \calM_\delta]r \phi}{\nu}
\end{equation}
since all operations involved are local.  It will now be shown that such a $\hat{\phi}$ may be found which additionally lies in $H^{1/2+}_\sc(M)$; this will complete the proof since the left-hand side of Equation (\ref{eqn2}) was already proved to vanish in this case.  Let $\scrC_1, \scrC_2$ be spacelike Cauchy surfaces for $\scrM$ with $\scrC_2 \subset I^+(\scrC_1)$ and $\overline{W} \subseteq I^+(\scrC_1) \cap I^-(\scrC_2)$.  Let $\rho_+$ be any smooth function on $M$ which equals $0$ on $J^-(\scrC_1)$ and $1$ on $J^+(\scrC_2)$, and $\sigma \in C^\infty_\c (M)$ be equal to $1$ on the compact set $J(\overline{W}) \cap J^+(\scrC_1) \cap J^-(\scrC_2)$.  Then $\hat{\eta} \doteq \sigma P(\rho_+ \phi)$ belongs to $H^{-1/2+}_\c(M)$ and, if $\hat{\phi} \doteq \rmG \hat{\eta}$ (see Footnote \ref{footnoteinverseop}),
\begin{equation*}
\hat{\phi} - \phi = \rmG[(\sigma-1)P(\rho_+ \phi)].
\end{equation*}
Now, by construction $\supp[(\sigma-1)P(\rho_+ \phi)] \subseteq J^+(\scrC_1) \cap J^-(\scrC_2) \cap \overline{M {\setminus } J(\overline{W})}$ and thus
\begin{multline*}
\supp{[\hat{\phi} - \phi]} \cap W \subseteq J\big( \supp{[(\sigma - 1)P(\rho_+\phi)]} \big) \cap W \subseteq J\Big( \overline{M {\setminus } J(\overline{W})} \Big) \cap W \\ \subseteq \overline{I\big(M {\setminus } \overline{I(W)}\big)} \cap W \subseteq \overline{I(M {\setminus } I(W))} \cap W = \emptyset.\end{multline*}
That is, $\hat{\phi} - \phi = 0$ on $W$ and the proof is complete.
\end{proof}

\section{Representation formulae and continuous dependence}\label{repformulaeunivsec}

Theorem \ref{uniquenessconcrete} implies that the solution map
\begin{equation}
\Sol_\scrN : C^\infty_\tc(\scrN) \oplus C^\infty_\scrN(\scrM) \to C^0_\scrN(\scrM) \quad \text{given by} \quad (f,F) \mapsto \phi_{(f,F)}, \label{solnmapinhom}
\end{equation}
constructed in the proof of Theorem \ref{existencethmnodistrsol} via a highly non-unique extension procedure using Borel's lemma, is actually independent of any arbitrary choices made in that procedure.  Uniqueness also implies straightforwardly that this map is linear.

An alternative way of representing solutions will now be presented.  It is analogous to the expressions used in the proof of Thm.\ 22 in \cite{baer2015initial}.

\begin{proposition}\label{corprotrepformula}
With the hypotheses and notation of Theorem \ref{existencethmnodistrsol} and of Definition \ref{defncharacteristicdata}, the solution map in Equation (\ref{solnmapinhom}) is given by
\begin{equation}\label{repformulainhom}
\Sol_\scrN(f,F) = e(f) - \sum_\pm \rmG_\pm [\bm{1}^\pm]\big( P (e(f)) - F\big) \quad \forall \ (f,F) \in C^\infty_\tc(\scrN) \oplus C^\infty_\scrN(\scrM),
\end{equation}
where $e : C^\infty_\tc(\scrN) \to C_\tc^\infty(M)$ is any (not necessarily linear) extension map such that $\supp{e(f)} \subseteq J(\supp{f})$.\footnote{Such a map can certainly be constructed, e.g.\ by using Borel's lemma with a globally timelike vector field.}  In particular,
\begin{equation}\label{protrepformula}
\Sol^{\mathrm{h}}_\scrN \doteq \Sol_\scrN(\cdot, 0) = - \bigg(\sum_\pm \rmG_\pm [[\bm{1}^\pm], P] \bigg) \circ e.
\end{equation}
\end{proposition}
\begin{proof}
Let $\hat{f} \doteq e(f)$.  Then, under the given hypotheses,
\begin{equation*}
\supp{\big[\bm{1}^\pm (P \hat{f} - F) \big]} \subseteq \supp{(\bm{1}^\pm P \hat{f})} \cup \supp{(\bm{1}^\pm F)} \subseteq  J^\pm(\supp{f}) \cup \supp{(\bm{1}^\pm F)}.
\end{equation*}
That is, $\bm{1}^+ (P \hat{f} - F) \in L^2_{\pc}(M)$ and $\bm{1}^- (P \hat{f} - F) \in L^2_{\fc}(M)$, whence $\rmG_+ [\bm{1}^+ (P\hat{f} - F)] \in H^1_{\pc}(M)$ and $\rmG_- [\bm{1}^- (P\hat{f} - F)] \in H^1_{\fc}(M)$.  Note that the trace of the latter two distributions vanishes on $\scrN$ owing to the support properties of $\rmG_+$ and $\rmG_-$.  Since $\rmG_\pm$ is a right inverse of $P$,
\begin{equation*}
P\bigg\{\hat{f} - \sum_\pm \rmG_\pm [\bm{1}^\pm]\big(P \hat{f} - F \big)\bigg\} = P\hat{f} - \sum_\pm \bm{1}^\pm \big( P\hat{f} - F \big) = F.
\end{equation*}
Consider now the difference
\begin{equation*}
\phi \doteq \Sol_\scrN(f,F) - \bigg[\hat{f} - \sum_\pm \rmG_\pm [\bm{1}^\pm]\big(P \hat{f} - F \big) \bigg].
\end{equation*}
Clearly, $\phi \in H^1_\loc(M)$, $\supp{\phi} \subseteq J(\scrN)$, and $\phi$ satisfies all the hypotheses of Theorem \ref{uniquenessconcrete}.  Hence, $\phi = 0$ on $M$, proving Equation (\ref{repformulainhom}).  Equation (\ref{protrepformula}) is the result of the following calculation:
\begin{align*}
\bigg( \id - \sum_\pm \rmG_\pm [\bm{1}^\pm]P \bigg) \circ e &= \bigg( \id - \sum_\pm \rmG_\pm \left\{ [[\bm{1}^\pm], P] + P [\bm{1}^\pm] \right\} \bigg) \circ e \\
&= \bigg( \id - \sum_\pm \rmG_\pm \left\{ [[\bm{1}^\pm], P] + P [\bm{1}^\pm] \right\} \bigg) \circ e \\
&=  \bigg( \id - \sum_\pm \rmG_\pm [[\bm{1}^\pm], P] - \sum_\pm \rmG_\pm P[\bm{1}^\pm] \bigg) \circ e \\
&= - \bigg(\sum_\pm \rmG_\pm [[\bm{1}^\pm], P] \bigg) \circ e.
\end{align*}
\end{proof}

The extension map $e$ in Proposition \ref{corprotrepformula} may certainly be chosen so that it is continuous from $C_\c^\infty(\scrN)$ to $C_\c^\infty(J(\scrN))$.\footnote{This can be done, for instance, by defining $e$ to be the extension map constructed in the proof of Theorem \ref{globalborel}, with $f_n \equiv 0$ for each $n \geq 1$.} For any $\varepsilon > 0$, $[\bm{1}^\pm]$ maps $C_\c^\infty(J(\scrN))$ continuously to $H^{1/2-\varepsilon}_\c(J(\scrN))$, and $\rmG_\pm$ maps $H^{1/2-\varepsilon}_\c(M)$ continuously to $H^{3/2-\varepsilon}_\loc(M)$.  Equation (\ref{repformulainhom}) then immediately implies the following statement on the continuous dependence of solutions of two-sided characteristic Cauchy problems.

\begin{corollary}\label{corcontdependence}
The map $\Sol_\scrN$ in Proposition \ref{corprotrepformula} is continuous from $C_\c^\infty(\scrN) \oplus C_\c^\infty(J(\scrN))$ to $H^{3/2-\varepsilon}_\loc(M)$, for any $\varepsilon > 0$. \qed
\end{corollary}

\subsection{Further representation formulae for the homogeneous problem}\label{repformulasubsec}

The causal Green operator may be made to appear in Equation (\ref{protrepformula}) as follows:
\begin{equation}\label{repformulawithGandcomm}
\Sol^{\mathrm{h}}_\scrN = - (-\rmG_+ + \rmG_-) [[\bm{1}^-], P] \circ e = \rmG [[\bm{1}^-], P] \circ e \quad \big( = \rmG [\bm{1}^-] P \circ e\big).
\end{equation}
(The equality in parentheses follows from the fact that $\rmG \circ P = 0$.)  Now let $n$ be a smooth conormal field along $\scrN$ which is outward-directed relative to $J^-(\scrN)$---and note that with the signature convention adopted in this paper this implies that $n^\sharp$ is future-directed.  The jump formula in Theorem \ref{thmcommidentity} may be applied to Equation (\ref{repformulawithGandcomm}), yielding
 \begin{equation}
\Sol^{\mathrm{h}}_\scrN = \rmG \circ \calS_{\scrN, \mu_g} \circ \rmT_{\scrN, n}\restriction_{C^\infty_\tc(\scrN)} \label{repformulafinal}
\end{equation}
with the operator $\rmT_{\scrN, n}$ defined as in Equation (\ref{defnofTop}).  This is the final representation formula obtained in this paper.

\section{Expansion density of a null hypersurface and ``universality'' in the two-sided characteristic Cauchy problem}\label{expdensunivsec}

\subsection{Expansion density of a null hypersurface}\label{expdensitysubsec}

The differential operator $\rmT_{\scrN, n}$ appearing in Equation (\ref{repformulafinal}) has been ostensibly defined, according to Equation (\ref{defnofTop}), via a non-unique choice of conormal field $n$ along $\scrN$, outward directed relative to $J^-(\scrN)$.  However, no such choice appears in the commutator $[[\bm{1}^-],P]$ which, by virtue of Equation (\ref{commutator1}), equals $\calS_{\scrN, \mu_g} \circ \rmT_{\scrN, n} \circ i^*$ where $i : \scrN \to M$ is the inclusion.  Since, for any $\varphi$, $\{ 2 \mathfrak{n}\varphi + n(\tilde{X})\varphi \} \iota_n\mu_g$ is easily seen to be invariant under changes in $n$, one may regard Equation (\ref{commutator1}) as evidence that the quantity $\Lie_{\mathfrak{n}} (\iota_n \mu_g)$ is also invariantly defined.  It is, however, of interest to provide a direct proof of this fact in the general setting of null hypersurfaces and without making reference to regular domains or normally hyperbolic operators.  The following property of Lie derivatives of densities will be used: for any continuous vector field $X$, $C^1$ function $\alpha$ and $C^1$ density $\mu$,
\begin{equation}\label{liederivdensitiesproperty} \Lie_{\alpha X} \mu = \alpha \Lie_X \mu +  (X\alpha) \mu. \end{equation}

\begin{theorem}\label{invarianceofexpansion}
Let $\scrN \xhookrightarrow{i} M$ be a smooth null hypersurface in a Lorentzian manifold $(M,g)$.  If $N^* \scrN \to \scrN$ is trivial and $m$, $n$ are any two smooth and nowhere-vanishing sections of $N^* \scrN$ then
\begin{equation}
	\Lie_{\mathfrak{m}} (\iota_m \mu_g) = \pm \Lie_{\mathfrak{n}} (\iota_n \mu_g),
\end{equation}
where $\mathfrak{m}, \mathfrak{n}$ are the vector fields on $\scrN$ given by raising indices on $m, n$ (respectively), and  the plus (resp.\ minus) sign occurs when $m$ and $n$ are pointwise positive (resp.\ negative) multiples of one other.
\end{theorem}
\begin{proof}
	Let $\alpha \in C^\infty(\scrN; \R {\setminus } \{ 0 \})$ be the function uniquely defined by $n = \alpha m$.  Recall [see Equation (\ref{intproductconormal})] that $\iota_n \mu_g = \iota_\Theta \mu_g$ for any $M$-vector field $\Theta$ along $\scrN$, transverse to $\scrN$ and such that $n(\Theta) = 1$.  It follows that $m(\alpha \Theta) = 1$ and that $\iota_m \mu_g = \iota_{\alpha \Theta} \mu_g = |\alpha| \, \iota_\Theta \mu_g = |\alpha| \,\iota_n \mu_g$.  Upon using the Leibniz rule for densities,
	\begin{align*} \Lie_{\mathfrak{m}} (\iota_m \mu_g) = \Lie_{\alpha^{-1} \mathfrak{n}} (|\alpha| \iota_n \mu_g) &=  \big(\Lie_{\alpha^{-1} \mathfrak{n}} |\alpha|\big) \iota_n \mu_g +  |\alpha| \Lie_{\alpha^{-1} \mathfrak{n}} (\iota_n \mu_g) \\ &=  (|\alpha|^{-1} \mathfrak{n} \alpha) \iota_n \mu_g + |\alpha| \Lie_{\alpha^{-1} \mathfrak{n}} (\iota_n \mu_g). \end{align*}
	Using Equation (\ref{liederivdensitiesproperty}), one may further calculate
	\begin{multline*}
	\Lie_{\mathfrak{m}} (\iota_m \mu_g) = (|\alpha|^{-1} \mathfrak{n}\alpha) \iota_n \mu_g +  |\alpha| \alpha^{-1} \Lie_{\mathfrak{n}} (\iota_n \mu_g) + |\alpha| \mathfrak{n}(\alpha^{-1}) \iota_n \mu_g \\
	= |\alpha| \alpha^{-1} \Lie_{\mathfrak{n}} (\iota_n \mu_g) + n(\alpha^{-1} |\alpha|) \, \iota_n \mu_g = |\alpha| \alpha^{-1} \Lie_{\mathfrak{n}} (\iota_n \mu_g),
	\end{multline*}
	which completes the proof.
\end{proof}

\begin{definition}\label{defnexpansiondensity}
Let $\scrN$ be a smooth null hypersurface in a Lorentzian manifold $(M,g)$.  Assume $N^*\scrN \to \scrN$ is trivial---i.e.\ orientable as a bundle---and that a choice of orientation class is made for it.\footnote{Of course (by Proposition \ref{nullhypglobalnullfield}) $N^*\scrN$ is automatically trivial whenever $(M,g)$ is time orientable.}  Then the \emph{expansion density} of $\scrN$ relative to this orientation is the smooth density on $\scrN$ given by $\Lie_{n^\sharp} (\iota_n \mu_g)$ for one (and hence any other) nowhere-vanishing and positively oriented section $n$ of $N^* \scrN$.  $\scrN$ will be said to be \emph{divergence-free} (relative to $g$) if its expansion density vanishes identically.
\end{definition}

The behaviour of the expansion density under conformal transformations of the ambient Lorentzian metric will now be examined.  Note that a generic smooth hypersurface $\scrN$ which is null in $\scrM = (M,g,\mathfrak{t})$ is also null in $(M,g',\mathfrak{t})$ if, on $\scrN$, $g'$ is a positive multiple of $g$.  Furthermore, the images of the null generators of $\scrN$ relative to $g$ and $g'$ coincide in this case.

\begin{theorem}\label{bmslikethm} Let $(M,g)$ be a Lorentzian manifold of dimension $d+1$ and $\scrN$ be a smooth null hypersurface with trivial conormal bundle.  Let $g'$ be another Lorentzian metric such that, on $\scrN$, $g' = \lambda g$ for some smooth positive function $\lambda$ on $\scrN$.  Then
\begin{equation*}
\Lie_{n^{\sharp'}} (\iota_n \mu_{g'}) = \Lie_{\lambda^{\frac{d-1}{2}} n^\sharp} (\iota_n \mu_{g}),
\end{equation*}
	where ${}^{\sharp'}$ denotes the musical isomorphism determined by $g'$.  If, in addition, $\lambda$ is constant along the null generators of $\scrN$, then the expansion densities of $\scrN$ relative to the two metrics are related by
\begin{equation}\label{liederivconformaltransf}
\Lie_{n^{\sharp'}} (\iota_n \mu_{g'}) = \lambda^{\frac{d-1}{2}}\Lie_{n^\sharp} (\iota_n \mu_{g}),
\end{equation}
and $\scrN$ is divergence-free relative to $g$ if and only if it is so relative to $g'$.
\end{theorem}
\begin{proof}
On $\scrN$, $n^{\sharp'} = \lambda^{-1} n^\sharp$ and $\mu_{g'} = \lambda^{\frac{d+1}{2}} \mu_g$.  Therefore,
\begin{multline*}
\Lie_{n^{\sharp'}} (\iota_n \mu_{g'}) = \Lie_{\lambda^{-1} n^\sharp} \big( \lambda^{\frac{d+1}{2}} \, \iota_n \mu_{g}\big)
= \lambda^{-1} n^\sharp\big[\lambda^{\frac{d+1}{2}}\big] \, \iota_n \mu_{g} +  \lambda^{\frac{d+1}{2}} \, \Lie_{\lambda^{-1} n^\sharp} (\iota_n \mu_{g}) \\
= \lambda^{-1} n^\sharp\big[\lambda^{\frac{d+1}{2}}\big] \, \iota_n \mu_{g} +  \lambda^{\frac{d+1}{2}} n^\sharp[\lambda^{-1}] \, \iota_n \mu_{g} +  \lambda^{\frac{d+1}{2}} \lambda^{-1} \Lie_{n^\sharp} (\iota_n \mu_{g}) \\
= n^\sharp \big[\lambda^{\frac{d - 1}{2}}\big] \, \iota_n \mu_{g} +  \lambda^{\frac{d-1}{2}}\Lie_{n^\sharp} (\iota_n \mu_{g}) = \Lie_{\lambda^{\frac{d-1}{2}} n^\sharp} (\iota_n \mu_{g}).
\end{multline*}
So far, the assumption that $n^\sharp\lambda = 0$ has not been used.  Doing so (in the second-to-last line above) yields Equation (\ref{liederivconformaltransf}).
\end{proof}

\subsection{``Universality'' in the two-sided characteristic Cauchy problem}\label{universalitysubsec}

In the setting of Definition \ref{defnexpansiondensity}, for any vector field $X$ one may canonically define an operator $\rmT_{\scrN} : C^\infty(\scrN) \to \Gamma^{\infty}(\calD\scrN)$ by $\rmT_{\scrN} = \rmT_{\scrN, n}$ [see Equation (\ref{defnofTop})] for any positively oriented $n$.  The dependence of $\rmT_{\scrN, g} \doteq \rmT_{\scrN}$ on the Lorentzian metric $g$ will now be considered.

\begin{theorem}\label{equivariancethm}
	Let $(M,g)$ have dimension $1+d$, and let $g'$ be another metric such that, on $\scrN$, $g' = \lambda g$ for some $\lambda \in C^\infty(\scrN)$ positive and constant along the generators of $\scrN$.  Then, if in addition the vector field $X$ is tangent to $\scrN$,
	\begin{equation*}
	\calS_{\scrN, \mu_{g'}} \circ \rmT_{\scrN, g'} = \calS_{\scrN, \mu_g} \circ \rmT_{\scrN, g} \circ [\lambda^{-1}].
	\end{equation*}
\end{theorem}
\begin{proof}
	Since $\mu_{g'} = \lambda^{\frac{d+1}{2}} \mu_g$ on $\scrN$, $\calS_{\scrN, \mu_{g'}} = \calS_{\scrN, \mu_g} \circ [\lambda^{-\frac{d+1}{2}}]$.  Hence, it suffices to show that $\rmT_{\scrN, g'} = [\lambda^{\frac{d+1}{2}}] \circ \rmT_{\scrN,n} \circ [\lambda^{-1}]$.  The second term in curly brackets in the right-hand side of Equation (\ref{defnofTop}) vanishes under the tangency assumption on $X$.  Since $n^\sharp \lambda = 0$, $n^{\sharp'} = \lambda^{-1} n^\sharp = n^\sharp \circ [\lambda^{-1}]$ by the Leibniz rule.  Hence, using Equation (\ref{liederivconformaltransf}),
	\begin{equation*}
	\rmT_{\scrN, g'} \varphi = 2 \big[n^{\sharp'}\varphi\big]\iota_n \mu_{g'} + \varphi \Lie_{n^{\sharp'}} \iota_n \mu_{g'} = \lambda^{\frac{d+1}{2}} [2 n^\sharp(\lambda^{-1}\varphi) \iota_n \mu_g + \lambda^{-1} \varphi \Lie_{n^\sharp} \iota_n \mu_g ]
	\end{equation*}
	as had to be shown.
\end{proof}

Let $M$ be a manifold and $\scrN \xhookrightarrow{i} M$ be a smooth hypersurface such that there exists a smooth one-dimensional and integrable distribution $Z$ in $T \scrN$, globally trivial as a bundle over $\scrN$.  Let $\mathbf{Lor}$ be the bundle of Lorentzian metrics on $M$, and let $\mathrm{Met}_{\scrN, Z}$ be the set consisting of those smooth sections of the pullback bundle $i^*\mathbf{Lor}$ according to which $\scrN$ is null and $Z$ is precisely the null line bundle of $\scrN$.  Then there is a multiplicative action on $\mathrm{Met}_{\scrN, Z}$ by the abelian group of positive smooth functions on $\scrN$, and another by the subgroup $G_{\scrN, Z}$ of positive smooth functions on $\scrN$ which are constant along the integral manifolds of $Z$.  There is also a representation $\rho : G_{\scrN, Z} \to V$ where $V = \mathrm{Lin}(C^\infty(\scrN), \scrD'(M))$, given by precomposition as follows:
\begin{equation*}\rho(\lambda) v \doteq v \circ [\lambda] \quad \forall \ v \in V, \ \forall \ \lambda \in G_{\scrN, Z}.\end{equation*}
Theorem \ref{equivariancethm} reveals that the function $\theta : \mathrm{Met}_{\scrN, Z} \to V$ given by $g \mapsto \calS_{\scrN, \mu_g} \circ \rmT_{\scrN, g}$ is equivariant in the sense that $\theta(\lambda g) = \rho(\lambda^{-1})\theta(g)$.  Let $\mathrm{Met}_{\scrN, Z} \times_{\rho} V$ denote the set of all orbits of pairs $(g, v) \in \mathrm{Met}_{\scrN, Z} \times V$ under the $G_{\scrN, Z}$-action given by
\begin{equation*} (g, v) \mapsto (\lambda g, \rho(\lambda^{-1}) v) \quad \lambda \in G_{\scrN, Z},\end{equation*}
and notice that there is a natural projection $\mathrm{Met}_{\scrN, Z} \times_{\rho} V \to \mathrm{Met}_{\scrN, Z}/G_{\scrN, Z}$ given by sending the equivalence class of $(g', v')$ to the equivalence class of $g'$ in $\mathrm{Met}_{\scrN, Z}/G_{\scrN, Z}$.  By the equivariance just discussed, $\theta$ corresponds uniquely to a section $\Theta : \mathrm{Met}_{\scrN, Z}/G_{\scrN, Z} \to \mathrm{Met}_{\scrN, Z} \times_{\rho} V$, namely
\begin{equation*}
\Theta : [g] \mapsto [(g, \theta(g))].
\end{equation*}
$\theta$ can be uniquely retrieved from $\Theta$ by letting $\theta(g)$ be the unique element $v$ of $V$ such that $[(g, v)] = \Theta(g)$.  One might therefore say that $\Theta$ encodes a ``universal structure'' underlying the characteristic Cauchy problem in the case in which the vector field $X$ in $P = \square_g + X + [q]$ is tangent to the initial data hypersurface.

\section{Final remarks and applications}\label{applicationssec}

\subsection{Relation with the existing literature}\label{exlit}
The line of attack on the characteristic Cauchy problem pursued in this paper was mostly inspired by the works \cite{rendall1990reduction} and \cite{baer2015initial}.  In fact, as already discussed in Section \ref{mainres}, the main existence result has been effectively (although not explicitly in the details of the proofs) obtained by solving two (generalised) Goursat problems using the approach in \cite{rendall1990reduction}, and subsequently merging the resulting solutions.  A reason for using Rendall's method has been the ease with which it yields ``one-sided smooth'' solutions given smooth data.  On the other hand, it is to be expected that the same (or similar) existence and regularity results can also be obtained by other means---for instance, by globalising the parametrix-based approach in \cite[Sec.\ 5.4]{friedlander1975wave}, or by using energy methods as done, for instance, in \cite{hormander1990remark,mullerzumhagen1990characteristic,cabet2014characteristic}.  

An alternative route towards obtaining the global existence results in this paper could involve first formulating the problem as a single ``Cauchy problem'' posed on a partially null Cauchy surface.  This is the point of view already sketched in \cite{kay1991theorems} (cf.\ in particular Fig.\ 3 there), and it is also related to the ``SC-Case'' discussed in \cite{mullerzumhagen1977characteristic} (cf.\ Fig.\ 1 there).  In the homogeneous case, one might proceed as follows: given a characteristic datum $f$,  first construct a bona fide (smooth or even just locally Lipschitz) Cauchy surface $\Sigma$ for the ambient manifold, such that $\Sigma \cap \scrN$ is a $\Sigma$-neighbourhood of $\supp{f}$ and contains the locus of points at which $\Sigma$ has a null normal.  Then, solve a global Cauchy problem of sorts, with initial data of mixed type in that only the zeroth-order Cauchy datum is prescribed (and is equal to $f$) on the null portion of $\Sigma$, while, on the remaining spacelike portion of $\Sigma$, the solution is required to vanish together with its first normal derivative.  Notice that, at least if one is content with obtaining solutions in finite energy spaces, the setup in H{\"o}rmander's paper \cite{hormander1990remark} naturally accommodates for such a space of Cauchy data of mixed type: in particular, the space $L^2(\Sigma, d\nu^0_\Sigma)$ in \cite[Eq.\ (8)]{hormander1990remark}, representing first-order Cauchy data, is zero-dimensional when localised in the null portion of $\Sigma$.  A \emph{global} solution will thus exist, certainly in a finite energy space, by the isomorphism in \cite[Eq.\ (8)]{hormander1990remark}.  By the standard domain of dependence property, the solution will have the support properties denoted by (b2) and (b3) in Section \ref{mainres}, and its trace on $\scrN$ will vanish everywhere outside $\supp{f}$.  Hence, by uniqueness as proved in this paper, it will actually be independent of the chosen ``deformation'' $\Sigma$ of $\scrN$.  Here, in addition, the resulting solution is shown to be smooth on each side of $\scrN$.

Aside from these references, there is a vast literature on the classical Goursat problem for characteristic cones \cite{cagnac1981probleme,dossa2002solutions, choquet2011existence,gerard2016construction}.  Finally, in the setting of spacetimes possessing a sufficiently regular \emph{asymptotic null infinity}, the construction of a scattering theory for conformally coupled fields is contingent on the resolution of a conformally related and ``one-sided'' characteristic Cauchy problem posed on null infinity.  This is because data for this problem are interpreted as radiation fields in the sense of Friedlander, i.e.\ as traces on null infinity of conformally rescaled fields.  The reader is referred to \cite{nicolas2016conformal} for an up-to-date account and further references.

\subsection{Bifurcate Killing horizons and spaces of solutions $S_A$ and $S_B$}\label{SASBareOK}

The results in this paper allow to vindicate or sharpen some claims made without proof in old literature on quantum field theory on curved spacetimes.  Arguably, the main example of such literature (which actually provided the initial motivation for the current paper) is the seminal work carried out by Kay and Wald in \cite{kay1991theorems}, on linear Klein--Gordon quantum fields on spacetimes possessing a bifurcate Killing horizon structure.  In the presence of such a structure, two smooth, null, closed and achronal hypersurfaces naturally arise, namely the two ``horizons'' denoted by $\mathcalligra{h}_A$ and $\mathcalligra{h}_B$ in \cite{kay1991theorems} (a detailed presentation can be found in \cite[Sec.\ 2.2]{lupo2015aspects}).  Under the assumptions on the underlying spacetime manifold made in \cite{kay1991theorems}, both $\mathcalligra{h}_A$ and $\mathcalligra{h}_B$ will satisfy all of the geometric hypotheses in the existence and uniqueness theorems proved in this paper.  One may thus define a space $S_A \doteq \Sol^{\mathrm{h}}_{\mathcalligra{h}_{A}} [C_\c^\infty(\mathcalligra{h}_{A})]$, and an analogous space $S_B$ by the replacement $A \leftrightarrow B$.  In the language of \cite{kay1991theorems}, solutions in these spaces ``propagate entirely through'' $\mathcalligra{h}_A$ and $\mathcalligra{h}_B$ (respectively).  It is precisely the existence of these spaces, and the regularity and support properties of their elements, that were claimed without proof in \cite{kay1991theorems} (the reader is referred in particular to the ``Note Added in Proof'' there).  The results in this paper have completely filled this gap.  Therefore, the discussion on this particular point in \cite{kay1991theorems} does not require any modification.

Furthermore, the necessary and sufficient conditions for finite-order differentiability across the initial data hypersurface presented in Section \ref{regularitythm} here make it possible to sharpen some of the arguments made in the ``Note Added in Proof'' to \cite{kay1991theorems}.  Namely, let $\scrN$ be the initial data hypersurface for a two-sided characteristic Cauchy problem and $\scrS$ be a cross-section of $\scrN$.  It follows from the discussion in Section \ref{regularitythm} that
\begin{equation*}
\Sol^{\mathrm{h}}_{\scrN}[C_{\tc}^\infty(\scrN)] \cap C^k(M) = \left\{ \phi \in S_{\scrN} \relmiddle| (t^\ell \phi)^+ = (t^\ell \phi)^- \text{ on $\scrS$, } \forall \ \ell = 0, \ldots, k \right\},
\end{equation*}
where $t$ is any vector field transverse to $\scrN$ and the plus [resp.\ minus] sign indicates a limit as $\scrN$ is approached from $I^+(\scrN)$ [resp.\ from $I^-(\scrN)$].  Applying these considerations to the case of interest in \cite{kay1991theorems}, in which $\scrN = \mathcalligra{h}_A$ or $\mathcalligra{h}_B$ and $\scrS = \Sigma$ is the ``bifurcation surface'', yields the following simple observation: In picking out subspaces of $S_A$ and $S_B$ according to the requirement that they should consist of globally $C^2$ solutions, it was not necessary to seek solutions whose transverse derivatives up to order $5$ vanish on $\Sigma$, as was done in the ``Note Added in Proof'' in \cite{kay1991theorems} and even in the recent revision of that work carried out in \cite[App.\ B]{kay2016non-existence} and in \cite[Ch.\ 4]{lupo2015aspects}.  Instead, solutions whose transverse derivatives up to order $2$ vanish on $\Sigma$ would have already fulfilled this desideratum.

\subsection{Expansion density and two-point functions of Hadamard states on null hypersurfaces}\label{expdensityhadamard}
A linear Klein--Gordon field is one defined by $P = \square_g + [q]$---that is, no first-order terms are present.  If a two-sided characteristic Cauchy problem for such a $P$ is posed on a null hypersurface $\scrN$ with identically vanishing expansion density (Definition \ref{defnexpansiondensity}) then the representation formula, Equation (\ref{repformulafinal}), is reduced to
\begin{equation*}
\Sol^{\mathrm{h}}_{\scrN}(f) = \text{``} 2 \rmG \left(\frac{\partial f}{\partial u} \cdot \delta_{\scrN, n}\right) \text{''},
\end{equation*}
where $u$ is a coordinate on $\scrN$ such that $\partial/\partial u = n^\sharp$.  Both the ``horizons'' $\mathcalligra{h}_A, \mathcalligra{h}_B$ in a bifurcate Killing horizon structure, and future/past null infinity $\scrI^{+/-}$ of an asymptotically flat spacetime in a Bondi gauge \cite[Ch.\ 11]{wald1984general} constitute examples of divergence-free null hypersurfaces.  One might therefore speculate that the simplicity of the representation formula in both cases will be a key ingredient in completely explaining the tantalising similarity between, on the one hand, the universal expressions derived in \cite{kay1991theorems} for the two-point functions of isometry-invariant Hadamard states evaluated on pairs of solutions both in $S_A$ or in $S_B$ and, on the other hand, the expression of the two-point function of the distinguished state on future/past null infinity described in \cite{dappiaggi2017hadamard} and in references therein.  This ought to be further investigated.

\section*{Acknowledgements}
The author is particularly indebted to Alan Rendall for pointing out reference \cite{rendall1992stability}, to Micha{\l} Wrochna for a careful reading of some passages and for bringing attention to the literature on jump formulae and to reference \cite{lerner2017unique}, and to an anonymous referee for useful suggestions and for clarifying some aspects of \cite{hormander1990remark} relevant to the work done here.  It is also a pleasure to thank Matthias Blau, Jonathan Luk, Alexander Strohmaier, Chris Fewster and Bernard Kay for useful discussions.  This work was supported through the Albert Einstein Center for Fundamental Physics, Bern, and the NCCR SwissMAP (The Mathematics of Physics) of the Swiss Science Foundation, during the author's post-doctoral employment at the Institute of Theoretical Physics, University of Bern.  Finally, financial support from the Institut Fourier (Universit\'{e} Grenoble Alpes, CNRS), for a visit during which a preliminary version of this work was presented, is gratefully acknowledged.

\appendix

\section{Causal completeness and causal compactness in Lorentzian geometry}\label{causcomplapp}

Throughout this section, $\scrM = (M,g,\mathfrak{t})$ is a time-oriented Lorentzian manifold.  The following notions were first introduced by Leray in \cite{leray1953hyperbolic} and used extensively in subsequent work, e.g.\ in \cite{friedlander1975wave, baer2007wave}.

\begin{definition}[Future/past compactness] $A \subseteq M$ is \emph{future} [resp.\ \emph{past}] compact if $J^+(p) \cap A$ [resp.\ $J^-(p) \cap A$] is compact for all $p \in M$.  $A$ is \emph{temporally compact} if it is both future and past compact.
\end{definition}

\begin{rk} There also exist the following subtly different notions \cite{galloway1986curvature, treude2011master}: $A \subseteq M$ is said to be \emph{future/past} \emph{causally complete} if, for each $q \in J^{+/-}(A)$, the \emph{closure} of $J^{-/+}(q) \cap A$ in $A$ is compact. For a generic $\scrM$, future (resp.\ past) causal completeness is a weaker condition than past (resp.\ future) compactness.  However, subsets of globally hyperbolic Lorentzian manifolds are future (resp.\ past) causally complete if and only if they are past (resp.\ future) compact.
\end{rk}

Clearly, any closed (in $M$ or in the relative topology, by Lemma \ref{causcompletelem2} below) subset of a future (resp.\ past) compact subsets is itself future (resp.\ past) compact.  Moreover, the union and intersection of two future (resp.\ past) compact subsets is future (resp.\ past) compact.  Compact sets are always temporally compact.  More interestingly, Cauchy surfaces (when they exist) are also temporally compact.  The following further results are proved e.g.\ in \cite[Sec.\ 3.1]{galloway2014achronal} or in \cite[Sec.\ 1.2]{baer2015green}.

\begin{lemma}\label{causcompletelem1}
	If $A \subseteq M$ is either future or past compact, then it is closed. \qed
	\end{lemma}

	\begin{lemma}\label{causcompletelem2}
		Suppose that $\scrM$ is globally hyperbolic.  If $A \subseteq M$ is past compact, then $J^+(A)$ is closed and thus $J^+(A) = \overline{I^+(A)}$.  Furthermore, $J^+(A)$ too is past compact.  Analogous statements hold for $J^-(A)$ if $A$ is future compact. \qed
	\end{lemma}
		
	If $\scrM$ is globally hyperbolic, a Cauchy-surface--based criterion for future/past compactness of closed sets exists.  A subset $A \subseteq M$ is said to be \emph{future} [resp.\ \emph{past}] \emph{bounded} if there exists a Cauchy surface $\mathscr{C}^+$ [resp.\ $\mathscr{C}^-$] such that $A \subseteq J^-(\mathscr{C}^+)$ [resp.\ $A \subseteq J^+(\mathscr{C}^-)$].

		\begin{lemma}\label{causcompletelem3}
			Suppose that $\scrM$ is globally hyperbolic, and that $A \subseteq M$ is closed.  Then $A$ is future (resp.\ past) compact if and only if it is future (resp.\ past) bounded. \qed
			\end{lemma}
			
	For a vector bundle $E \to M$ one may now define subspaces of $\Gamma^k(E)$, with $k \in \N_0 \cup \{\infty\} \cup \{ -\infty\}$, consisting of (distributional) sections with causally restricted supports.  Namely:
\begin{align*}
\Gamma^k_{\mathrm{sc}}(E) &\doteq \left\{ u \in \Gamma^k(E) \relmiddle| \supp{u} \subseteq J(K) \ \text{for a compact } K \subseteq M \right\}, \\
\Gamma^k_{\mathrm{fc}}(E) &\doteq \left\{ u \in \Gamma^k(E) \relmiddle| \supp{u} \text{ is future compact}  \right\}, \\
\Gamma^k_{\mathrm{pc}}(E) &\doteq \left\{ u \in \Gamma^k(E) \relmiddle| \supp{u} \text{ is past compact}  \right\}, \\
\text{and} \quad \Gamma^k_{\mathrm{tc}}(E) &\doteq \left\{ u \in \Gamma^k(E) \relmiddle| \supp{u} \text{ is temporally compact}  \right\} = \Gamma^k_{\mathrm{fc}}(E) \cap \Gamma^k_{\mathrm{pc}}(E)
\end{align*}
are, respectively, the space of $C^k$ sections with \emph{spatially compact support}, the space of $C^k$ sections with future compact support, the space of $C^k$ sections with past compact support, and the space of $C^k$ sections with temporally compact support.  Furthermore, if $E \to M$ is the trivial line bundle, obvious notations will be used for locally Sobolev sections with causally restricted supports:
\begin{equation}\label{sobolevrestrsupp}
\forall \ s \in \R, \ H^s_{\sc/\pc/\fc/\tc}(M) \doteq H^s_\loc(M) \cap \scrD'_{\sc/\pc/\fc/\tc}(M).
\end{equation}

\section{Null hypersurfaces in Lorentzian manifolds}\label{nullhypappx}

\subsection{Elementary properties}\label{nullhypappxelem}

This appendix will follow the presentation in \cite{kupeli1987null}.  Let $(M,g)$ be a Lorentzian manifold of dimension $d+1$ and let $\scrN \subset M$ be a smooth null hypersurface.  Any non-zero vector tangent to $\scrN$ is either spacelike or null, and for any $p \in \scrN$ the set of null or zero vectors in $T_p \scrN$ is a one-dimensional vector subspace which coincides, as a subspace of $T_p M$, with the $g$-orthogonal space ${T_p \scrN}^\perp$ \cite[Lem.\ 5.28]{oneill1983semi-riemannian}.  There is therefore a canonical line bundle on $\scrN$---\ the \emph{null line bundle} $K_{\scrN}$ of $\scrN$---defined by
\begin{equation*}
	K_{\scrN} \doteq \coprod_{p \in \scrN} {T_p \scrN}^\perp \xrightarrow{\pi} \scrN.
\end{equation*}
The following result is standard (see, e.g., \cite[Prop.\ 4]{kupeli1987null}).

\begin{proposition}\label{nullhypglobalnullfield}
	Let $(M,g, \mathfrak{t})$ be a time-oriented Lorentzian manifold, and let $\scrN \subset M$ be a smooth null hypersurface.  Then there exists a global future-directed section of $K_{\scrN} \xrightarrow{\pi} \scrN$.  Any two such sections $n,\, n'$ are related by $n' = f n$ where $f$ is a smooth positive function on $\scrN$.  In particular, $K_{\scrN}$---equivalently, the conormal bundle to $\scrN$---is globally trivial and orientable as a vector bundle.
\end{proposition}

\begin{proof}
Let $\Theta$ be a smooth, future-directed and timelike vector field on $M$.  View the metric as a map $g : TM \times_M TM \to \R$ (where $\times_M$ denotes fibre product), and define $\theta : K_{\scrN} \to \R$ by $\theta(X) = g(X, \Theta \circ \pi(X))$.  Then $\theta$ is smooth, and $\theta(X) = 0 \Leftrightarrow X \in \bm{0}$ where $\bm{0}$ denotes the image of the zero section of $K_{\scrN}$.  $\theta$ has no critical points on $K_{\scrN} {\setminus } \bm{0}$.  Therefore, for any $\alpha \in \R {\setminus } 0$ the preimage $\theta^{-1}\{ \alpha \}$ is a smooth hypersurface in $K_{\scrN}$, and $\pi_\alpha \doteq \pi \restriction_{\theta^{-1}\{ \alpha \}}$ defines an injective smooth map onto $\scrN$.  By e.g.\ expressing $\theta$ in a local trivialization of $K_{\scrN}$, it is also easy to see that $\pi_\alpha$ is an immersion, and therefore a diffeomorphism.  Its inverse defines a global, smooth, nowhere-vanishing section of $K_{\scrN}$, which is future-directed if and only if $\alpha > 0$.
\end{proof}

\begin{definition}\label{nullgendefn}
Let $\scrN$ be a smooth null hypersurface in a time-oriented Lorentzian manifold, and let $n$ be a global future-directed section of the null line bundle of $\scrN$.  Any maximally extended integral curve of $n$ is called a \emph{null generator of} $\scrN$.
\end{definition}

\begin{proposition}\label{nullhyppregeo}
	Let $(M,g)$ be a Lorentzian manifold and $\scrN$ be a smooth null hypersurface with (local) tangent null vector field $n$.  Then, on the subset of $\scrN$ on which $n$ is defined, there exists a smooth real-valued function $f$ such that $\nabla_n n = f n$.  That is, the integral curves of $n$ are null pregeodesics.  In particular, the null generators of a smooth null hypersurface in a time-oriented Lorentzian manifold are null pregeodesics and can be reparametrised to null geodesics. \qed
\end{proposition}

Proposition \ref{nullhyppregeo} does \emph{not} imply that any global future-directed null vector field $n$ tangent to an arbitrary smooth null hypersurface can be globally rescaled to yield a \emph{geodesic vector field} $\tilde{n}$, i.e.\ one satisfying $\nabla_{\tilde{n}} \tilde{n} = 0$.  An obvious obstruction arises if one of the null generators is a closed curve, and any of (and therefore all) its null geodesic reparametrisations returns to the same point with a different velocity.  To ensure that no such obstructions occur, it is necessary to inject further causal assumptions.  A number of sufficient conditions were derived in \cite[Sec.\ 4]{kupeli1987null}.

\begin{definition}\label{crosssectdefn}
Let $\scrN$ be a smooth null hypersurface in a time-oriented Lorentzian manifold, and let $\mathscr{S}$ be a smoothly embedded submanifold of $\scrN$ which is spacelike and of codimension $1$ in $\scrN$.  Then $\mathscr{S}$ will be called a \emph{cross-section} of $\scrN$\footnote{Kupeli \cite[Def.\ 16]{kupeli1987null} prefers to say that $\scrN$ is \emph{causally separated by} $\mathscr{S}$.} if there exists a diffeomorphism $\chi : \scrN \to \R \times \mathscr{S}$ such that: $\chi(\mathscr{S}) = \{0\} \times \mathscr{S}$; denoting by $\partial/\partial U$ the vector field on $\R \times \mathscr{S}$ induced from the standard $\d/\d t$ vector field on $\R$ by the identification $T(\R \times \mathscr{S}) \cong T\R \times T\mathscr{S}$, $\chi^* (\partial/\partial U )$ is a null, future-directed, vector field tangent to $\scrN$ (equivalently, for any $x \in \mathscr{S}$ the curve $\chi^{-1} (\cdot, x) : \R \to \scrN$ is a null generator of $\scrN$).
\end{definition}

A submanifold is a cross-section for a smooth null hypersurface in the abstract sense of Definition \ref{crosssectdefn} if and only if any null generator ``registers'' on it once and never returns to it.

\begin{lemma}[{\cite[Lem.\ 17]{kupeli1987null}}]\label{lemintersectonce}
	Let $\scrN$ be a smooth null hypersurface in a time-oriented Lorentzian manifold $(M,g,\mathfrak{t})$, and let $n$ be a null and future-directed vector field tangent to $\scrN$. Then a spacelike, codimension-$1$ smoothly embedded submanifold $\mathscr{S}$ of $\scrN$ is a cross-section of $\scrN$ according to Definition \ref{crosssectdefn} if and only if any maximal integral curve of $n$ intersects $\mathscr{S}$ at precisely one parameter value. \qed
\end{lemma}

Lemma \ref{lemintersectonce} allows to find a sufficient condition for the rescalability of a null vector field, tangent to a smooth null hypersurface, to a geodesic one.

\begin{proposition}[{\cite[Thm.\ 18]{kupeli1987null}}]\label{geodesicrescaling}
	Let $\scrN$ be a smooth null hypersurface in a time-oriented Lorentzian manifold $\scrM$ and $n$ be a global future-directed null vector field on $\scrN$.  If $\scrN$ admits a cross-section $\mathscr{S}$ then $n$ can be globally rescaled to yield a future-directed null vector field $\tilde{n}$ on $\scrN$ satisfying $\nabla_{\tilde{n}} \tilde{n} = 0$. \qed
\end{proposition}

If $\scrN$ admits a geodesic, future-directed, null, global tangent vector field $n$ then the maximally extended (in $\scrN$) integral curves of $n$ will be referred to as the \emph{null geodesic generators} of $\scrN$.  Viewed as geodesics in the ambient Lorentzian manifold, these may of course fail to be future or past inextensible.

The following proposition is a simple consequence of Lemma \ref{lemintersectonce} and of the equivalence between global hyperbolicity and the existence of smooth spacelike Cauchy surfaces.

\begin{proposition}\label{crosssecglobhyp}
Let $\scrM$ be a globally hyperbolic Lorentzian manifold.  Then any smooth null hypersurface $\scrN$ whose null generators, when reparametrised as null geodesics entirely contained in $\scrN$, are future and past inextensible as geodesics in $\scrM$, admits a cross-section. \qed
\end{proposition}

\subsection{A further result}\label{gensetupCIVP}

Recall that the equalities $\overline{I^\pm(V)} = \overline{J^\pm(V)}$, $\del{I}^\pm(V) = \del{J}^\pm(V)$ and $J^\pm(V)^{\circ} = I^\pm(V)$ hold for any subset $V$ of a time-oriented Lorentzian manifold.  

\begin{proposition}\label{intersectJN}
	Let $\scrM = (M,g,\mathfrak{t})$ be a time-oriented Lorentzian manifold, $\scrN$ be an achronal smooth null hypersurface, and $S$ be a subset of $\scrN$.  Define\footnote{Recall that, according to Definition \ref{nullgendefn}, a null (geodesic) generator of a null hypersurface is always future-directed and maximally extended.}
\begin{equation*}S^+ \doteq S \cup \left\{ p \in \scrN \relmiddle| \exists \ q \in S \text{ s.t.\ $p$ comes after $q$ along the null generator through $q$} \right\}\end{equation*}
[resp.\ define $S^-$ by replacing ``after'' with ``before''].
	\begin{enumerate}[label=\textit{\textbf{(\alph*)}}]
		\item It holds that $S^\pm \subseteq \scrN \cap \del{I}^\pm(S) = \scrN \cap \del{J}^\pm(S)$.  In particular, $\scrN \subseteq \del{J}^\pm(\scrN)$. \label{SpmNinc}
	\end{enumerate}
Suppose in addition that the null generators of $\scrN$, when reparametrised as null geodesics entirely contained in $\scrN$, are future [resp.\ past] inextensible as geodesics in $\scrM$.
\begin{enumerate}[label=\textit{\textbf{(\alph*)}},resume]
	\item $J^+(\scrN) {\setminus } I^+(\scrN) = \scrN$ [resp.\ $J^-(\scrN) {\setminus } I^-(\scrN) = \scrN$]. \label{JminusI}
	\item In either case (``future'' or ``past'') $J^+(\scrN) \cap J^-(\scrN) = \scrN$ and therefore $J^+(A) \cap J^-(A) \subseteq \scrN$ for any $A \subseteq \scrN$. \label{J+capJ-}
	\item If both conditions hold, then $J(\scrN)$ is open in $M$ and $\del{J}(\scrN) = [\del{J}^+(\scrN) {\setminus } \scrN] \cup [\del{J}^-(\scrN) {\setminus } \scrN]$. \label{JN}
\end{enumerate}	
Finally assume that, in addition to the above, $\scrM$ is globally hyperbolic and $S$ is future [resp.\ past] causally complete in $\scrM$.
\begin{enumerate}[\textit{\textbf{(\alph*)}},resume]
	\item $S^+ = \scrN \cap \del{I}^+(S) = \scrN \cap \del{J}^+(S) = \scrN \cap J^+(S)$ [resp.\ $S^- = \scrN \cap \del{I}^-(S) = \scrN \cap \del{J}^-(S) = \scrN \cap J^-(S)$]. \label{SpmNeq}
	\item If $\scrN$ is closed in $M$ then $S^+$ is closed in $M$ [resp.\ $S^-$ is closed in $M$]. \label{Sclosed}
\end{enumerate}
\end{proposition}
\begin{proof}
	The arguments for \ref{SpmNinc} will be given in the case of $S^+$ and $\del{I}^+(S)$, since the statements involving the corresponding objects with $+$ replaced by $-$ then follow simply by a change in time orientation.  Similarly, \ref{JminusI}, \ref{J+capJ-}, \ref{SpmNeq} and \ref{Sclosed} will be proved in the case where the assumption on $\scrN$ holds with the word ``future''. \\
The generic inclusion $S^+ \subseteq \scrN \cap \del{I}^+(S)$ in \ref{SpmNinc} follows from the fact that, on the one hand, $S^+ \subseteq J^+(S) \subseteq \overline{I^+(S)}$ by construction and, on the other hand, $S^+ \cap I^+(S) = \emptyset$ because $\scrN$ is achronal and $S \subseteq S^+ \subseteq \scrN$.  Hence, $S^+ \subseteq \scrN \cap [\overline{I^+(S)} {\setminus } I^+(S)] = \scrN \cap \del{I}^+(S)$.  \\
Since $\scrN$ is achronal, the inclusion $\scrN \subseteq J^+(\scrN) {\setminus } I^+(\scrN)$ in \ref{JminusI} is obvious and it only needs to be shown that, under the assumption on the null generators of $\scrN$, $J^+(\scrN) {\setminus } [I^+(\scrN) \cup \scrN] = \emptyset$.  But this follows from a standard result in Lorentzian geometry \cite[Cor.\ 14.5]{oneill1983semi-riemannian}: namely, if $q \in J^+(\scrN) {\setminus } [I^+(\scrN) \cup \scrN]$ and $\gamma$ is a future-directed causal curve connecting a point $p \in \scrN$ to $q$, then $\gamma$ must be, up to reparametrisation, a null geodesic.  Let $X \in T_p\scrN$ be null and future-directed, and denote by $\nu : I \to M$ the unique null geodesic, maximally extended in $\scrM$, such that $\nu(0)=p$ and $\dot{\nu}(0)=X$.  By assumption, $\nu(I \cap [0, +\infty)) \subseteq \scrN$.  Concatenating $\nu\restriction_{[-\varepsilon,0]}$ (for a sufficiently small $\varepsilon > 0$) with $\gamma$ defines a causal curve $\lambda$ connecting points in $\scrN$ to $q$.  Again using \cite[Cor.\ 14.5]{oneill1983semi-riemannian} since $p \notin I^+(\scrN) \cup \scrN$, $\lambda$ must admit a reparametrisation into a smooth null geodesic.  But then (by geodesic uniqueness), the image of $\lambda$ must be a subset of $\nu(I \cap [-\varepsilon, +\infty))$, and hence it must be entirely contained in $\scrN$.  In particular, one would conclude that $q \in \scrN$, contradicting the initial assumption.  \\
To prove \ref{J+capJ-}, note first that the equalities $I^+(V) \cap J^-(V) = I^-(V) \cap J^+(V) = \emptyset$ hold for any achronal subset $V$ of $M$ (this is yet again an application of \cite[Cor.\ 14.5]{oneill1983semi-riemannian}).  Hence, under the additional assumption on the null generators of $\scrN$, using \ref{JminusI} one sees that
	\begin{equation*} J^+(\scrN) \cap J^-(\scrN) = \left\{I^+(\scrN) \cap J^-(\scrN)\right\} \cup \left\{[J^+(\scrN) {\setminus } I^+(\scrN)] \cap J^-(\scrN)\right\}= \scrN.\end{equation*}
The first statement in \ref{JN} follows from \ref{JminusI} together with the fact that, for any smooth hypersurface $N \subset M$ which is nowhere timelike, the set $I^+(N) \cup N \cup I^-(N)$ is always open.\footnote{Proof: Let $\calU$ be an open neighbourhood of $N$ in $M$ arising from the flow of a global and timelike vector field on $M$.  Then $N \subset \calU \subseteq N \cup I^+(\scrN) \cup I^-(N)$, and the claim follows.}  Clearly, no point in $\scrN$ can belong to $\del{J}(\scrN)$ and thus the inclusion $\del{J}(\scrN) \subseteq [\del{J}^+(\scrN) {\setminus } \scrN] \cup [\del{J}^-(\scrN) {\setminus } \scrN]$ in the second statement in \ref{JN} is certainly satisfied.  It therefore needs to be shown that $\del{J}^+(\scrN) {\setminus } \scrN \subseteq \del{J}(\scrN)$---the arguments for $\del{J}^-(\scrN) {\setminus } \scrN$ being completely analogous.  Since $\del{J}^+(\scrN) \subseteq \overline{J^+(\scrN)} \subseteq \overline{J(\scrN)}$ automatically, it suffices to argue that $\del{J}^+(\scrN) {\setminus } \scrN \subseteq \overline{M {\setminus } J(\scrN)} = M {\setminus } J(\scrN)$.  Using (in the order indicated) item \ref{JminusI}, the fact that $\del{J}^+(\scrN) = \del{I}^+(\scrN)$, the fact that $I^+(\scrN)$ is open, and the fact that $\scrN$ is achronal, one sees that indeed
\begin{equation*}
[\del{J}^+(\scrN) {\setminus } \scrN] \cap J(\scrN) \subseteq \del{I}^+(\scrN) \cap [I^+(\scrN) \cup I^-(\scrN)] \subseteq \del{I}^+(\scrN) \cap I^-(\scrN) = \emptyset.
\end{equation*}
One inclusion in \ref{SpmNeq} was already proved in \ref{SpmNinc} under general assumptions, so one need only prove that $\scrN \cap \del{I}^+(S) \subseteq S^+$ in the presence of global hyperbolicity and of future causal completeness of $S$.  By Lemma \ref{causcompletelem2}, in this case $J^+(S)$ is equal to the closure of $I^+(S)$, whence $\del{I}^+(S) = \overline{I^+(S)} {\setminus } I^+(S) = J^+(S) {\setminus } I^+(S)$.  It follows that the points in $\scrN \cap \del{I}^+(S)$ are precisely those points on $\scrN$ which cannot be reached from $S$ by following a future-directed timelike curve, but can be reached from $S$ by following a future-directed causal curve.  That any such causal curve must be a null geodesic entirely contained in $\scrN$, and thus that such points must belong to $S^+$, then follows from an argument analogous to the one used in the above proof of part \ref{JminusI}; the details will be omitted.
  Finally, $\scrN \cap \del{J}^+(S) = \scrN \cap J^+(S)$ in this case follows from the closedness of $J^+(S)$ together with the achronality of $\scrN$. \\ \ref{Sclosed} then follows immediately since $\del{I}^+(S)$ is a closed set.
\end{proof}

\section{Global Borel Lemma for smooth hypersurfaces}\label{borelslemmaapp}

For completeness, a version of Borel's lemma will be provided here which applies to smooth hypersurfaces in smooth (finite-dimensional) manifolds.

\begin{theorem}[Borel's lemma, smooth hypersurface version]\label{globalborel}
Let $M$ be a manifold of dimension $d+1$, $S \subset M$ be a smooth hypersurface, and $V$ be a smooth vector field which is everywhere transverse to $S$.  Then, given any collection $\{f_n\}_{n=0}^\infty \subset C^{\infty}(S)$ such that $\bigcup_{n=0}^\infty \supp{f_n} \subseteq K$ for some $K$ closed in $M$, there exists an $F \in C^\infty(M)$ solving $V^n F \restriction_{S} = f_n$ for each $n \geq 0$.  
\end{theorem}
\begin{proof}
Let $\calD \subseteq \R \times M$ be the maximal flow domain of $V$, $\theta : \calD \to M$ be the flow of $V$, $\calO \doteq (\R \times S) \cap \calD$, and $\Phi \doteq \theta\restriction_{\calO}$.  By standard results \cite[Thm.\ 9.20]{lee2013introduction}, there exists a smooth positive function $\delta$ on $S$ such that the restriction of $\Phi$ to $\calO_\delta \doteq \left\{ (t,p) \in \calO \relmiddle| |t| < \delta(p) \right\}$ is a diffeomorphism onto an open subset $\calV$ of $M$.  Moreover, denoting this restriction by $\Phi_\delta$, one has $V = [\Phi_\delta]_* (\partial/\partial t)$.  It may then also be deduced (e.g.\ from \cite[Cor.\ 8.21]{lee2013introduction}) that, for all integers $n\geq 0$ and all $F \in C^\infty(\calV)$, $V^n F \circ {\Phi_\delta} = \frac{\partial^n}{\partial t^n}(F \circ \Phi_\delta)$.  Thus, the extension problem at hand is reduced to the problem of finding a function $\tilde{F} \in C^\infty(\calO_\delta)$ such that $F \doteq \tilde{F} \circ {\Phi_\delta}^{-1}$ is smoothly extendible from $\calV$ to $M$ and $\frac{\partial^n \tilde{F}}{\partial t^n}(0,p) = f_n(p) \ \forall \ p \in S$.

Without loss of generality, assume that $\delta(p) \leq 1$ for all $p \in S$.  Let now $\sigma: \R \to \R$ be a smooth function such that
\begin{equation*}
\sigma(t) = \begin{cases} 1 & |t| \leq 1/4 \\ 0 & |t| \geq 1/2.\end{cases}
\end{equation*}
Let $\{ U^{(j)} \}_{j \in J}$ be an open cover of $S$ by $S$-open sets with compact closure, and let $\{ \rho^{(j)} \}_{j \in J}$ be a (locally finite) partition of unity subordinate to this cover.  For all $j \in J$ let also $\delta^{(j)} > 0$ be the infimum of $\delta$ restricted to the set $U^{(j)}$, so that $(-\delta^{(j)}, \delta^{(j)}) \times U^{(j)}$ is an open submanifold of $\calO_\delta$ and its image under $\Phi_\delta$ is an open subset of $M$.  Finally, for all $j$ and $n$ define $f^{(j)}_n \doteq \rho^{(j)} f_n \in C_{\c}^\infty(U^{(j)})$ and let $\mu^{(j)}_n$ be in $[1,+\infty)$.  Then the functions $\{\tilde{g}^{(j)}_n\}_{n=0}^\infty$ defined by
\begin{equation*} \tilde{g}^{(j)}_n(t,p) = \sigma\left(\frac{\mu^{(j)}_n t}{\delta^{(j)}}\right) \frac{t^n}{n!}f^{(j)}_n(p) \end{equation*}
are certainly smooth and with compact support on $(-\delta^{(j)}, \delta^{(j)}) \times U^{(j)}$, for all $j$ and $n$.  Let $h$ be an auxiliary Riemannian metric on $U^{(j)}$, and let $e \doteq h + \d t^2$ on $(-\delta^{(j)}, \delta^{(j)}) \times U^{(j)}$.  For any non-negative integer $k$, let $\norm{\cdot}_{k}$ and $\nnorm{\cdot}_{k}$ (respectively) be the $C^k$-norms of smooth functions on $U^{(j)}$ and $(-\delta^{(j)}, \delta^{(j)}) \times U^{(j)}$ obtained by using the Levi-Civita connections of $h$ and $e$ (respectively).  For each $j$, one can ensure convergence of the series $\sum_{n=0}^\infty \tilde{g}^{(j)}_n $ to a smooth function $\tilde{F}^{(j)}$ on $(-\delta^{(j)}, \delta^{(j)}) \times U^{(j)}$, as well as equality of all derivatives of the limit function with the infinite series of corresponding termwise derivatives, by choosing the sequence $\big(\mu^{(j)}_n\big)_{n=0}^\infty$ in such a way as to make the series absolutely convergent in all $\nnorm{\cdot}_k$ norms.  By combining standard estimates (see, e.g., Lemmas 1.1.11, 1.1.12 and 2.4.1 in \cite{baer2007wave}), it can be shown that one may set $\mu_0^{(j)} = 1$ and
\begin{equation*}
\mu_n^{(j)} = \mu_n^{(j)}\left[\big( f^{(j)}_i \big)_{i=0}^\infty\right] \doteq 1+ \frac{\max_{\ell \leq n-1} \left\{c(\ell,n) \norm{\sigma}_{C^\ell(\R)}\right\} \cdot \sum_{k \leq n-1}2^k\beta_k \big\|f^{(j)}_n\big\|_k}{n!\alpha_n} \quad \forall \ n \geq 1,
\end{equation*}
where the $\beta_k$ and the $c(\ell,n)$ are non-negative \emph{universal} constants,\footnote{In terms of the notation used in \cite[Lem.\ 1.1.12]{baer2007wave}, $\beta_k = \max\{1,\alpha(k,1)\}$.} and $(\alpha_n)_{n=0}^\infty$ is an arbitrary positive and summable sequence.  It is clear that the resulting smooth function $\tilde{F}^{(j)}$ has compact support in $(-\delta^{(j)}, \delta^{(j)}) \times U^{(j)}$, and that $\frac{\partial^n \tilde{F}^{(j)}}{\partial t^n}(0,p) = f_n^{(j)}(p) \ \forall \ p \in U^{(j)}$.  Hence, $\tilde{F} \doteq \sum_{j \in J} \tilde{F}^{(j)}$ is smooth on $\calO_\delta$ and $\frac{\partial^n \tilde{F}}{\partial t^n}(0,p) = f_n(p) \ \forall \ p \in S$.  Let $K \subseteq S$ be the $M$-closed set in the statement of theorem.  By construction, the support of $F = \tilde{F} \circ {\Phi_\delta}^{-1}$ is contained in $C \doteq {\Phi_\delta}\big(\left\{ (t,p) \in \calO \relmiddle| p \in K, \, |t| \leq 3\delta(p)/4 \right\}\big)$ and this set is closed in $M$.  Hence, $F$ can be smoothly extended from $\calV$ to $M$ by combining it with the zero function on $M {\setminus } C$ via a partition of unity subordinate to the open cover $\{ \calV, M {\setminus } C \}$.
\end{proof}

\section{Divergence theorem for densities and formal adjoints}\label{divthmappendix}

Let $M$ be a $(d+1)$-dimensional manifold with or without boundary, on which are defined a $C^1$ nowhere vanishing density $\mu$ and a $C^1$ vector field $\Upsilon$.  The \emph{divergence of $\Upsilon$ relative to $\mu$} is the continuous function ${\div}_\mu \Upsilon$ uniquely defined by
\begin{equation}\label{defnofdivden} ({\div}_\mu \Upsilon) \mu = \Lie_\Upsilon \mu, \end{equation}
where $\Lie_\Upsilon$ denotes the Lie derivative along $\Upsilon$.  In the special case in which $(M,g)$ is a semi-Riemannian manifold, $\nabla$ is the Levi-Civita connection associated with $g$, and $\mu = \mu_g$ is the volume density arising from $g$, it is well known that ${\div}_{\mu_g}{\Upsilon} = \mathrm{Tr}{(\nabla \Upsilon)} = (\nabla_a \Upsilon^a) \, \mu_g$.

The first version of the divergence theorem presented here applies to domains whose topological boundaries are smooth hypersurfaces in $M$.  Following \cite[Secs.\ 10.5 \& 10.6]{loomis1990advanced}, a Borel measurable set $D \subseteq M$ will be called a \emph{domain with regular boundary in $M$} if, for every $p \in M$, there is a chart $(U, \varphi)$ about $p$ such that one of the following three possibilities holds: (i) $U \cap D = \emptyset$; (ii) $U \subseteq D$; (iii) $\varphi(U \cap D) = \varphi(U) \cap \left\{ (y^0, \ldots, y^{d}) \in \R^{d+1} \relmiddle| y^{d} > 0 \right\}$.  Given such a $D$ and $x \in \del D$, the two connected components of $T_x M {\setminus } T_x \del D$ consist of the inward-pointing or outward-pointing vectors with respect to $D$. If $Y$ is a vector at $x$, $\mu$ is a density on $T_x \del D$, and $i$ denotes the inclusion of $\del D$ into $M$, the interior product $\iota_Y \mu$ is the density on $T_x \del D$ defined by $[\iota_Y \mu] (X_1, \ldots, X_d) = \mu(i_* X_1, \ldots,  i_* X_d, Y) \  \forall \ X_1, \ldots, X_d \in T_x \del D$.  It is important to note that $\iota_Y \mu = 0$ whenever $Y \in T_x \del D$, and that $\iota_Y \mu = \iota_{Y'} \mu$ whenever $Y-Y' \in T_x \del D$. [To prove the second statement assuming $Y \notin T_x \del D$, pick an ordered basis $(X_1, \ldots, X_d)$ of $T_x \del D$ and complete it to bases $\calB \doteq (X_1, \ldots, X_d, Y)$ and $\calB' \doteq (X_1, \ldots, X_d, Y')$ of $T_x M$.  Then use the defining property of densities and the fact that the unique linear map $A \in \mathrm{GL}(T_x M)$ taking $\calB'$ to $\calB$ has unit determinant.]  It follows that any $n \in T_x^*M$ normal to $\del D$ uniquely defines a density $\iota_n \mu$ on $T_x \del D$ by
\begin{equation}\label{intproductconormal} \iota_n \mu \doteq \iota_Y \mu \quad \text{for any }Y \in T_x M \text{ such that } n(Y)=1.  \end{equation}
Dually to the notion for vectors, a covector $n$ normal to $\del D$ will be said to be outward pointing (relative to $D$) if $n(Y) > 0$ for any outward-pointing vector $Y$, and inward pointing if $n(Y) < 0$ for any such $Y$.

\begin{theorem}[Classical divergence theorem for densities, following \cite{loomis1990advanced}]\label{divthmdensities}
	Let $M$ be a manifold (without boundary), $\mu$ be a $C^1$ density on $M$, and $\Upsilon$ be a $C^1$ vector field with compact support on $M$.  If $D$ is a domain with regular boundary in $M$ and $n$ is any (not necessarily continuous) field of conormals to $\del{D}$, define the functions $\epsilon_\Upsilon$ and $\sgn_D(n)$ on $\del{D}$ by
	\begin{align*}
		\epsilon_\Upsilon(x) &= \begin{cases} 1 & \text{if $\Upsilon(x)$ points out of $D$}\\ -1 & \text{if $\Upsilon(x)$ points into $D$} \\ 0 &\text{if $\Upsilon(x)$ is tangent to $\del{D}$} \end{cases} \\
		\text{and} \quad \sgn_D(n)(x) &= \begin{cases} 1 & \text{if $n(x)$ is outward pointing}\\ -1 & \text{if $n(x)$ is inward pointing.} \end{cases}
	\end{align*}
Let also $\iota_n \mu$ be as defined in Equation (\ref{intproductconormal}).  Then
\begin{equation}\label{divthmdensitiesgeneraleq}
	\left[\int_{D} ({\div}_\mu \Upsilon) \mu = \right] \int_{D} \Lie_\Upsilon \mu = \int_{\del{D}} \epsilon_\Upsilon \, \iota_\Upsilon \mu = \int_{\del{D}} \sgn_D(n) \, n(\Upsilon) \, \iota_n \mu, \end{equation}
	where the equality in square brackets is valid if $\mu$ is nowhere vanishing. \end{theorem}
\begin{proof}
	The first (unbracketed) equality in (\ref{divthmdensitiesgeneraleq})  is Theorem 6.1, p.\ 421 in \cite{loomis1990advanced}.  It will now be shown that $\epsilon_\Upsilon \, \iota_\Upsilon \mu = \sgn_D(n) \, n(\Upsilon) \, \iota_n \mu$ pointwise under the given assumption on $n$.  Let $x \in \del{D}$ and $Y \in T_x M$ be a vector pointing out of $D$ and such that $n(Y) = \sgn_D(n)(x)$.  Let $\Upsilon_Y$ be uniquely defined by
	\begin{equation*} \Upsilon(x) - \Upsilon_Y Y \in T_x\del{D};\end{equation*}
	then $0 = n(\Upsilon(x)) - \Upsilon_Y n(Y) = n(\Upsilon(x)) - \Upsilon_Y \sgn_D(n)(x)$, i.e.\ $\Upsilon_Y = [\sgn_D(n) \, n(\Upsilon)](x)$.  On the other hand, it obviously holds that $\epsilon_\Upsilon(x) = \sgn(\Upsilon_Y)$.  Hence, everywhere on $\del{D}$,
\begin{equation}\label{epsilonupsilon}
	\epsilon_\Upsilon = \sgn [\sgn_D(n) n(\Upsilon)] = \sgn_D(n) \, \sgn[n(\Upsilon)].
\end{equation}
The pointwise behaviour of $\iota_\Upsilon \mu$ will now be examined. By the considerations preceding the statement of the theorem, $\iota_{\Upsilon(x)} \mu = \iota_{(\Upsilon_Y Y)} \mu$, and the latter equals $|\Upsilon_Y| \, \iota_Y \mu$ by the defining property of densities.  Therefore,
\begin{equation}\label{iotaupsilon}
	\iota_{\Upsilon(x)} \mu = |\Upsilon_Y| \, \iota_Y \mu = |n(\Upsilon(x))| \, \iota_{n(x)} \mu.
\end{equation}
Since $x \in \del{D}$ was arbitrary, combining Equations (\ref{epsilonupsilon}) and (\ref{iotaupsilon}) gives $\epsilon_\Upsilon \, \iota_\Upsilon \mu = \sgn_D(n) \, n(\Upsilon) \, \iota_n \mu$ everywhere on $\del{D}$, as claimed.
\end{proof}

The restriction to domains with regular boundary in the statement of Theorem \ref{divthmdensities} is unnecessary.  Various density/continuity arguments may be used to extend the result to rougher domains and/or to vector fields of regularity lower than $C^1$.  For the purposes of this paper, the following strengthening will suffice: as shown in App.\ I in \cite{taylor2006measure} (see also \cite[Thm.\ 5.16]{evans2015measure}), Theorem \ref{divthmdensities} holds almost \emph{verbatim} if $\del D$ is a locally Lipschitz topological hypersurface in $M$.  In that case, the integrand $\epsilon_\Upsilon \, \iota_\Upsilon \mu$ is not defined everywhere on $\del D$, but is still defined \emph{almost everywhere} on $\del D$ since, by Rademacher's theorem on Lipschitz functions, $\del D$ has a tangent space at almost all of its points.

\subsection{Formal adjoints of vector fields}\label{adjointvecfieldappendix}

As a corollary of Theorem \ref{divthmdensities}, one may now find a general expression for the formal adjoint of a smooth vector field (seen as a first-order scalar differential operator) with respect to a given nowhere-vanishing smooth density.  The proof is a standard exercise in integration by parts and will be omitted.
\begin{corollary}\label{coradjointvecfield}
	Let $\mu$ be a nowhere-vanishing smooth density on a manifold $M$ (with or without boundary), and let $X$ be a $C^1$ vector field.  Then the formal adjoint differential operator to $X$, with respect to $\mu$, is
	\begin{equation*} X^\dagger = -X - [{\div}_\mu X]. \end{equation*}
In particular, the formal adjoints of $X$ relative to two distinct (nowhere-vanishing) densities differ by a multiplication operator.  If $\mu = \mu_g$ is the volume density of a semi-Riemannian metric $g$ then $X^\dagger = -X - {\div}X = -X - [\nabla_a X^a]$. \qed
\end{corollary}

\subsection{Jump formulae for differential operators and single-layer distributions}\label{gencommidentityapp}

The main ingredient in the argument for uniqueness used in the main body of text is a certain commutator identity involving the partial differential operator of interest and the indicator function of a relevant domain.  Analogous identities may be obtained for general operators and domains, as simple corollaries of the divergence theorem for densities.  In the theory of distributions and of partial differential equations, these and related commutator identities (particularly when expressed in distributional language as will be done below) are often referred to as \emph{jump formulae}.

In this subsection, $M$ will denote a manifold on which is defined a preferred smooth, nowhere-vanishing density $\mu$.  The following is a simple consequence of the divergence theorem together with the definition of a Green vector field, Equation (\ref{greenvinogradov}).

\begin{corollary}\label{corcommidentitygen} Let $P$ be a scalar differential operator on $M$ with smooth coefficients, for which a Green vector field $\calG$ exists in the sense of Equation (\ref{greenvinogradov}).  Let $D$ be a domain with locally Lipschitz boundary in $M$; $\bm{1}_D$ will denote the indicator function of $D$, whose corresponding multiplication operator is $[\bm{1}_D] : C^\infty(M) \to \scrD'(M)$.  Then, for any (possibly rough and/or almost everywhere defined) field of outward-pointing conormals to $\del D$,
\begin{equation}
  \pair{[[\bm{1}_D], P]\phi}{\chi \mu} = \int_{D} ({\div}_\mu \calG[\chi, \phi]) \, \mu = \int_{\del{D}} n(\calG[\chi, \phi]) \, \iota_n \mu \label{commidentity1}
\end{equation}
whenever $\phi, \chi \in C^\infty(M)$ and $\supp{\chi} \cap \supp{\phi} \cap \overline{D}$ is compact. \qed
\end{corollary}

Recall that pointwise multiplication by $\mu$ yields an isomorphism $[\mu] : C_{(\c)}^\infty(M) \to \Gamma_{(\c)}^\infty(\calD M)$.

\begin{definition}[Single-layer distributions]\label{singlayerdistr}
	Let $X$ be a smoothly embedded closed submanifold of $M$ and $i : X \to M$ be the inclusion map.  Given a $\rho \in \big[C_{\c}^\infty(X)\big]'$, the assignment
	\begin{equation*}
	\Gamma_{\c}^\infty(\calD M) \ni \nu \mapsto \rmS_{X,\rho}(\nu) \doteq \big\{\rho \circ i^* \circ [\mu]^{-1}\big\}(\nu)
	\end{equation*}
	defines a distribution $\rmS_{X, \rho}$ on $M$ which will be referred to as the \emph{single-layer distribution} associated with the pair $(X, \rho)$. $\rmS_{X, \rho}$ has compact support whenever $\rho$ does.  The linear and continuous mapping $[C_\c^\infty(X)]' \ni \rho \mapsto \rmS_{X, \rho} \in \scrD'(M)$, uniquely determined by $X$ and the background density $\mu$, will be denoted by $\calS_{X,\mu}$.
\end{definition}

Note that, by the divergence theorem, $\pair{[\calS_{X, \mu} \circ \Lie_{\mathfrak{y}}](\rho)}{\mu} = 0$ for any $\rho \in [C_\c^\infty(X)]'$ and $\mathfrak{y} \in \mathfrak{X}(X)$ such that $\supp{\mathfrak{y}} \cap \supp{\rho}$ is compact.  A basic example of a single-layer distribution is what might be referred to as the \emph{delta distribution} supported on a hypersurface $Y$ relative to a (sufficiently regular) field of non-zero conormals along $Y$, namely the distribution $\delta_{Y,n}$ given by $\pair{\delta_{Y,n}}{\nu} \doteq \int_Y \iota_n \nu$.  Indeed, $\delta_{Y,n} = \calS_{Y,\mu}(\iota_n \mu)$ for any smooth nowhere-vanishing density $\mu$.  The following property will be relevant: if $F \in C^\infty(Y; \R {\setminus } \{ 0 \})$ and $n' \doteq Fn$, then $\delta_{Y, n'}(\nu) = \delta_{Y, n}(|F|^{-1} \nu)$.  In particular, $\delta_{Y,n} = \delta_{Y,-n}$.

Single-layer distributions (as well as their generalisations, the \emph{multi-layer} distributions) appear frequently in the literature on PDEs \cite{seeley1962distributions, wagner2010distributions, stampfer2011pullback, ortner2015fundamental}, for they provide a convenient way to reformulate jump formulae in a distributional language.  To wit, the following is a restatement of Corollary \ref{corcommidentitygen}.  

\begin{corollary}\label{cordivthmdelta}
Let $P$ and $\calG$ be as in Corollary \ref{corcommidentitygen}, and let $D$ be a domain with regular boundary in $M$.  Assume there exists a smooth, everywhere outward-pointing conormal field $n$ along $\del D$.  Then, whenever $\phi, \chi \in C^\infty(M)$ and $\supp{\chi} \cap \supp{\phi} \cap \overline{D}$ is compact,
	\begin{equation*}
	\pair{[[\bm{1}_D], P]\phi}{\chi \mu} = \pair{\calS_{\del D, \mu}(M[\chi, \phi])}{\mu},
	\end{equation*}
	where $M[\chi, \phi] \doteq n(\calG[\chi, \phi]\restriction_{\del{D}}) \, \iota_n \mu$. \qed
\end{corollary}

\bibliographystyle{amsalpha}
\bibliography{Bibliography}

\providecommand{\bysame}{\leavevmode\hbox to3em{\hrulefill}\thinspace}
\providecommand{\MR}{\relax\ifhmode\unskip\space\fi MR }
\providecommand{\MRhref}[2]{%
  \href{http://www.ams.org/mathscinet-getitem?mr=#1}{#2}
}
\providecommand{\href}[2]{#2}
\begin{thebibliography}{CBCMG11}

\bibitem[AB02]{alonso-blanco2002green}
R.~J. Alonso-Blanco, \emph{On the {G}reen--{V}inogradov formula}, Acta Appl.
  Math. \textbf{72} (2002), no.~1, 19--32.

\bibitem[ABV04]{alonso-blanco2004green}
R.~J. Alonso-Blanco and A.~M. Vinogradov, \emph{Green formula and {L}egendre
  transformation}, Acta Appl. Math. \textbf{83} (2004), no.~1, 149--166.

\bibitem[Agr15]{agranovich2015sobolev}
M.~S. Agranovich, \emph{{Sobolev Spaces, Their Generalizations and Elliptic
  Problems in Smooth and Lipschitz Domains}}, Springer Monographs in
  Mathematics, Springer, 2015.

\bibitem[Amb80]{ambrose1980products}
W.~Ambrose, \emph{Products of distributions with values in distributions}, J.
  Reine Angew. Math. \textbf{315} (1980), 73--91.

\bibitem[B{\"a}r15]{baer2015green}
C.~B{\"a}r, \emph{Green-hyperbolic operators on globally hyperbolic
  spacetimes}, Commun. Math. Phys. \textbf{333} (2015), no.~3, 1585--1615,
  \href{https://arxiv.org/abs/1310.0738}{arXiv:1310.0738}.

\bibitem[B{\"a}r17]{baer2017lectures}
\bysame, \emph{{Geometric Wave Equations}}, Preprint lecture notes, University of Potsdam, 2017, \url{https://www.math.uni-potsdam.de/fileadmin/user_upload/Prof-Geometrie/Dokumente/Lehre/Lehrmaterialien/Waves.pdf}.

\bibitem[BG12]{baer2012classical}
C.~B{\"a}r and N.~Ginoux, \emph{Classical and quantum fields on {L}orentzian
  manifolds}, Global Differential Geometry (C.\ B{\"a}r, J.\ Lohkamp, and M.\
  Schwarz, eds.), Springer Proceedings in Mathematics, vol.~17, Springer, 2012,
  \href{https://arxiv.org/abs/1104.1158}{arXiv:1104.1158}, pp.~359--400.

\bibitem[BGP07]{baer2007wave}
C.~B{\"a}r, N.~Ginoux, and F.~Pf{\"a}ffle, \emph{{Wave Equations on Lorentzian
  Manifolds and Quantization}}, ESI {L}ectures in {M}athematics and {P}hysics,
  European Mathematical Society, 2007,
  \href{https://arxiv.org/abs/0806.1036}{arXiv:0806.1036}.

\bibitem[BS07]{bernal2007globally}
A.~N. Bernal and M.~S\'anchez, \emph{Globally hyperbolic spacetimes can be
  defined as `causal' instead of `strongly causal'}, Class. Quantum Grav.
  \textbf{24} (2007), no.~3, 745,
  \href{https://arxiv.org/abs/gr-qc/0611138}{gr-qc/0611138}.

\bibitem[BW15]{baer2015initial}
C.~B\"ar and R.~T. Wafo, \emph{Initial value problems for wave equations on
  manifolds}, Math. Phys. Anal. Geom. \textbf{18} (2015), no.~1, {A}rticle 7,
  1--29, \href{https://arxiv.org/abs/1408.4995}{arXiv:1408.4995}.

\bibitem[Cag81]{cagnac1981probleme}
F.~Cagnac, \emph{{Probl{\`e}me de Cauchy sur un cono{\"\i}de
  caract{\'e}ristique pour des {\'e}quations quasi-lin{\'e}aires}}, Ann. Mat.
  Pura Appl. \textbf{129} (1981), no.~1, 13--41.

\bibitem[CBCMG11]{choquet2011existence}
Y.~Choquet-Bruhat, P.~T. Chru\'{s}ciel, and J.~M. Mart{\'i}n-Garc{\'i}a,
  \emph{An existence theorem for the {C}auchy problem on a characteristic cone
  for the {E}instein equations}, Complex Analysis and Dynamical Systems IV:
  Part 2. General Relativity, Geometry, and PDE, Contemporary Mathematics, vol.
  554, 2011, \href{https://arxiv.org/abs/1006.5558}{arXiv:1006.5558},
  pp.~73--81.

\bibitem[CCW14]{cabet2014characteristic}
A.~Cabet, P.~T. Chru\'{s}ciel, and R.~T. Wafo, \emph{{On the characteristic
  initial value problem for nonlinear symmetric hyperbolic systems, including
  {E}instein equations}},
  \href{https://arxiv.org/abs/1406.3009}{arXiv:1406.3009}, 2014.

\bibitem[CL55]{coddington1955theory}
A.~Coddington and N.~Levinson, \emph{{Theory of Ordinary Differential
  Equations}}, International Series in Pure and Applied Mathematics,
  McGraw-Hill, 1955.

\bibitem[CP12]{chrusciel2012many}
P.~T. Chru\'{s}ciel and T.-T. Paetz, \emph{{The many ways of the characteristic
  {C}auchy problem}}, Class.\ Quantum Grav. \textbf{29} (2012), 145006,
  \href{https://arxiv.org/abs/1203.4534}{arXiv:1203.4534}.

\bibitem[DH72]{duistermaat1972fourier}
J.~J. Duistermaat and L.~H{\"o}rmander, \emph{Fourier integral operators.
  {I}{I}}, Acta Math. \textbf{128} (1972), 183--269.

\bibitem[DMP17]{dappiaggi2017hadamard}
C.~Dappiaggi, V.~Moretti, and N.~Pinamonti, \emph{{Hadamard States From
  Light-like Hypersurfaces}}, SpringerBriefs in Mathematical Physics, Springer,
  2017, \href{https://arxiv.org/abs/1706.09666}{arXiv:1706.09666}.

\bibitem[Dos02]{dossa2002solutions}
M.~Dossa, \emph{Solutions {$C^\infty$} d'une classe de probl{\`e}mes de
  {C}auchy quasi-lin{\'e}aires hyperboliques du second ordre sur un
  cono{\"\i}de caract{\'e}ristique}, Ann. Fac. Sci. Toulouse Math., S{\'e}rie 6
  \textbf{11} (2002), no.~3, 351--376.

\bibitem[EG15]{evans2015measure}
L.~C. Evans and R.~F. Gariepy, \emph{{Measure Theory and Fine Properties of
  Functions}}, 2nd ed., Textbooks in Mathematics, CRC Press, 2015.

\bibitem[Fri75]{friedlander1975wave}
F.~G. Friedlander, \emph{{The Wave Equation on a Curved Space-Time}}, Cambridge
  Monographs in Mathematical Physics, Cambridge University Press, 1975.

\bibitem[Gal86]{galloway1986curvature}
G.~J. Galloway, \emph{Curvature, causality and completeness in space-times with
  causally complete spacelike slices}, Math. Proc. Camb. Philos. Soc.
  \textbf{99} (1986), no.~2, 367--375.

\bibitem[Gin09]{ginoux2009linear}
N.~Ginoux, \emph{Linear wave equations}, Quantum {F}ield {T}heory on {C}urved
  {S}pacetimes: {C}oncepts and {M}athematical {F}oundations (C.\ B{\"a}r and
  K.\ Fredenhagen, eds.), Lecture Notes in Physics, vol. 786, Springer, 2009,
  pp.~59--84.

\bibitem[GV14]{galloway2014achronal}
G.~J. Galloway and C.~Vega, \emph{Achronal limits, {L}orentzian spheres, and
  splitting}, Ann. H. Poincar{\'e} \textbf{15} (2014), no.~11, 2241--2279,
  \href{https://arxiv.org/abs/1211.2460}{arXiv:1211.2460}.

\bibitem[GW16]{gerard2016construction}
C.~G{\'e}rard and M.~Wrochna, \emph{Construction of {H}adamard states by
  characteristic {C}auchy problem}, Anal. PDE \textbf{9} (2016), no.~1,
  111--149, \href{https://arxiv.org/abs/1409.6691}{arXiv:1409.6691}.

\bibitem[HE73]{hawking1973large}
S.~W. Hawking and G.~F.~R. Ellis, \emph{{The Large Scale Structure of
  Space-Time}}, Cambridge University Press, 1973.

\bibitem[H{\"o}r90]{hormander1990remark}
L.~H{\"o}rmander, \emph{A remark on the characteristic {C}auchy problem}, J.
  Funct. Anal. \textbf{93} (1990), no.~2, 270--277.

\bibitem[H{\"o}r97]{hormander1997lectures}
\bysame, \emph{{Lectures on Nonlinear Hyperbolic Differential Equations}},
  Math{\'e}matiques et Applications, Springer, 1997.

\bibitem[JS02]{junker2002adiabatic}
W.~Junker and E.~Schrohe, \emph{Adiabatic vacuum states on general spacetime
  manifolds: Definition, construction, and physical properties}, Ann. H.
  Poincar{\'e} \textbf{3} (2002), no.~6, 1113--1181,
  \href{https://arxiv.org/abs/math-ph/0109010}{math-ph/0109010}.

\bibitem[Kha14]{khavkine2014covariant}
I.~Khavkine, \emph{Covariant phase space, constraints, gauge and the {P}eierls
  formula}, Int. J. Mod. Phys. \textbf{29} (2014), no.~5, 1430009,
  \href{https://arxiv.org/abs/1402.1282}{arXiv:1402.1282}.

\bibitem[KL16]{kay2016non-existence}
B.~S. Kay and U.~Lupo, \emph{Non-existence of isometry-invariant {H}adamard
  states for a {K}ruskal black hole in a box and for massless fields on 1+1
  {M}inkowski spacetime with a uniformly accelerating mirror}, Class.\ Quantum
  Grav. \textbf{33} (2016), no.~21, 215001,
  \href{https://arxiv.org/abs/1502.06582}{arXiv:1502.06582}.

\bibitem[Kup87]{kupeli1987null}
D.~N. Kupeli, \emph{On null submanifolds in spacetimes}, Geom. Dedic.
  \textbf{23} (1987), no.~1, 33--51.

\bibitem[KW91]{kay1991theorems}
B.~S. Kay and R.~M. Wald, \emph{Theorems on the uniqueness and thermal
  properties of stationary, nonsingular, quasifree states on spacetimes with a
  bifurcate {K}illing horizon}, Phys. Rep. \textbf{207} (1991), no.~2, 49--136.

\bibitem[Lee13]{lee2013introduction}
J.~M. Lee, \emph{{Introduction to Smooth Manifolds}}, 2nd ed., Graduate Texts
  in Mathematics, Springer, 2013.

\bibitem[Ler53]{leray1953hyperbolic}
J.~Leray, \emph{{Hyperbolic Differential Equations}}, Princeton Lecture Notes,
  Institute for Advanced Study, Princeton University (mimeographed), 1953.

\bibitem[{Ler}17]{lerner2017unique}
N.~{Lerner}, \emph{Unique continuation through transversal characteristic
  hypersurfaces}, \href{https://arxiv.org/abs/1601.07814}{arXiv:1601.07814},
  2017.

\bibitem[LS90]{loomis1990advanced}
L.~H. Loomis and S.~Sternberg, \emph{{Advanced Calculus}}, 2nd ed., Jones and
  Bartlett Publishers, Burlington, 1990.

\bibitem[Lup15]{lupo2015aspects}
U.~Lupo, \emph{Aspects of (quantum) field theory on curved spacetimes,
  particularly in the presence of boundaries}, Ph.D. thesis, University of
  York, 2015.

\bibitem[MzH90]{mullerzumhagen1990characteristic}
H.~M{\"u}ller~zum Hagen, \emph{Characteristic initial value problem for
  hyperbolic systems of second order differential equations}, Ann. Inst. H.
  Poincar{\'e} -- Phys. Th{\'e}or. \textbf{53} (1990), no.~2, 159--216.

\bibitem[MzHS77]{mullerzumhagen1977characteristic}
H.\ M{\"u}ller~zum Hagen and H.-J. Seifert, \emph{On characteristic
  initial-value and mixed problems}, Gen. Rel. Gravit. \textbf{8} (1977),
  no.~4, 259--301.

\bibitem[Nic07]{nicolaescu2007lectures}
L.~I. Nicolaescu, \emph{{Lectures on the Geometry of Manifolds}}, 2nd ed.,
  World Scientific, 2007, \url{https://www3.nd.edu/~lnicolae/Lectures.pdf}.

\bibitem[Nic16]{nicolas2016conformal}
J.-P. Nicolas, \emph{Conformal scattering on the {S}chwarzschild metric}, Ann.
  Inst. Fourier \textbf{66} (2016), no.~3, 1175--1216,
  \href{https://arxiv.org/abs/1312.1386}{arXiv:1312.1386}.

\bibitem[Obe86]{oberguggenberger1986products}
M.~Oberguggenberger, \emph{Products of distributions}, J. Reine Angew. Math.
  \textbf{365} (1986), 1--11.

\bibitem[O'N83]{oneill1983semi-riemannian}
B.~O'Neill, \emph{{Semi-Riemannian Geometry With Applications to Relativity}},
  Pure and Applied Mathematics, Academic Press, 1983.

\bibitem[OW15]{ortner2015fundamental}
N.~Ortner and P.~Wagner, \emph{Fundamental {S}olutions of {L}inear {P}artial
  {D}ifferential {O}perators: {T}heory and {P}ractice}, Springer, 2015.

\bibitem[Pee59]{peetre1959characterisation}
J.~Peetre, \emph{Une caract{\'e}risation abstraite des op{\'e}rateurs
  diff{\'e}rentiels}, Math. Scand. \textbf{7} (1959), no.~1, 211--218,
  \textit{R{\'e}ctification} \textbf{8} (1960), no.\ 1, 116--120.

\bibitem[Ren90]{rendall1990reduction}
A.~D. Rendall, \emph{Reduction of the characteristic initial value problem to
  the {C}auchy problem and its applications to the {E}instein equations}, Proc.
  R. Soc. Lond. A \textbf{427} (1990), no.~1872, 221--239.

\bibitem[Ren92]{rendall1992stability}
\bysame, \emph{Stability in the characteristic initial value problem},
  Relativity {T}oday (Z.~Perj\'{e}s, ed.), Commack, New York: Nova Science,
  1992, pp.~57--64.

\bibitem[Rin09]{ringstrom2009cauchy}
H.~Ringstr{\"o}m, \emph{{The Cauchy Problem in General Relativity}}, ESI
  lectures in mathematics and physics, European Mathematical Society, 2009.

\bibitem[See62]{seeley1962distributions}
R.~T. Seeley, \emph{Distributions on {S}urfaces}, Toegepaste Wiskunde, no. TW
  78, Stichting Mathematisch Centrum Amsterdam, 1962.

\bibitem[Sta11]{stampfer2011pullback}
F.~Stampfer, \emph{Pullback of measures on {R}iemannian manifolds}, J. Math.
  Anal. Appl. \textbf{381} (2011), no.~2, 812--820.

\bibitem[Tay06]{taylor2006measure}
M.~E. Taylor, \emph{{Measure Theory and Integration}}, Graduate Studies in
  Mathematics, American Mathematical Society, Providence, 2006.

\bibitem[Tre11]{treude2011master}
J.-H. Treude, \emph{Ricci curvature comparison in {R}iemannian and {L}orentzian
  geometry}, Diploma thesis, Albert-Ludwigs-Universit\"at Freiburg, 2011.

\bibitem[Wag10]{wagner2010distributions}
P.~Wagner, \emph{Distributions supported by hypersurfaces}, Appl. Anal.
  \textbf{89} (2010), no.~8, 1183--1199.

\bibitem[Wal84]{wald1984general}
R.~M. Wald, \emph{{General Relativity}}, University of Chicago Press, 1984.

\bibitem[Wal94]{wald1994quantum}
\bysame, \emph{{Quantum Field Theory in Curved Spacetime and Black Hole
  Thermodynamics}}, Chicago Lectures in Physics, University of Chicago Press,
  1994.

\bibitem[Whi34]{whitney1934analytic}
H.~Whitney, \emph{Analytic extensions of differentiable functions defined in
  closed sets}, Trans. Am. Math. Soc. \textbf{36} (1934), no.~1, 63--89.

\end{thebibliography}

\end{document}